\documentclass[preprint,12pt]{elsarticle}

\usepackage{amssymb}
\usepackage{amsmath}
\usepackage{amsthm}
\usepackage{natbib}

\newcommand{\uu}{\mathbf{u}}
\newcommand{\vv}{\mathbf{v}}
\newcommand{\q}{\mathbf{q}}
\newcommand{\p}{\mathbf{p}}
\newcommand{\rr}{\mathbf{r}}
\renewcommand{\a}{\mathbf{a}}

\newcommand{\X}{\mathbf{X}}
\newcommand{\Y}{\mathbf{Y}}
\newcommand{\TS}{\boldsymbol{\mathcal{S}}}

\newcommand{\Var}{\text{Var}}
\newcommand{\Cov}{\text{Cov}}

\newcommand{\ParSumVar}{T}

\newcommand{\E}{\text{E}}
\newcommand{\V}{\mathbf{V}}
\newcommand{\Vecc}{\texttt{vec}}

\newcommand{\sgn}{\texttt{sgn}}

\newcommand{\CS}{\textrm{Count-Sketch}}
\newcommand{\HCS}{\textrm{Higher Order Count Sketch}}
\newcommand{\cs}{\textrm{CS}}
\newcommand{\csb}[1]{\mathbf{CS(#1)}}

\newcommand{\hcs}{\textrm{HCS}}
\newcommand{\hcsb}[1]{\mathbf{HCS(#1)}}
\newcommand{\etal}{\textit{et al.}}
\newcommand{\aka}{\textit{a.k.a.}}
\newcommand{\ELSH}{\textrm{E2LSH}}

\newcommand{\CSELSH}{\textrm{CS-E2LSH}}

\newcommand{\HCSELSH}{\textrm{HCS-E2LSH}}

\newcommand{\SRP}{\textrm{SRP}}
\newcommand{\CSSRP}{\textrm{CS-SRP}}
\newcommand{\HCSSRP}{\textrm{HCS-SRP}}

\newcommand{\Tr}{\textrm{Tr}}

\newcommand{\Sig}{\boldsymbol{\Sigma}}
\renewcommand{\S}{\mathbf{S}}
\renewcommand{\Pr}{\text{Pr}}

\usepackage{mathtools}
\usepackage[utf8]{inputenc} 
\usepackage[T1]{fontenc}    
\usepackage{hyperref}       
\usepackage{url}            
\usepackage{booktabs}       
\usepackage{amsfonts}       
\usepackage{nicefrac}       
\usepackage{microtype}      
\usepackage{xcolor}         
\usepackage{graphicx}
\usepackage{adjustbox}
\usepackage{amssymb, bm}
\usepackage{stmaryrd}
\usepackage{bbm}
\usepackage{float}
\usepackage{soul}
\usepackage{enumerate,enumitem}
\allowdisplaybreaks

\newtheorem{thm}{Theorem}

 \newtheorem{definition}[thm]{Definition}
 
  \newtheorem{cor}[thm]{Corollary}
  \newtheorem{remark}[thm]{Remark}
\usepackage{lineno,hyperref}

\newcommand{\R}{\mathbb{R}}
\newcommand{\Z}{\mathbb{Z}}

\newcommand\numberthis{\addtocounter{equation}{1}\tag{\theequation}}

\begin{document}

\begin{frontmatter}



\title{Faster and Space Efficient Indexing for Locality Sensitive Hashing }
\author[aff]{Bhisham Dev Verma\fnref{label2}}
\ead{bhishamdevverma@gmail.com}
\author[aff1]{Rameshwar Pratap}
\ead{rameshwar@cse.iith.ac.in}
\affiliation[aff]{organization={Indian Institute of Technology Mandi},
            state={Himachal Pradesh},
            country={India}}
\affiliation[aff1]{organization={Indian Institute of Technology Hyderabad},
            state={Telangana},
            country={India}}




\begin{abstract}
This work suggests faster and space-efficient index construction algorithms for LSH for  Euclidean distance (\textit{a.k.a.}~\ELSH) and cosine similarity (\textit{a.k.a.}~\SRP).  The index construction step of these LSHs relies on grouping data points into several bins of hash tables based on their hashcode.   To generate an $m$-dimensional hashcode of the $d$-dimensional data point, these  LSHs first project the data point onto a $d$-dimensional random Gaussian vector and then discretise the resulting inner product.  The time and space complexity of both  \ELSH~and \SRP~for computing an $m$-sized hashcode of a $d$-dimensional vector is $O(md)$, which becomes impractical for large values of $m$ and $d$. To overcome this problem, we propose two alternative LSH hashcode generation algorithms both for Euclidean distance and cosine similarity, namely, \CSELSH, \HCSELSH~and \CSSRP, \HCSSRP, respectively.  \CSELSH~and \CSSRP~are based on count sketch \cite{count_sketch} and  \HCSELSH~and \HCSSRP~utilize  higher-order count sketch \cite{shi2019higher}. These proposals significantly reduce the hashcode computation time from $O(md)$ to $O(d)$. Additionally, both \CSELSH~and  \CSSRP~reduce the space complexity from $O(md)$ to $O(d)$; ~and \HCSELSH, \HCSSRP~ reduce the space complexity from  $O(md)$ to $O(N \sqrt[N]{d})$ respectively, where $N\geq 1$ denotes the size of the input/reshaped tensor. 
Our proposals are backed by strong mathematical guarantees, and we validate their performance through simulations on various real-world datasets.
\end{abstract}




\end{frontmatter}

\section{Introduction}
Approximate nearest neighbor (ANN) search is one of the most common subroutines in several real-life applications such as recommendation engines, duplicate detection, clustering, etc.
For a given query point, the ANN problem allows it to return a point within a distance of some constant times the distance of the query point to its actual nearest neighbor. The seminal work of Indyk and Motwani~\cite{indyk1998approximate} suggested an algorithmic framework, namely Locality Sensitive Hashing (LSH),  to solve this problem. The idea of LSH~\cite{indyk1998approximate} is based on 
designing a family of hash functions that group points into several buckets of hash tables such that with high probability, it maps  "\textit{similar}" points in the same bucket and "\textit{not so similar}" points in different buckets.
Consequently, this restricts the search space for a query point to that particular bucket in which it falls, leading to a sublinear search algorithm. We refer to the step of arranging points into hash tables as the indexing step. 
LSH functions have been studied for different similarity/distance measures, \textit{e.g.}, cosine similarity~\cite{charikar2002similarity,ji2012super, yu2014circulant, andoni2015practical,kang2018improving,dubey2022improving}, Euclidean distance~\cite{datar2004locality, DBLP:conf/focs/AndoniI06}, Hamming distance~\cite{indyk1998approximate,gionis1999similarity}, ~Jaccard similarity~\cite{BroderCFM98,b-bit,OPH,DOPH}, and many more. Further, such results have been utilized in a wide range of applications, \textit{e.g.}, near-duplicate detection~\cite{das2007google},  clustering~\cite{koga2007fast,cochez2015twister}, audio processing~\cite{ryynanen2008query,yu2010combining}, image/video processing~\cite{kulis2009kernelized,xia2016privacy}, blockchain~\cite{zhuvikin2018blockchain}, etc.

 The sublinear search time of the LSH algorithm comes at the cost of (preprocessing) time required to construct the indexing/hash tables.  The indexing (\textit{a.k.a.} hash table construction) step requires arranging data points into several bins based on their hashcode. We elaborate it for LSH for Euclidean distance (\textit{a.k.a.} \ELSH)~\cite{datar2004locality,indyk1998approximate}. In the \ELSH~algorithm, a hashcode of a data point is computed by projecting a $d$-dimensional data point on a random vector whose elements are sampled from $\mathcal{N}(0, 1)$. The resultant inner product is discretized by adding a random number (let's say $b$) uniformly sampled from  $[0,w]$, dividing it by $w$, and finally computing the floor (see Definition~\ref{def:e2lsh}).
 The process is independently repeated $m$ times to generate an $m$-sized hashcode of the data point. This can be considered as projecting a $d$-dimensional input vector on a $d\times m$ random projection matrix whose each entry is sampled from a Gaussian distribution, followed by discretizing the components of the resultant vector. Therefore, to generate an $m$ sized hashcode of a datapoint, \ELSH~ requires $O(md)$  running time and space. Similarly, signed-random-projection (\textit{a.k.a.} SRP)~\cite{charikar2002similarity}, which is an LSH for real-valued vectors for cosine similarity, also requires $O(md)$  running time and space to generate an $m$-sized hashcode of the input vector. The idea of SRP is similar to \ELSH,  where the $d$-dimensional input vector is projected on $d$-dimensional random Gaussian vectors, and the hashcode is assigned to $0$ or $1$ based on the sign of the resultant inner product (see Definition~\ref{def:srp}).

%

The time and space required by the~\ELSH ~and SRP in the preprocessing step become prohibitive for high-dimensional datasets. This work addresses this challenge and suggests a simple, intuitive algorithmic approach that significantly improves the indexing step's time and space complexity. 
At a high level, our idea is to use a super sparse projection matrix for projecting the input vector, followed by discretizing the resultant inner product. We use the count-sketch based projection matrix~\cite{charikar2002similarity} for this purpose. This reduces the preprocessing step's space and execution time to $O(d)$ in contrast to the $O(md)$ space and time needed by \ELSH~and SRP to produce an $m$-sized hashcode. We give another result that gives a more space-efficient algorithm. We use the recently proposed higher order count sketch (HCS)~\cite{shi2019higher} for this purpose.  

At the core, our idea is to visualize the input vector as a high-order tensor, compress it using higher order count sketch~\cite{shi2019higher}, discretize the resultant tensor, and finally unfold the tensor to a vector that serves as the hashcode of the input vector. For example, for $N=2$, where $N$ denotes the order of tensor,  we reshape the $d$-dimensional input vector as a matrix of size $\sqrt{d}\times \sqrt{d}$. \footnote{{
For ease of analysis, we consider the same dimension along each mode of the reshaped tensor as well as the compressed (HCS) tensor \textit{i.e.}, $d^{\frac{1}{N}}$ and $m^{\frac{1}{N}}$ respectively. However, our results can easily be extended to different values of mode dimensions}.} We then compress this matrix using a higher order count sketch and obtain an $\sqrt{m}\times \sqrt{m}$ matrix. This compression step can be considered as left multiplying the input matrix with a count-sketch projection matrix of size $\sqrt{m}\times \sqrt{d}$, and right multiplying it with a count-sketch matrix of size $\sqrt{d}\times \sqrt{m}$. We then discretize the resultant $\sqrt{m}\times \sqrt{m}$  matrix, and finally reshape this matrix and obtain the desired $m$-sized hashcode of the input $d$-dimensional vector. In both results,  we show that the probability of hash collision is monotonic to the pairwise similarity between input points. Our key contributions are summarized as follows:
\begin{itemize}[noitemsep,topsep=4pt]
    \item We suggested two proposals to improve the running time and space complexity of the \ELSH. 
Our first proposal \CSELSH~(Definition~\ref{def:CS_Eucl_LSH}) exploits the count-sketch matrix and reduces the space and sketching time complexity from $O(md)$ to $O(d)$, where $d$ denotes the dimension of the input vector, and $m$ denotes the size of hashcode. Our second proposal \HCSELSH~(Definition~\ref{def:HCS_Eucl_LSH}), exploits higher order count sketch (HCS)~\cite{shi2019higher}~tensor and reduces the time complexity to $O(d)$ from $O(md)$, and space complexity to $O(N \sqrt[N]{d})$ from $O(d)$,
 {\color{black} $O(N \sqrt[N]{d}) \ll O(d)$ for $N = o(d)$. }

    
    \item Building on the same ideas, we suggest two improvements for SRP, namely  \CSSRP~(Definition~\ref{def:cs_consine_lsh}) and $\HCSSRP$~(Definition~\ref{def:hcs_consine_lsh}). They exploit the count-sketch matrix and higher order count sketch tensor, respectively.
    \CSSRP~ reduces the space and time complexity to $O(d)$ from $O(md)$ (of SRP), whereas, \HCSSRP~reduces time complexity to $O(d)$ and space complexity to $O(N\sqrt[N]{d})$, where $O(N\sqrt[N]{d}) \ll O(d) $ for $N = o(d)$.
\end{itemize}

\ul{We would like to reiterate that our main goal is to reduce the time and space complexity of the indexing step (\textit{i.e.} the computation involved in creating the hash tables), and not on improving the query time that is often considered in the literature}~\cite{DBLP:conf/focs/AndoniI06,zheng2020pm}. We note that improving the time and space complexity of the indexing step is equally important and has been studied in \cite{dasgupta2011fast} for LSH for Euclidean distance and in \cite{andoni2015practical} for LSH for cosine similarity. We compare our proposals with them in Table \ref{tab:tab1}. {\color{black} It is worth noting that \cite{dubey2022improving} also suggested improving SRP using the count-sketch based projection matrix. However, their proof techniques are different from ours. They first show that the entries of the count-sketch matrix follow a sparse Bernoulli distribution and utilize it to prove that the coordinates of the compressed vectors follow the normal distribution under the limiting condition of   $d \to \infty$. However, our proof techniques are more generic, exploit central limit theorems,  and can be extended to higher-order count sketches.
}

\section{Related work} \label{sec:related_work}

 This section presents the baseline algorithms and compares our results with them.
 
  \paragraph*{\textcolor{black}{\textbf{LSH for Euclidean distance:}}}
  The seminal work of~\cite{indyk1998approximate,datar2004locality} suggests a sublinear time algorithm for the approximate nearest neighbor search problem for Euclidean distance. To compute an $m$-sized hashcodes of $d$ dimensional input vectors, their algorithm can be considered as projecting the input vectors on a $d\times m$ Gaussian matrix, followed by discretizing the resultant vectors. The time and space complexity of their approach (per data point) is $O(md)$. 
   
  Dasgupta \etal~\cite{dasgupta2011fast} suggested a fast locality-sensitive hash function for Euclidean distance based on the preconditioning via a sparse projection matrix and applying a randomized Fourier transform. Their proposal reduced the projection time to $O(d\log d + m)$. However, the space complexity of their proposal remains the same as that of \ELSH, \textit{i.e.}, $O(md)$.

  Our work focuses on improving the projection step by using the super sparse count-sketch projection matrix, which has only $d$ non-zero entries, one per column, instead of the dense Gaussian matrix used in~\cite{indyk1998approximate}. As a consequence, the running time and space complexity of our algorithm is $O(d)$ - achieving $O(m)$ times speedup/space saving as compared to~\cite{indyk1998approximate}. We further show the probability of collision of hash code for a data pair remains monotonic to their similarity. However, our results hold when \textcolor{black}{$m = o \left(d^{\frac{\delta}{2(2+\delta)}}\right)$, for some $\delta >0$.} 

We give another proposal that further improves the space complexity. Our idea is to view the input vector as a higher order tensor; \textit{e.g.} a $d$-dimensional vector is viewed as a $N$-order tensor {each of size $d^{1/N}$}. We then project this tensor on a suitably chosen high-order projection tensor~\cite{shi2019higher}, followed by discretizing entries of the resultant vector. The space complexity of our proposal is $O(Nd^{1/N})$, as compared to $O(md)$ space complexity of ~\cite{indyk1998approximate}. The running time of our proposal is $O(d)$. However, our result holds when \textcolor{black}{$\sqrt{m} N^{\left(\frac{4}{5}\right)} = o\left(d^{\left(\frac{3N-8}{10N}\right)} \right)$.} 

{\color{black}
The work of \cite{DBLP:conf/focs/AndoniI06} proposes ball-carving LSH for Euclidean distance, which accelerates query time compared to \ELSH. Their approach, like E2LSH, relies on space partitioning via random projections. Additionally, several query-centric variants—such as  SRS~\cite{sun2014srs}, QALSH~\cite{huang2015query}, R2LSH~\cite{lu2020r2lsh}, LCCS-LSH~\cite{lei2020locality}, PM-LSH~\cite{zheng2020pm}, VHP~\cite{lu2020vhp} and DB-LSH~\cite{tian2023db} —have been introduced that offers faster query time LSH for  Euclidean distance. At a high-level, these methods are also based  on randomized space partitioning techniques. We believe that our techniques can also be applied to these methods to give a space-efficient and faster preprocessing time algorithm. Pursuing this integration constitutes a promising direction for future work.
}

   \paragraph{\textbf{LSH for Cosine Similarity:}} Sign Random projection (SRP)~\cite{charikar2002similarity} gives LSH for real-valued vectors focusing on the pairwise cosine similarity and is also based on the idea of random projection followed by discretization of the resultant vector. Their random projection step is identical to that of \ELSH, and at the discretization step, each component of the resulting vector is assigned a value, either $0$ or $1$, based on their sign. Consequently, the SRP's time and space complexity becomes $O(md)$ per data point to compute an $m$-dimensional hashcode of a $d$ dimensional vector. Similar to the case of \ELSH, here also we give two proposals that improve the time and space complexity of SRP. Our first proposal is based on the count-sketch projection matrix, whereas the second proposal is based on the idea of a higher-order count sketch.  Our proposals are almost similar to those of improved variants of \ELSH. The only change occurs in the discretization step, which is done by considering the sign of the resultant vector. Improvement in the space and time complexity for both proposals remains identical to what we obtained in improved variants of \ELSH.

   We note that improving SRP using the count-sketch projection matrix has been recently proposed in~\cite{dubey2022improving}, and achieves the same improvement as our first proposal \CSSRP. However, we would like to highlight that our analysis technique differs significantly from that proposed in ~\cite{dubey2022improving} and extends for higher order count sketch projection matrix as well.

In addition, several improvements for SRP exist in the literature. We list a few notable results as follows: Circulant Binary Embedding (CBE)~\cite{yu2014circulant} suggests a faster and space-efficient algorithm that gives an unbiased estimate of the pairwise cosine similarity. Their core idea is to use a carefully designed circulant matrix for projection that improves the space and time complexity to $O(d)$ and $O(d \log d)$, respectively. Superbit LSH~\cite{ji2012super} offers a more accurate algorithm for pairwise cosine similarity estimation. Their idea is to create a projection matrix whose columns are orthonormal to each other. However, the accurate estimation comes at the cost of higher running time required to generate orthonormal vectors. The work of Kang \textit{et al.}~\cite{kang2018improving} offers a more accurate similarity estimation by exploiting the maximum likelihood estimation technique. However, this improved similarity estimation comes at the cost of higher running time. Table~\ref{tab:tab1} summarizes the comparison of space and time complexity among LSH families for cosine similarity and Euclidean distances.

 Several other locality-sensitive hash families (and improved versions) exist for different data types and underlying similarity/distance measures. We discuss a few of them here. Minwise independent hashing (\aka~Minhash)~\cite{BroderCFM98} and its improved variants~\cite{b-bit,OPH,DOPH} suggest LSH for sets (or binary vectors) approximating Jaccard similarity.  The result of \cite{indyk1998approximate,gionis1999similarity} suggested LSH for binary data that approximate hamming distance. Datar \textit{et al.} ~\cite{datar2004locality} gives LSH for real-valued vectors estimating pairwise $\ell_p$ distance, for $p \in [0,2]$. Our techniques do not seem trivially extend to these LSH families. We leave this as possible open directions of the work. 

\begin{table}
  \centering
  \caption{Comparison of space and time required by different LSH families to compute $m$-dimensional hashcode of a $d$-dimensional vector. $N$ denotes the order of \HCS~ tensor.
In CSSRP-L, $l$ denotes the number of non-zero entries in each column of CSSRP-L matrix. }
\vspace{0.1cm}
  \resizebox{0.9\columnwidth}{!}{%
    \begin{tabular}{lcc}
    \toprule
    \multicolumn{3}{c}{\textbf{LSH for Euclidean Distance}} \\
    \midrule
    \textbf{LSH Family} & \textbf{Time Complexity } & \textbf{Space Complexity} \\
    \midrule
    \ELSH~\cite{datar2004locality} &   $O(md)$    & $O(md)$ \\
    ACHash~\cite{dasgupta2011fast} & $O(d \log d + m)$ & O(md)\\
    \CSELSH~(Definition~\ref{def:CS_Eucl_LSH}) &     $O(d)$  & $O(d)$  \\
    \HCSELSH~(Definition~\ref{def:HCS_Eucl_LSH}) &  $O(d)$     &  $O(N\sqrt[N]{d})$\\
    \midrule
    \multicolumn{3}{c}{\textbf{LSH for Cosine Similarity}} \\
    \midrule
    SRP~\cite{charikar2002similarity}   &  $O(md)$     &  $O(md)$ \\
    CBE~\cite{yu2014circulant}   &  $O(d \log d)$     &  $O(d)$ \\
    Superbit~\cite{ji2012super} &  $O(m^2d)$ &$O(md)$ \\
    MLE~\cite{kang2018improving}   &  $O(md)$     & $O(md)$  \\
    CSSRP~\cite{dubey2022improving} &   $O(d)$    &  $O(d)$\\
    CSSRP-L~\cite{dubey2022improving} &     $O(ld)$  & $O(ld)$ \\
    \CSSRP~(Definition \ref{def:cs_consine_lsh}) & $O(d)$      &  $O(d)$\\
    \HCSSRP~(Definition \ref{def:hcs_consine_lsh}) &   $O(d)$    & $O(N\sqrt[N]{d})$ \\
    \bottomrule
    \end{tabular}%
    }
      \label{tab:tab1}
\end{table}%
\section{Notation and Background}  \label{sec:background}
  We denote vectors by bold small case letters, matrices by bold capital letters, and tensors by bold calligraphic letters, \textit{e.g.}, $\p = (p_1, \ldots, p_d) \in \R^{d}$ denotes a $d$-dimensional vector, $\mathbf{H} \in \R^{m \times d}$ represents a $m \times d$ matrix and $\TS \in \R^{d_1 \times d_2 \times \cdots \times d_{N}}$ a $N$-order tensor with dimension $d_{k}$ on each mode, for $k \in [N]$. For a vector $\p$, $||\p||$ and $||\p||_{\infty}$ denotes its $\ell_2$ and $\ell_{\infty}$ norm, respectively. 

\subsection{Approximate Nearest Neighbor Search (ANN)  and its solution via LSH} 
For a given query point, the approximate nearest neighbor search allows returning points within a distance of some constant (say $c>1$) times the distance of the query to its actual nearest neighbor. One of the popular approaches for solving the ANN problem is randomized space partitioning via Locality-Sensitive-Hashing (LSH)~\cite{indyk1998approximate}. The work of Indyk and Motwani~\cite{indyk1998approximate} suggests  LSH to construct a data structure that ensures a constant probability of finding the nearest neighbor if it is there.
For a metric space $(X, D)$ and a set $U$, their  LSH family is defined as follows:
\begin{definition}[Locality Sensitive Hashing (LSH)~\cite{indyk1998approximate}]~\label{def:LSH}
 A family $\mathcal{H} = \{h:X \rightarrow U\}$, is called $(R_1, R_2, P_1, P_2)$-sensitive family, if ~$\forall$ $\p,\q \in X$ the following conditions hold true for $P_1>P_2$ and $R_1 < R_2$:\\
    ~~$~~~~\bullet$ if $D(\p,\q) \leq R_1$, then $\Pr[h(\p) = h(\q)] \geq P_1 $,\\
    ~~$~~~~\bullet$ if $D(\p,\q) \geq R_2$, then $\Pr[h(\p) = h(\q)] \leq P_2$. 
\end{definition}

The work of Indyk and Motwani~\cite{indyk1998approximate} described how to use LSH to construct a data structure that ensures a constant probability of finding the nearest neighbor if it is there, as follows: Given a $(R_1, R_2, P_1,P_2)$-sensitive hash family $\mathcal{H} = \{h : X \rightarrow U \}$ as defined in Definition~\ref{def:LSH}, for $R_1 = R$ and $R_2 = c R$, we can increase the magnitude of the difference between probabilities $P_1$ and $P_2$ by concatenating the several independent copies of the hash function $h(\cdot)$. For any integer $m > 1$, define a new family of hash function $\bar{\mathcal{H}} =\{\hbar: X \rightarrow U^{m}\}$ as follows $\hbar(\p) = (h_{1}(\p), \ldots, h_{i}(\p), \ldots, h_{m}(\p))$, where $h_{i}(\cdot)\in \mathcal{H}$. For an integer $\tau$, independently and uniformly at random sample $\tau$ functions from $\mathcal{\bar{H}}$ \textit{i.e.} $\hbar_{1}, \ldots, \hbar_{\tau}$. Let $P$ be an input point set of interest, store every point $\p \in P$ to the bucket $\hbar_{t}(\p)$ for $t \in [\tau]$. For any query point $\q \in \R^d$, hash it using $\hbar_{t}(\q)$ for all  $t \in [\tau]$,  collect the index of all the elements of set $P$ hashed in $\hbar_{t}(\q),~ t \in [\tau]$. Check first $3\tau$ indexed whether there is an element say $\p$  such that $D(\p,\q) \leq R_2$, if yes then return ``\textit{Yes}'' and $\p$, otherwise return ``\textit{No}''. For a detailed description, refer to \cite{indyk1998approximate}.
\\





\textbf{\ELSH:} 
Datar \etal~\cite{datar2004locality} proposed a family of LSH functions for Euclidean distance called \ELSH~and is defined as follows:

\begin{definition}[{\ELSH~\cite{datar2004locality}} ]~\label{def:e2lsh}
Let $\p \in \R^d$ and $\rr$ be a $d$-dimensional vector such that each entry of $\rr$ is a sample from the standard normal distribution. In the following, we define a hash function $h_{\rr, b}: \R^d \ \rightarrow \mathbb{Z}$ that maps a $d$-dimensional vector to an integer (hashcode):
\begin{align}
h_{\rr, b} (\p) &= \left \lfloor  \frac{\langle \rr, \p \rangle + b}{w} \right \rfloor \numberthis \label{eq:datar_lsh}
\end{align}
where, $b \in [0,w]$ uniformly at random and $w>0$.
\end{definition}

Let $\p, \q \in \R^{d}$ and $R := || \p - \q||$, then from the above defined function (Equation~\eqref{eq:datar_lsh})  we can easily validate the following: 
\begin{align}
   p(R) =  \Pr[h_{r,b}(\p) = h_{r,b}(\q)] &= \int_{0}^{w} \frac{1}{R} \, f\left(\frac{t}{R}\right) \, \left(1 - \frac{t}{w} \right) dt
\end{align}
where $f(\cdot)$ represents the probability density function of the absolute value of the standard normal distribution. Moreover, as per Definition~\ref{def:LSH},  the hash function defined above in Equation~\eqref{eq:datar_lsh} is $(R_1,R_2, P_1,P_2)$-sensitive for $P_1 = p(1)$, $P_2 = p(R)$ and $R_2/R_1 = R$; and its evaluation requires $O(d)$ operations.\\ 

\textbf{Sign Random Projection (SRP/SimHash):} 
The work of Charikar~\cite{charikar2002similarity} introduced a random projection-based LSH for cosine similarity, popularly known as Sign Random Projection (SRP). We state it as follows:

\begin{definition}[{Sign Random Projection  (SRP)~\cite{charikar2002similarity}}] \label{def:srp}  
Let $\mathbf{u} \in \R^d$. We define \SRP~hash function $\zeta: \R^d \rightarrow \{0,1\}$  as 
$\zeta(\mathbf{u}) = \sgn(\langle\rr, \mathbf{u} \rangle)$, 
where, $r_{j} \sim \mathcal{N}(0,1)~ \forall~j \in [d]$ and  $\sgn\left(\langle \rr, \mathbf{u} \rangle \right) = 1$ if $\langle \rr, \mathbf{u} \rangle > 0$, otherwise, $\sgn\left(\langle \rr, \mathbf{u} \rangle \right) = 0$.
\end{definition}

Let $\mathbf{u}, \mathbf{v} \in \R^d$ and $\theta_{(\mathbf{u},\mathbf{v})}=\cos^{-1}\left(\frac{\mathbf{u}^T \mathbf{v}}{\|\mathbf{u}\| \|\mathbf{v}\|}\right)$ denotes the angular similarity between $\mathbf{u}$ and $\mathbf{v}$. Then from~\cite{goemans1995improved,shrivastava2014defense},  the probability of collision satisfies the following equation: 
\begin{align}
    \Pr\left[ \sgn(\langle \rr,\mathbf{u} \rangle) = \sgn(\langle \rr, \mathbf{v}\rangle)\right] = 1 - \frac{\theta_{(\mathbf{u},\mathbf{v})}}{\pi}. \numberthis \label{eq:eq191022}
\end{align}

Here, $1 - \theta_{(\mathbf{u}, \mathbf{v})}/\pi$ is monotonic decreasing function  in cosine similarity. Let $S := \frac{\mathbf{u}^T \mathbf{v}}{\|\mathbf{u}\| \|\mathbf{v}\|}$, where $\frac{\mathbf{u}^T \mathbf{v}}{\|\mathbf{u} \| \|\mathbf{v}\|}$ denotes the cosine similarity between the vectors $\mathbf{u}$ and $\mathbf{v}$. Then, using Definition~\ref{def:LSH},  \SRP ~is $\left(S, c S, \big(1- \cos^{-1}(S)/\pi \big), \big(1- \cos^{-1}(cS)/\pi \big) \right)$ -sensitive LSH family.


\subsection{\CS ~and \HCS:}
\begin{definition}[\CS~(CS)~\cite{count_sketch}] \label{def:def_cs}
Let $\p  \in \R^{d}$ and $\csb{\p} = \Big(\cs(\p)_1, \ldots, $ $\cs(\p)_m \Big) \in \R^{m}$ be a count-sketch of $\p$, then for any $l\in[m]$, $\cs(\p)_l$ is defined as follows:
\begin{align}
    \cs(\p)_l &:= \sum_{h(i) = l,~ i \in [d]} s(i) p_i, \label{eq:eq_cs}
\end{align}
where  $h:[d] \rightarrow [m]$ and $s:[d] \rightarrow \{+1,-1\}$ are 2-wise independent random hash functions.
\end{definition}
\begin{definition} [{\HCS~(HCS)~\cite{shi2019higher}}] \label{def:hcs}
Let $\p \in \R^d$,  $h_{k}:[d_{k}] \rightarrow [m_k]$, $s_{k}:[d_{k}] \rightarrow \{+1,-1\}$ are 2-wise independent random hash functions, where  $k\in[N]$ and $d= \prod_{k=1}^{N} d_{k}$. Then $\hcsb{\p}$ denotes a higher order count sketch of vector $\p$ and is defined element-wise as follows:
\begin{align*}
  \hcs(\p)_{l_1, \ldots, l_N} &:=  \sum_{h_1(i_1) = l_1, \ldots, h_{N}(i_N) = l_N} \hspace*{-1.2cm} s_1(i_1) s_2(i_2) \cdots s_N(i_N) p_{j}, \numberthis
  \label{eq:eqhcs}  
\end{align*}
where, $i_k \in [d_k]$  and $  j= \sum_{k=2}^{N} \left(i_{k} \prod_{p=1}^{k-1} d_{p} \right) + i_1$ is the index mapping between the vector and its reshaping
result - a $N$-order tensor with dimensions $d_{k}$ on each mode, for $k \in [N]$.
\end{definition}
\begin{remark}~\textbf{\cite{shi2019higher}} \label{rem:hcs_space}
We note that 
each mode of HCS requires a $(m_k \times d_{k})$ sized hash matrix with $d_k = O(\sqrt[N]{d})$ nonzero entries. Therefore, its space complexity is  $O(N \sqrt[N]{d})$, whereas the space complexity of  CS is $O(d)$.  Thus, when $N = o(d)$, $O(N\sqrt[N]{d}) \ll  O(d) $ and  HCS offers a more space efficient solution than that of CS. 
\end{remark}

\subsection{Central limit theorems}

\begin{thm}[{Univariate Lyapunov CLT~\cite{feller1}}]~\label{thm:univariate_lyapunov}
Suppose $ \{X_{1},\ldots,X_{d}\}$ is a sequence of independent random variables, each with finite expected value $ \E[X_{i}]$ and variance $ \Var[X_{i}]$. If for some $\delta >0$ the following condition holds true
\begin{align*}
    \lim_{d \to \infty} \frac{1}{s_{d}^{2+\delta}} \sum_{i=1}^{d} \E\left[ |X_{i} - \E[X_i]|^{2+\delta}\right] &= 0  \numberthis   \label{eq:univariate_lyapunov}
\end{align*}
then
\begin{align*}
    \quad \frac{1}{s_{d}} \sum_{i=1}^{d} (X_{i} - \E[X_{i}]) &\overset{\mathcal{D}}{\to} \mathcal{N}(0,1) \numberthis
\end{align*} 
as $d$ tends to infinity. Where  $s_{d}^2 = \sum_{i=1}^{d} \Var(X_{i})$ and $\overset{\mathcal{D}}{\to}$ indicates the convergence in distribution.
\end{thm}

\begin{thm}[Multivariate Lyapunov CLT~\cite{feller1}] \label{thm:multivariate_Lyapunov_CLT}
Let  $\{\mathbf{X}_{1},\ldots, \mathbf{X}_{d}\}$ be a sequence of independent random vectors such that each entry of both the expected value of the random vector $\{\mathbf{X}_{i}\}_{i=1}^d$, and the corresponding covariance matrix $\mathbf{\Sigma}_{i}$, is finite.  
If for some $ \delta >0$  the following is true
\begin{align*}
  &\lim_{d \to \infty} \left \|\V_{d}^{-\frac{1}{2}} \right \|^{2+\delta} \sum_{i=1}^{d} \E\left[||\mathbf{X}_{i} - \E[\mathbf{X}_{i}]||^{2 +\delta} \right] = 0 \numberthis \label{eq:multivariate_lyapunov}
\end{align*}
then
\begin{align*}
  \V_{d}^{-\frac{1}{2}} \sum_{i=1}^{d} (\mathbf{X}_{i} -\E[\mathbf{X}_{i}]) \overset{\mathcal{D}}{\to} \mathcal{N}(\mathbf{0}, \mathbf{I}) 
\end{align*}
as $d$ tends to infinity. Where  $\V_{d} = \sum_{i=1}^{d} \mathbf{\Sigma}_{i}$ and $\overset{\mathcal{D}}{\to}$ denotes the convergence in distribution;  $\mathbf{0}$ and  $\mathbf{I}$  denote zero vector and  the identity matrix, respectively. 
\end{thm}

\begin{thm}[\cite{janson1988normal}]~ \label{thm:clt_grpah_var}
Let $\{Y_{1}, \ldots, Y_{d}\}$ be a family of bounded random variables, \textit{i.e.}, $|Y_{i}| \leq A$ for $i \in [d]$. Suppose that $\Gamma_{d}$ is a dependency graph for this family where each node represents a random variable $Y_i$ and an edge exists between nodes $i$ and $j$ if $Y_i$ and $Y_j$ are dependent. Let $M$ be the maximal degree of $\Gamma_{d}$ (if $\Gamma_{d}$ has no edges, in that case, we set $M=1$). Define $S_{d} := \sum_{i=1}^d Y_{i}$ and $\sigma_{d}^2 = \Var(S_d)$. If there exists an integer $\alpha$ such that 
\begin{align}
    &\lim_{d \to \infty} \left(\frac{d}{M}\right)^{\frac{1}{\alpha}} \frac{M A}{\sigma_{d}} = 0 \quad \text{then} \quad  
    \frac{S_d - \E[S_d]}{\sigma_d} \overset{\mathcal{D}}{\to} \mathcal{N}(0, 1)  \text{ as } d \rightarrow \infty. \numberthis \label{eq:eq_clt_graph_var}
\end{align}
\end{thm}
 
{\color{black}

In  Theorem~\ref{thm:clt_grpah_vec} stated below, we provide a multivariate extension of Thoerem~\ref{thm:clt_grpah_var}. Our extension is based on the application of the Cramer Wold Device theorem.

\begin{thm}[Cramer Wold Device \footnote{In the Cramer-Wold device, considering a unit vector or a general vector \(\mathbf{a} \in \mathbb{R}^r\) is equivalent. This equivalence can be easily proven using the following fact: "If \(\{X_n : n = 1, 2, \ldots\}\) is a sequence of random variables converging in distribution to \(X\), then for any constant \(c\), \(cX_n\) also converges in distribution to \(cX\)."}~\cite{ billingsley2013convergence,wooldridge2010econometric}]\label{thm:cramer}

   A sequence of $r$-dimensional random vectors \(\{\mathbf{X}_n : n = 1, 2, \ldots\}\) converges in distribution to a random vector $\mathbf{X}$ if and only if for every unit vector $\mathbf{a} \in \R^r$, $\langle \mathbf{a}, \mathbf{X}_n \rangle$ converges in distribution to $\langle \mathbf{a}, \mathbf{X} \rangle$. 
\end{thm}

\begin{thm} \label{thm:clt_grpah_vec}
    Let $\{\Y_{1}, \ldots, \Y_{d}\}$ be a family of $r$-dimensional bounded random vectors, \textit{i.e.}, $\|\Y_i\| \leq A$ for all $i \in [d]$ where $\|\Y_i\|$ denotes the $\ell_2$ norm of vector $\Y_i$. Suppose that $\Gamma_{d}$ is a dependency graph for this family and let $M$ be the maximal degree of $\Gamma_{d}$ (in the dependency graph, we have one node corresponding to each random vector $\Y_i$, and we put an edge between two nodes $\Y_i$ and $\Y_j$, if they are dependent; further, if $\Gamma_{d}$ has no edges, then we set $M=1$). Let  $\mathbf{S}_{d} := \sum_{i=1}^d \Y_{i}$ and $\boldsymbol{\Sigma}_{d} = \Cov(\mathbf{S}_d)$. For any unit vector $\a \in \R^d$, define a random variable $\ParSumVar_{d} := \sum_{i=1}^d \mathbf{a}^T\Y_{i}$, with its variance denoted by $\sigma_d^2$. If there exists an integer $\alpha$ such that
    \begin{align}
        \lim_{d \to \infty} \left(\frac{d}{M}\right)^{\frac{1}{\alpha}} \frac{M A}{\sigma_{d}} = 0 \label{eq:eq051024_0}
    \end{align}
    then 
    \begin{align}
        \mathbf{S}_d \overset{\mathcal{D}}{\to} \mathcal{N} \left(\E[\mathbf{S}_d], \boldsymbol{\Sigma}_{d} \right)  \text{ as } d \rightarrow \infty.
    \end{align}
\end{thm}
\begin{proof}
Let $\Gamma_{d}$ be a dependency graph with maximal degree $M$ corresponding to a family of   $r$-dimensional bounded random vectors $\{\Y_{1}, \ldots, \Y_{d}\}$ where $\|\Y_i\| \leq A$ for $i \in [d]$. Assume that $\mathbf{S}_d:= \sum_{i=1}^d \Y_i$. We denote the  expected value and covariance of $\mathbf{S}_d$ by $\E[\mathbf{S}_d]$ and $\boldsymbol{\Sigma}_d$, respectively. For any unit vector $\a \in \R^d$, define  a bounded family of random variables $X_i : = \a^T \Y_i$ where $|X_i| = \|\a^T\Y_i\| \leq \|\a\| \|\Y_i\| = A$ for $i \in [d]$. Using these $X_i$'s, we define a new random variable  $\ParSumVar_d := \sum_{i=1}^d X_i$. We can write the expected value and variance of $\ParSumVar_d$ in terms of expected value and covariance  of $\mathbf{S}_d$ as follows:  $\E[\ParSumVar_d]= \a^T\E[\mathbf{S}_d]$ and $\sigma_d^2 = \Var(\ParSumVar_d) = \a^T \boldsymbol{\Sigma}_d \a$. For any positive integer $\alpha$, suppose that the following holds true 
\begin{align}
    \lim_{d \to \infty} \left(\frac{d}{M}\right)^{\frac{1}{\alpha}} \frac{M A}{\sigma_{d}} = 0
\end{align}
then for $d \rightarrow \infty$ from Theorem~\ref{thm:clt_grpah_var}, we have
\begin{align}
       &\frac{\ParSumVar_d - \E[\ParSumVar_d]}{\sigma_d} \overset{\mathcal{D}}{\to} \mathcal{N}(0, 1) \notag\\
       \implies & \ParSumVar_d \overset{\mathcal{D}}{\to} \mathcal{N}\left(\E[\ParSumVar_d], \sigma_d^2 \right)  \notag\\
       \implies & \ParSumVar_d = \a^T \S_d \overset{\mathcal{D}}{\to} \langle \a, \mathbf{Z}\rangle \text{ where } \mathbf{Z} \sim \mathcal{N}\left(\E[\mathbf{S}_d], \boldsymbol{\Sigma}_d \right). \numberthis \label{eq:eq051024_1}
    \end{align}
   Form Equation~\eqref{eq:eq051024_1} and Cramer Wold Device Theorem (Theorem~\ref{thm:cramer}), we have
    \begin{align}
        &\mathbf{S}_d \overset{\mathcal{D}}{\to} \mathbf{Z} \label{eq:eq051024_2}
    \end{align}
    Equation~\eqref{eq:eq051024_2} completes a proof of the theorem.
\end{proof}

}

\section{Improving \ELSH} \label{sec:euclidean}

\subsection{Improvement Using Count Sketch}

In this section, we introduce our proposal that improves \ELSH~\cite{datar2004locality}~ using count-sketch. The idea behind our proposal is to use a sparse count-sketch projection matrix instead of the standard dense projection matrix used in random projection. The main challenge of the work is to show that our proposal is valid LSH; that is, the collision probability of hashcodes is monotonic to pairwise similarity. To compute an $m$-size hashcode for $d$-dimensinal real-valued vector, we project it on a $d\times m$ count-sketch projection matrix that has only $d$ non-zeros; one per column. The resulting $m$-dimensional vector is then discretized, providing the desired hashcode. The main bottleneck is to show that after projection, entries of the resultant vectors follow the normal distribution and are independent of each other. 
In Theorem~\ref{thm:uni_cs}, we show that entries of the resultant projected vector follow the asymptotic univariate normal distribution. Then in Theorem~\ref{thm:mul_cs_new}, we show that the $m$-dimensional vectors obtained after projection follow the asymptotic multivariate normal distribution. In Corollary~\ref{cor:indp_cs}, using Theorems~\ref{thm:uni_cs} and \ref{thm:mul_cs_new}, we show that entries of the projected vector are independent of each other. Finally, in Theorem~\ref{thm:CS_Eucl_LSH_property}, we show that the hashcode obtained using our proposal satisfies the property of LSH. We start with stating our LSH for Euclidean distance as follows:
\begin{definition}[\CSELSH]  \label{def:CS_Eucl_LSH}
 Let $\p \in \R^d$ and $\csb{\p} \in \R^m$ be the corresponding  count-sketch vector obtained using Equation~\eqref{eq:eq_cs}. We define a hash function $g : \R^{d} \rightarrow \Z^{m}$ element-wise as 
\begin{align}
    g(\p)_{l} = \left\lfloor \frac{\sqrt{m}\cdot\cs(\p)_{l} + b}{w} \right\rfloor
\end{align}
where $l \in [m]$, $b \in [0,w]$ uniformly at random and $w > 0$.
\end{definition}

In the following theorem, we state the asymptotic normality of the bins/buckets of the count-sketch. 
\begin{thm}~\label{thm:uni_cs}
Let $\p\in \R^{d}$ and $\csb{\p} \in \R^{m}$ be the corresponding count-sketch vector obtained using Equation~\eqref{eq:eq_cs}. If  $\forall ~ j \in [d]$, $\E[|p_{j}|^{2 + \delta}]$ is finite and $m=o(d)$ for some $\delta > 0$, then as $d \rightarrow \infty$, for $l \in [m]$, we have
\begin{align}
    \cs(\p)_{l} \overset{\mathcal{D}}{\to} \mathcal{N} \left( \mathbf{0},
\frac{||\p||^2}{m}
 \right).
\end{align} 
\end{thm}
\begin{proof}
Let $\p \in \R^{d}$ and $\csb{\p} \in \R^{m}$ be its count-sketch vector. Hence, from the definition of count-sketch (Definition~\ref{def:def_cs}), we have
\begin{align*}
   \cs(\p)_{l} &= \sum_{\substack{h(j) = l ~s.t.~ j \in [d]}} s(j) p_j  = \sum_{j \in [d]} K(j,l) s(j) p_j.
\end{align*}
where $l \in [m]$ and $K(j,l)$ is an indicator random variable that takes value $1$ if $h(j) = l$ and $0$ otherwise. Let 
\begin{align}
    X_{j} &:= K(j,l) s(j) p_j, \text{ where } j \in [d].
\end{align}
The expected value and variance of the $X_{j}$ are
\begin{align*}
\E[X_{j}] &= \E \left[K(j,l)s(j) p_{j} \right] = \E [K(j,l)] \E[s(j)] p_{j} = 0, \qquad \left[\because \E[s(j)] = 0 \right]. \numberthis \label{eq:eq_cs_univ_exp}\\
\Var \left[ X_{j}\right] &= \E[X_{j}^2] - \left(\E[X_j]\right)^2 = \E[K(j,l)^2 s(j)^2 p_{j}^2] - 0
          = \frac{p_{j}^2}{m}. \numberthis \label{eq:eq_cs_univ_var}
 \end{align*}
 Equation \eqref{eq:eq_cs_univ_var} holds because $K(j,l)^2 = K(j,l)$, $s(j)^2 = 1$ and $\E[K(j,l)] = 1/m$. 
Let 
\begin{align*}
    s_{d}^2 & := \sum_{j = 1}^{d} \Var(X_{j})= \sum_{j = 1}^{d} \frac{p_{j}^2}{m} = \sum_{j = 1}^{d} \frac{p_{j}^2}{m} = \frac{||\p||^{2}}{m}.\\
    \implies s_{d}^{2 +\delta} &= \left(\frac{||\p||^{2}}{m} \right)^{\frac{2 + \delta}{2}}.\numberthis \label{eq:eq_03001}
\end{align*}
To prove the theorem, we have to show that the following Lyapunov condition holds true
\begin{align*}
    \lim_{d \to \infty} \frac{1}{s_{d}^{2+\delta}} \sum_{j=1}^{d} \E\left[ |X_{j} - \E[X_j]|^{2+\delta}\right] = 0.
\end{align*}
First, we compute the following
\begin{align*}
    \E\left[ |X_{j} - \E[X_j]|^{2+\delta}\right]  &= \E\left[ |K(j,l) s(j) p_{j}|^{2 +\delta}\right]
    = \frac{|p_{j}|^{2+\delta}}{m}. \numberthis \label{eq:eq_03002}
\end{align*}
Euqation \eqref{eq:eq_03002} holds because $|s(j)| = 1,~ |K(j,l)|^{2 + \delta}=K(j,l)$ and $\E[K(j,l)] = 1/m$. Therefore, from Equations~\eqref{eq:eq_03001} and \eqref{eq:eq_03002}, we have
\begin{align*}
    \frac{1}{s_{d}^{2+\delta}} \sum_{j=1}^{d} \E\left[ |X_{j} - \E[X_j]|^{2+\delta}\right] &=  \frac{\sum_{j=1}^{d} |p_{j}|^{2+\delta}/m}{\left(\|\p\|^{2}/{m} \right)^{\frac{2 + \delta}{2}}}
    = m^{\frac{\delta}{2}} \frac{\sum_{j=1}^{d} |p_{j}|^{2+\delta} }{\left( \sum_{j = 1}^{d} p_{j}^2 \right)^{\frac{2 + \delta}{2}}}\\
    &= \left( \frac{m}{d}\right)^{\frac{\delta}{2}} \frac{\sum_{j=1}^{d} |p_{j}|^{2+\delta}/d }{\left( \sum_{j = 1}^{d} p_{j}^2/{d} \right)^{\frac{2 + \delta}{2}}}
    = \left( \frac{m}{d}\right)^{\frac{\delta}{2}} \frac{\E[ |p_{j}|^{2+\delta}]}{\left( \E[p_{j}^2] \right)^{\frac{2 + \delta}{2}}}\\
    &\rightarrow 0 \text{ as } d \rightarrow \infty \text{ and } m=o(d). \numberthis \label{eq:eq03003}
\end{align*}
Equation~\eqref{eq:eq03003} holds as $m=o(d)$ and 
 $\E[|p_{j}|^{2 + \delta}] < \infty$ and $\E[p_{j}^2] <  \infty.$ 
Thus, from  Theorem~\ref{thm:univariate_lyapunov} and Equation~\eqref{eq:eq03003}, we have 
\begin{align*}
    \frac{\sum_{j=1}^{d}\left(X_{j} - \E[X_{j}] \right)}{s_{d}} & \overset{\mathcal{D}}{\to} \mathcal{N}(0,1) \\
    \implies \cs(\p)_{l} = \sum_{j=1}^{d} X_{j}&\overset{\mathcal{D}}{\to} \mathcal{N}(0,s_{d}^2) = \mathcal{N}\left(0,\frac{||\p||^2}{m}\right). 
\end{align*}
\end{proof}

{\color{black}
\begin{remark}
Theorem~\ref{thm:uni_cs} requires the condition $\E[|p_{j}|^{2 + \delta}]$ is finite for $\delta>0$ to guarantee asymptotic normality. This condition generally requires that each component of data points is more or less equally important. Data points in which one (or a few) components dominate significantly behave like outliers and can be handled separately. Thus, assuming that second (or higher) moments are finite is reasonable for ensuring uniform importance across all data dimensions. Similar assumptions are often considered for indexing algorithms~\cite {li2006very,li2006improving,dubey2022improving}.
\end{remark} 
}

We now prove that the buckets/bins of the count-sketch become independent as $d \rightarrow \infty$. To do so, we first show that the joint distribution of the $m$-sized vector obtained from the count-sketch follows an asymptotically multivariate normal distribution. Consequently, we show that the joint probability distribution of bins equals the product of the corresponding distribution of individual bins. In the following theorem, we state the asymptotic multivariate normality of the count-sketch vector.
\begin{thm}~\label{thm:mul_cs_new}
Let $\p\in \R^{d}$ and $\csb{\p} \in \R^{m}$ be the corresponding count sketch vector obtained using Equation~\eqref{eq:eq_cs}. If  $\forall ~ j \in [d]$, $\E[|p_{j}|^{2 + \delta}]$ is finite and $m = o \left(d^{\frac{\delta}{2(2+\delta)}}\right)$ for some $\delta > 0$, then 
\begin{align}
    \cs(\p) \overset{\mathcal{D}}{\to} \mathcal{N} \left( \mathbf{0},
\V_{d}
 \right)
\end{align}
as $d \rightarrow \infty$, where $\mathbf{0}$ is zero vector and $\V_{d}$ is a diagonal matrix with $\|{\p}\|^2/{m}$ as diagonal entries.
\end{thm}
\begin{proof}
Let $\p \in \R^{d}$ and $\csb{\p} \in \R^{m}$ be its count-sketch vector. So, from the definition of count-sketch (Definition~\ref{def:def_cs}), we have
\begin{align*}
   \cs(\p)_{l} &= \sum_{\substack{h(j) = l ~s.t.~ j \in [d]}} s(j) p_j = \sum_{j \in [d]} K(j,l) s(j) p_j
\end{align*}
where $l \in [m]$ and $K(j,l)$ is an indicator random variable that takes value $1$ if $h(j) = l$, otherwise takes value $0$. Let 
\begin{align*}
    \mathbf{X}_{j} &:= \begin{bmatrix}
                        K(j,1) s(j) p_j\\
                        \vdots\\
                        K(j,m) s(j) p_j
                        \end{bmatrix}, \text{ where } j \in [d].
\end{align*}
The expected value and covariance matrix of the $\mathbf{X}_{j}$ are
\begin{align*}
\E[\mathbf{X}_{j}] & 
                        = \begin{bmatrix}
                        \E \left[K(j,1)\right] \E\left[s(j) \right] p_{j}\\
                        \vdots\\
                        \E \left[K(j,m)\right] \E\left[s(j)\right] p_{j}
                        \end{bmatrix} = \begin{bmatrix}
                        0\\
                        \vdots \\
                        0
                        \end{bmatrix} = \mathbf{0}, \quad \left[\because \E[s(j)] = 0 \right]. \numberthis \label{eq:eq_cs_mul_e2lsh_exp}\\
    \Cov \left[ \mathbf{X}_{j}\right] &= \begin{bmatrix}
         \frac{p_{j}^2}{m}& \cdots & 0\\
         \vdots & \ddots& \vdots\\
         0 & \cdots &  \frac{p_{j}^2}{m}
    \end{bmatrix} \numberthis \label{eq:eq_110223_1}.
\end{align*}
Equation~\eqref{eq:eq_110223_1} holds due to the following:
\begin{align*}
    &\Cov[K(j,l_1) s(j) p_{j}, K(j,l_2) s(j) p_{j}]  \notag\\
    & \quad = \E[K(j,l_1)  s(j) p_{j} K(j,l_2)  s(j) p_{j}]  - \E[K(j,l_1) s(j) p_{j}] \E[K(j,l_2) s(j) p_{j}]\\
    & \quad = \E[K(j,l_1)K(j,l_2)] p_{j}^2 
    = \begin{cases}
    \frac{p_{j}^2}{m}, & \text{if $l_1 = l_2$}\\
    0, &\text{otherwise}.
    \end{cases} \numberthis \label{eq:eq_270922}
 \end{align*}
Equation~\eqref{eq:eq_270922} holds true because $E[K(j,l_1)K(j,l_2)] = 0$  if $l_1 \neq l_2$ and $E[K(j,l_1)K(j,l_2)] = [K(j,l_1)^2] = 1/m$  if $l_1 = l_2$. Let
\begin{align*}
    \V_{d} & := \sum_{j = 1}^{d} \Cov[\mathbf{X}_{j}]= \begin{bmatrix}
         \frac{||\p||^{2}}{m}& \cdots & 0\\
         \vdots & \ddots& \vdots\\
         0 & \cdots &  \frac{||\p||^{2}}{m}
    \end{bmatrix} \\
    \implies \V_{d}^{-1} &= \frac{m}{||\p||^2}
    \begin{bmatrix}
         1& \cdots & 0\\
         \vdots & \ddots& \vdots\\
         0 & \cdots &  1
    \end{bmatrix}.\\	
    \left \|\V_{d}^{-1/2} \right\|_{F}^2 &= \Tr \left( \V_{d}^{-1/2} ( \V_{d}^{-1/2} )^{T}\right)
      = \Tr \left( \V_{d}^{-1/2} \V_{d}^{-1/2} \right) 
     = \Tr(\V_{d}^{-1})= \frac{m^2 }{||\p||^2}. \numberthis \label{eq:eq_new_03001}
\end{align*}
Equation~\eqref{eq:eq_new_03001} uses the fact that the square root of a positive definite matrix is a symmetric matrix. In order to prove the theorem, we have to show that the following holds
\begin{align*}
    \lim_{d \to \infty} \left\|\V_{d}^{-1/2} \right\|^{2+\delta} \sum_{j=1}^{d} \E\left[ \|\mathbf{X}_{j} - \E[\mathbf{X}_j]\|^{2+\delta}\right] = 0.
\end{align*}
First, we compute the following
\begin{align*}
   \hspace{-0.1cm} \E\left[ \|\mathbf{X}_{j} - \E[\mathbf{X}_j]\|^{2+\delta}\right] 
     &= \E\left[ \left(\sum_{l=1}^{m} K(j,l)^2 s(j)^2 p_{j}^2 \right)^{\frac{2 +\delta}{2}}\right]
     \\
    &
    = \E\left[ \left(\sum_{l=1}^{m} K(j,l)\right)^{\frac{2 +\delta}{2}}  |s(j)|^{2+\delta}  |p_{j}|^{2+\delta} \right]\\
    & = \E\left[ \left(\sum_{l=1}^{m} K(j,l)\right)^{\frac{2 +\delta}{2}}\right]|p_{j}|^{2+\delta}
    = |p_{j}|^{2+\delta}. \numberthis \label{eq:eq_new_03002}
\end{align*}
From Equations~\eqref{eq:eq_new_03001} and \eqref{eq:eq_new_03002}, we have
\begin{align*}   
\left \|\V_{d}^{-1/2} \right \|^{2+\delta} \sum_{j=1}^{d} \E\left[ \|\mathbf{X}_{j} - \E[\mathbf{X}_j]\|^{2+\delta}\right] 
   &= \left(\frac{m^2}{||\p||^2}\right)^{\frac{2+\delta}{2}} \cdot \sum_{j=1}^{d}|p_{j}|^{2+\delta}\\
   &= m^{2 +\delta} \cdot \frac{\sum_{j=1}^{d}|p_{j}|^{2+\delta}}{\left(||\p||^2\right)^{\frac{2+\delta}{2}} }\\
   &= m^{2 +\delta}  \cdot\frac{\sum_{j=1}^{d}|p_{j}|^{2+\delta}}{\left(\sum_{j=1}^{d}p_{j}^2\right)^{\frac{2+\delta}{2}} }
   \\
   &
   = \frac{m^{2 +\delta}}{d^{\frac{\delta}{2}}} \cdot \frac{\sum_{j=1}^{d}\frac{|p_{j}|^{2+\delta}}{d}}{\left(\sum_{j=1}^{d} \frac{p_{j}^2}{d}\right)^{\frac{2+\delta}{2}} }
   \\
   &
   = \left(\frac{m}{d^{\frac{\delta}{2(2+\delta)}}}\right)^{2 +\delta} \cdot \frac{\E[|p_{j}|^{2+\delta}]}{\left(\E[p_{j}^2]\right)^{\frac{2+\delta}{2}} }\\
   & \rightarrow 0 \text{ as } d \rightarrow \infty \text{ and } m = o \left(d^{\frac{\delta}{2(2+\delta)}}\right). \numberthis \label{eq:eq_new_03003}
\end{align*}
Equation~\eqref{eq:eq_new_03003} holds as $m = o \left(d^{\frac{\delta}{2(2+\delta)}}\right)$ and 
 $\E[|p_{j}|^{2 + \delta}] < \infty$ and $\E[p_{j}^2] <  \infty.$ 
Thus, from Theorem~\ref{thm:multivariate_Lyapunov_CLT}, we have 
\begin{align*}
    \V_{d}^{-\frac{1}{2}} \sum_{i=1}^{d} (\mathbf{X}_{i} -\E[\mathbf{X}_{i}])  & = \V_{d}^{-\frac{1}{2}} \sum_{j=1}^{d} \mathbf{X}_{j} \overset{\mathcal{D}}{\to} \mathcal{N}(\mathbf{0},\mathbf{I})
    . \\
    \implies \cs(\p) &= \sum_{j=1}^{d} \mathbf{X}_{j} \overset {\mathcal{D}}{\to} \mathcal{N}(\mathbf{0},\V_{d}). 
\end{align*}
\end{proof}

Building on the results of Theorems~\ref{thm:uni_cs} and \ref{thm:mul_cs_new}, in the following corollary, we show the distributional independence of the buckets of the count-sketch.

\begin{cor} \label{cor:indp_cs}
Let $\p\in \R^{d}$ and $\csb{\p} \in \R^{m}$ be the corresponding count-sketch vector obtained using Equation~\eqref{eq:eq_cs}. If  $\forall ~ j \in [d]$, $\E[|p_{j}|^{2 + \delta}]$ is finite and $m = o \left(d^{\frac{\delta}{2(2+\delta)}}\right)$ for some $\delta > 0$, then as $d \rightarrow \infty$, the elements  $\cs(\p)_1, \ldots, \cs(\p)_{m}$ are independent.
\end{cor}

\begin{proof}
 From Theorem~\ref{thm:mul_cs_new}, we have 
 \begin{align*}
     \csb{\p} \sim \mathcal{N}(\mathbf{0},\V_{d}) \text{ where } \mathbf{0} \in \R^m \text{ and }  \V_{d} &= \begin{bmatrix}
          \frac{||\p||^2}{m} & \cdots & 0\\
          \vdots & \ddots & \vdots \\
          0 & \cdots & \frac{||\p||^2}{m}
     \end{bmatrix} \in \R^{m \times m}.
 \end{align*}
 Therefore, 
 \begin{align*}
     &\Pr\left(\cs(\p)= (\Tilde{p}_{1}, \cdots, \Tilde{p}_{m}) \right) \\
     &= \frac{1}{\sqrt{(2\pi)^m det( \V_{d})}} \cdot \exp{\left(-\frac{1}{2} (\Tilde{p}_{1}, \cdots, \Tilde{p}_{m})  \V_{d}^{-1} (\Tilde{p}_{1}, \cdots, \Tilde{p}_{m})^{T} \right)}\\
     &= \frac{1}{\sqrt{(2\pi)^m det( \V_{d})}} \cdot \exp{\bigg(-\frac{1}{2} (\Tilde{p}_{1}, \cdots, \Tilde{p}_{m}) \frac{m}{||\p||^2}}  \begin{bmatrix}
         1& \cdots & 0\\
         \vdots & \ddots& \vdots\\
         0 & \cdots &  1
    \end{bmatrix}  (\Tilde{p}_{1}, \cdots, \Tilde{p}_{m})^{T} \bigg)\\
    &= \frac{1}{\sqrt{(2\pi)^m det( \V_{d})}} \cdot \exp{\left(-\frac{m}{2||\p||^2} \sum_{l=1}^{m} \Tilde{p}_{l}^2 \right)}\\
    &= \frac{1}{\sqrt{(2\pi)^m (||\p||^2/m)^m}} \cdot \exp{\left(-\frac{m}{2||\p||^2} \sum_{l=1}^{m} \Tilde{p}_{l}^2 \right)}\\
    &= \frac{1}{\sqrt{2\pi (||\p||^2/m)}} \cdot \exp{\left(-\frac{m
    }{2\|\p\|^2} \Tilde{p}_1^2\right)} \times \cdots\\
    &\qquad \qquad \qquad \qquad \qquad \qquad \qquad \cdots \times \frac{1}{\sqrt{2\pi (||\p||^2/m)}} \cdot \exp{\left(-\frac{m
    }{2\|\p\|^2} \Tilde{p}_m^2\right)}\\
    & \quad= \Pr\left(\cs(\p)_{1} = \Tilde{p}_1 \right)\times  \cdots \times\Pr \left(\cs(\p)_{m} = \Tilde{p}_m \right).  
 \end{align*}
 Hence, $\Pr\left(\cs(\p) = (\Tilde{p}_{1}, \cdots, \Tilde{p}_{m}) \right) = \Pr\left(\cs(\p)_{1} = \Tilde{p}_1 \right) \cdots \Pr \left(\cs(\p)_{m} = \Tilde{p}_m \right)$ implies $\cs(\p)_1,...,\cs(\p)_{m}$ are independent.
\end{proof}





The following theorem states that the probability of collision for any two points using our proposal \CSELSH~is inversely proportional to their pairwise Euclidean distance. 
\begin{thm}\label{thm:CS_Eucl_LSH_property}
Let $\p, \q \in \R^d$ such that $R = || \p -\q ||$ and $m=o\left(d^\frac{\delta}{2(2+\delta)}\right)$, then  asymptotically the following holds true 
\begin{align} 
    &p(R) = Pr[g(\p)_l = g(\q)_l]  = \int_0^{w} \frac{1}{R} \, f\left(\frac{t}{R} \right) \, \left(1 - \frac{t}{w} \right)dt, \text{ and,}\numberthis \label{eq:eq_datar_prob_cs}\\
    &\Pr\left[g(\p)_1 = g(\q)_1, \ldots, g(\p)_m = g(\q)_m \right] = \left( p(R) \right)^m. \numberthis \label{eq:eq201022}
\end{align}
where $f(\cdot)$  denotes the density function of the absolute value of the standard normal distribution.
\end{thm}
\begin{proof} As we know that $m=o\left(d^\frac{\delta}{2(2+\delta)}\right)$ implies $m=o(d)$. Thus, from Theorem~\ref{thm:uni_cs}, we have, $\cs(\p)_{l}$ and $\cs(\q)_{l}$ asymptotically follows $\mathcal{N}\left(0,  \frac{\|\p\|^2}{m}\right)$ and $\mathcal{N}\left(0,  \frac{\|\q\|^2}{m}\right)$ distribution respectively  for $m=o(d^\frac{\delta}{2(2+\delta)})$. We know that if some random variable $X\sim \mathcal{N}(\mu,\sigma^2)$ distribution then for some constant $a$, $aX \sim \mathcal{N}(a\mu, a^2 \sigma^2)$. Thus for $m=o(d^\frac{\delta}{2(2+\delta)})$, $\sqrt{m}\cdot \cs(\p)_{l}\sim \mathcal{N}(0, \|\p\|^2)$ and $\sqrt{m} \cdot \cs(\q)_{l}\sim \mathcal{N}(0, \|\q\|^2)$ as $d \rightarrow \infty$.  Using this fact, it is easy to show that Equation~\eqref{eq:eq_datar_prob_cs} holds asymptotically true. Moreover, using Corollary~\ref{cor:indp_cs} and Equation~\eqref{eq:eq_datar_prob_cs}, we can easily show that left-hand side and right-hand side of Equation~\eqref{eq:eq201022} are equal. 
\end{proof}


From Equations~\eqref{eq:eq_datar_prob_cs} and \eqref{eq:eq201022}, it is clear that the probability of collision declines monotonically with $R = ||\p -\q||$. Therefore, due to  Definition~\ref{def:LSH}, our hash function stated in Definition~\ref{def:CS_Eucl_LSH}  is $(R_1,R_2, P_{1}^{m}, P_2^{m})$-sensitive for $P_1 = p(1)$, $P_2 = p(R)$ and $R_2/R_1 = R$.

\begin{remark} \label{rem:cse2lsh}
     We note that our hash function (Definition~\ref{def:CS_Eucl_LSH}) offers more space and time-efficient LSH for Euclidean distance as compared to  \ELSH ~\cite{datar2004locality}. To compute an $m$-sized hashcode of a $d$-dimensional input, the time and space complexity of \ELSH ~is $O(md)$, whereas our proposal's time and space complexity is $O(d)$. However,  our results hold when $m=o\left(d^\frac{\delta}{2(2+\delta)}\right)$ and $d \to \infty$. 

\end{remark}




\vspace{-0.1cm}
{\color{black}
\subsection{Improvement Using \HCS}

In this subsection, we explore the more space advantage possibility and suggest a new proposal. Our proposal improves \ELSH~\cite{datar2004locality} using a higher order count sketch.   At a high level, the idea is to consider the input vector $\p \in \R^d$ as a $N$ mode tensor with $d^{\frac{1}{N}}$ dimension along each mode. We then use a higher order count sketch~(HCS) to compress this tensor. Similar to \CSELSH, here also, the main challenge is to show that each element of the higher order count sketch vector follows an asymptotic normal distribution and is independent of each other.  


{\color{black}
However, in the case of HCS, we cannot apply Lyapunov's central limit theorems to prove asymptotic normality. We explain it as follows: let $\p \in \R^d$ and $\Tilde{\p} = vec(\hcsb{\p}) \in \R^m$. From Definition~\ref{def:hcs}, we have, $\widetilde{p_l} = \sum_{i_1 \in [d_1], \ldots, i_{N} \in [d_N]} Z_{jl} s_1(i_1) s_2(i_2) \cdots s_N(i_N) p_{j}$, where $Z_{jl}$ is an indicator random variable which takes value $1$ if $h_1(i_1) = l_1, \ldots, h_{N}(i_N) = l_N$ and $0$ otherwise, with $j= \sum_{k=2}^{N} \left(i_{k} \prod_{t=1}^{k-1} d_{t} \right) + i_1 \in [d], i_k \in [d_k]$, and $l= \sum_{k=2}^{N} \left(l_{k} \prod_{t=1}^{k-1} m_{t} \right) + l_1 \in [m]$. 
Let  $Y_j := Z_{jl} s_1(i_1) s_2(i_2) \cdots s_N(i_N) p_{j}$. So,  $\widetilde{p_l} = \sum_{j \in [d]} Y_j$ and we have to prove that  $\widetilde{p_l}$ follows asymptotic normal distribution for any $l \in [m]$. From the construction of $Y_j$ we can observe that $Y_j$'s are not pairwise independent for $j\in[d]$. Any two random variables $Y_{j}$ and $Y_{j'}$ such that $j \neq j'$ and have the same value for at least one index $i_{k}$  for $k \in [N]$ are not pairwise independent because they share the same sign function and hash function value for that index. So, we can't apply the Lyapunov central limit theorem to prove asymptotic normality for HCS. To address this, we visualize the set of random variables $Y_j$ as a graph, with nodes/vertices represented by $Y_j$ and an edge between nodes $Y_j$ and $Y_{j'}$, $j\neq j'$ and $j, j' \in [d]$, if they share the same value for at least one index $i_k$ for $k\in [N]$. In Theorem~\ref{thm:hcs_normality}, we show that entries of the resultant projected vector obtained from HCS follow the asymptotic univariate normal distribution. To prove it, we use the Theorem~\ref {thm:clt_grpah_var} based on the dependency graph by ~\cite{janson1988normal}. 

In Theorem~\ref{thm:hcs_normality_multi}, we show that the $m$-dimensional  projected vector obtained from HCS follows the asymptotic multivariate normal distribution. It is a multivariate extension of Theorem~\ref{thm:hcs_normality}.  To prove it, we use Theorem~\ref{thm:clt_grpah_vec}, which extends  Theorem~\ref{thm:clt_grpah_var} for dependent random vectors. In Corollary~\ref{cor:hcs_eucl_ind}, we show that entries of the projected vector are independent of each other. 
Finally, in Theorem~\ref{thm:HCS_Eucl_LSH_property}, we show that the hashcode generated using our proposal satisfies the property of LSH.  
We define our proposal for LSH for Euclidean distance based on the higher order count sketch as follows:
}

\begin{definition}[\HCSELSH]  \label{def:HCS_Eucl_LSH}
 Let $\p \in \R^d$ and $\hcs(\p) \in \R^{m_1 \times \cdots \times m_N}$ s.t. $m = \prod_{k=1}^{N} m_{k}$ be its higher order count sketch obtained using Equation~\eqref{eq:eqhcs}. We define a hash function $g' : \R^{d} \rightarrow \Z^{m}$ element-wise as follows:
\begin{align*}
    g'(\p)_{l} = \left\lfloor \frac{ \sqrt{m} \cdot \Vecc(\hcs(\p))_{l} + b}{w} \right\rfloor,  \numberthis \label{eq:eq_hca_eucl_lsh}
\end{align*}
where  $l \in [m]$, $b \in [0,w]$, $w>0$ and $\Vecc(\hcs(\p)) \in \R^{m}$.
\end{definition}

The following theorem states that each element of the higher order count sketch follows an asymptotic normal distribution.

\begin{thm} \label{thm:hcs_normality}
Let $\p \in \R^d$ and $\hcsb{\p}\in \R^{m_1 \times \cdots \times m_N}$ be its higher order count sketch. We define $\Tilde{\p} := \Vecc(\hcsb{\p}) \in \R^{m}$ where $ m = \prod_{k=1}^{N} m_{k}$ and $k \in [N]$. If $~\forall j \in [d],$  $0<\E\left[|p_j|^{2}\right]<\infty$, $\|\p\|_{\infty} \leq \infty$ and  $\sqrt{m} N^{\left(\frac{4}{5}\right)} = o\left(d^{\left(\frac{3N-8}{10N}\right)} \right)$  then as $d \to \infty$ 
\begin{equation}
\Tilde{p}_{l}
\overset{\mathcal{D}}{\to} \mathcal{N} \left(0,
\frac{||\p||^2}{m}
 \right) \quad \text{for} \quad l \in [m].
  \end{equation}
\end{thm}
\begin{proof}
Let $\p \in \R^d$ and $\hcsb{\p}\in \R^{m_1 \times \cdots \times m_N}$ be its higher order count sketch. Thus, from the definition of HCS (Definition~\ref{def:hcs}), we have 
\begin{align*} 
 \hcs(\p)_{l_1, \ldots, l_N}  &:= \sum_{h_1(i_1) = l_1, \ldots, h_{N}(i_N) = l_N} s_1(i_1) s_2(i_2) \cdots s_N(i_N) p_{j} \numberthis \label{eq:eq261022_1}
\end{align*}
where, $\hcs(\p) \in \R^{m_1 \times \cdots \times m_N}$, $ d = \prod_{k=1}^{N} d_{k}$ and $  j= \sum_{k=2}^{N} \left(i_{k} \prod_{t=1}^{k-1} d_{t} \right) + i_1$. Let $  m := \prod_{k=1}^{N} m_{k}$, $\Tilde{\p} := \Vecc(\hcsb{\p}) \in \R^{m}$. Let $Z_{jl}$ be an indicator random variable (where, $  j= \sum_{k=2}^{N} \left(i_{k} \prod_{t=1}^{k-1} d_{t} \right) + i_1 \in [d]$ and $  l= \sum_{k=2}^{N} \left(l_{k} \prod_{t=1}^{k-1} m_{t} \right) + l_1 \in [m]$), that takes value $1$ if $h_1(i_1) = l_1, \ldots, h_{N}(i_N) = l_N$, otherwise $0$. Therefore, we can write the components of $\Tilde{\p}$ as follows:
\begin{align*}
    \Tilde{p}_{l} &:= \sum_{i_1 \in [d_1], \ldots, i_{N} \in [d_N]} Z_{jl} s_1(i_1) s_2(i_2) \cdots s_N(i_N) p_{j}\quad \text{ for } l \in [m]. \numberthis \label{eq:eq261022_2} \\
\end{align*}
Let $Y_{j} := Z_{jl} s_1(i_1) s_2(i_2) \cdots s_N(i_N) p_{j}$. Therefore,
\begin{align*}
    \Tilde{p}_{l} &= \sum_{j \in [d]} Y_{j}. \numberthis \label{eq:eq261022_3}
\end{align*}
We compute the expected value and variance of $\Tilde{p}_{l}$ one by one.
\begin{align*}
    \E[\Tilde{p}_{l}] &=\sum_{j \in [d]}\E[Y_{j}]\\
    &= \sum_{j \in [d]} \E\left[Z_{jl} s_1(i_1) s_2(i_2) \cdots s_N(i_N) p_{j} \right]\\
    &= \sum_{j \in [d]} \E\left[Z_{jl}\right] \E[s_1(i_1)] \E[s_2(i_2)] \cdots \E[s_N(i_N)] p_{j}\\
    &= 0. \numberthis \label{eq:eq261022_4}
\end{align*}
\begin{align*}
   \Var(\Tilde{p}_{l}) &=  \Var \left(\sum_{j \in [d]}Y_{j} \right)= \sum_{j\in[d]} \Var(Y_{j})  + \sum_{j \neq j'; j,j' \in[d]} \Cov(Y_{j}, Y_{j'}). \numberthis \label{eq:eq261022_5}\\
\end{align*}
We compute each term of Equation~\eqref{eq:eq261022_5} one by one.
\begin{align*}
     \sum_{j\in[d]} \Var(Y_{j})  &=  \sum_{j\in[d]} \left(\E[Y_{j}^2] - \E[Y_{j}]^2\right) = \sum_{j\in[d]} \E[Y_{j}^2] - 0\\
    &= \sum_{j\in[d]}\E[Z_{jl}^2 s_1(i_1)^2 s_2(i_2)^2 \cdots s_N(i_N)^2 p_{j}^2]= \sum_{j\in[d]}\E[Z_{jl}] p_{j}^2\\
    & = \sum_{j\in[d]}\frac{1}{m}p_{j}^2 \quad \left[\because \E[Z_{jl}] = \frac{1}{m} \right]\\
    &= \frac{1}{m} \|\p\|^2.\numberthis \label{eq:eq261022_6} 
\end{align*}
\begin{align*}
  &\sum_{j \neq j'; j,j' \in[d]} \Cov(Y_{j}, Y_{j'})  =  \sum_{j \neq j'; j,j' \in[d]} \E[Y_{j}Y_{j'}] - \E[Y_{j}]\E[Y_{j'}]\\
  &\quad= \sum_{j \neq j'} \E \Big[Z_{jl} s_1(i_1) s_2(i_2) \cdots s_N(i_N) p_{j} Z_{j'l} s_1(i_1') s_2(i_2') \cdots s_N(i_N') p_{j'} \Big]  - 0\\
  &\quad = \sum_{j \neq j'} \E \Big[Z_{jl} Z_{j'l} s_1(i_1)s_1(i_1') s_2(i_2)s_2(i_2') s_N(i_N)s_N(i_N') p_{j} p_{j'} \Big]\\
  &\quad = \sum_{j \neq j'} \E[Z_{jl} Z_{j'l}\ \E[s_1(i_1)s_1(i_1')] \E[s_2(i_2)s_2(i_2')] \cdots \E[s_N(i_N)s_N(i_N')] p_{j} p_{j'}\\
  &\quad= 0. \numberthis \label{eq:eq261022_7}
\end{align*}
Equation~\eqref{eq:eq261022_7} holds true because $j \neq j'$ implies for at least one value of $k$ ($k\in[N]$), $i_k \neq i_k'$. Thus, $\E[s_{k}(i_k) s_{k}(i_k')] = \E[s_{k}(i_k)] \E[s_{k}(i_k')] = 0$ because $s_{k}(\cdot)$ is $2$-wise independent hash function. 
Let $\sigma_d^2 := \Var(\Tilde{p}_{l})$. From Equations~\eqref{eq:eq261022_5}, \eqref{eq:eq261022_6} and \eqref{eq:eq261022_7}, we have
\begin{align*}
    \sigma_d^2 = \Var(\Tilde{p}_{l}) &=  \frac{\|\p\|^2}{m}. \numberthis \label{eq:eq261022_8}
\end{align*}
To complete the proof, we need to prove that for some value of $\alpha$, the Equation~\ref{eq:eq_clt_graph_var} of Theorem \ref{thm:clt_grpah_var}. holds true. We recall it as follows:
\begin{align*}
    \left(\frac{d}{M}\right)^{\frac{1}{\alpha}} \frac{M A}{\sigma_{d}} \rightarrow 0 \quad \text{ as } \quad d \rightarrow \infty.
\end{align*}
where $\alpha$ is an integer, $A$ is the $\ell_{\infty}$ norm of vector $\p$, and $M$ denotes the maximum degree of the dependency graph generated by the random variables $Y_{j}$, $j\in[d]$ and is equal to $ \sum_{i=1}^{N} d_{i} - N$. For ease of calculation, in the following, we assume that $d_{1} = d_{2}= \ldots=d_{N} = d^{\frac{1}{N}}$.
Therefore, we have the following
\begin{align*}
 \left(\frac{d}{M}\right)^{\frac{1}{\alpha}} \frac{M A}{\sigma_{d}} 
 &= \left( \frac{d}{\sum_{i=1}^{N} d_{i} - N}\right)^{\frac{1}{\alpha}} \left(\sum_{i=1}^{N} d_{i} - N \right) \left(\frac{\|\p\|_{\infty}}{\|\p\|/\sqrt{m}}\right)\\
 &= \sqrt{m} \, d^{\frac{1}{\alpha}}\, \left(\sum_{i=1}^{N} d_{i} - N \right)^{1-\frac{1}{\alpha}} \, \frac{\|\p\|_{\infty}}{\|\p\|}\\
 &= \sqrt{m} \, d^{\frac{1}{\alpha}} \, \left(Nd^{\frac{1}{N}} - N \right)^{1-\frac{1}{\alpha}}  \, \frac{\|\p\|_{\infty}}{\|\p\|} \\
 &= \sqrt{m} \, d^{\frac{1}{\alpha}} \, \left(Nd^{\frac{1}{N}} - N \right)^{1-\frac{1}{\alpha}} \,  \frac{\|\p\|_{\infty}}{\sqrt{\sum_{j\in [d]} p_{j}^2}}\\
 &= \sqrt{m} \, d^{\frac{1}{\alpha}} \, \left(Nd^{\frac{1}{N}} - N \right)^{1-\frac{1}{\alpha}} \, \frac{\|\p\|_{\infty}}{\sqrt{d} \times \sqrt{\sum_{j\in [d]} \frac{p_{j}^2}{d}}}\\
 &= \frac{\sqrt{m} \, d^{\frac{1}{\alpha}} \, N^{1-\frac{1}{\alpha}} \left(d^{\frac{1}{N}} - 1 \right)^{1-\frac{1}{\alpha}}}{\sqrt{d} } \, \frac{\|\p\|_{\infty}}{\sqrt{\E[p_{j}^2]}}\\
 &\leq \frac{\sqrt{m} \, d^{\frac{1}{\alpha}} \, N^{1-\frac{1}{\alpha}} \left(d^{\frac{1}{N}} \right)^{1-\frac{1}{\alpha}}}{\sqrt{d} } \, \frac{\|\p\|_{\infty}}{\sqrt{\E[p_{j}^2]}}\\
 &= \frac{\sqrt{m} \,  N^{\left(1-\frac{1}{\alpha}\right)} \, d^{\left(\frac{N+\alpha-1}{\alpha N}\right)}}{\sqrt{d} } \, \frac{\|\p\|_{\infty}}{\sqrt{\E[p_{j}^2]}} \\
 &= \frac{\sqrt{m} \,  N^{\left(1-\frac{1}{\alpha}\right)}}{d^{\left(\frac{1}{2}- \frac{N+\alpha-1}{\alpha N}\right)} } \, \frac{\|\p\|_{\infty}}{\sqrt{\E[p_{j}^2]}}\\
 &= \frac{\sqrt{m} \,  N^{\left(1-\frac{1}{\alpha}\right)}}{d^{\left(\frac{\alpha N-2N-2\alpha +2}{2 \alpha N}\right)} } \, \frac{\|\p\|_{\infty}}{\sqrt{\E[p_{j}^2]}}\\
 &= \frac{\sqrt{m} \,  N^{\left(1-\frac{1}{\alpha}\right)}}{d^{\left(\frac{\alpha N-2N-2 \alpha +2}{2 \alpha N}\right)} } \, \frac{\|\p\|_{\infty}}{\sqrt{\E[p_{j}^2]}}\\
 &  \rightarrow 0 \text{ as } d \rightarrow \infty \text{ for } \alpha > \frac{2(N-1)}{(N-2)} \notag\\
 &\hspace{3cm}\text{ and }  \sqrt{m} \,  N^{\left(1-\frac{1}{\alpha}\right)} = o\left(d^{\left(\frac{\alpha N-2N-2 \alpha + 2}{2 \alpha N}\right)} \right). \numberthis \label{eq:eq261022_9}
\end{align*}
Equation~\eqref{eq:eq261022_9} hold true for $\alpha > \frac{2(N-1)}{(N-2)}$ \text{ and }  $\sqrt{m} \,  N^{\left(1-\frac{1}{\alpha}\right)} = o\left(d^{\left(\frac{ \alpha N-2N-2 \alpha + 2}{2 \alpha N}\right)} \right)$, provided $0 < \sqrt{\E[p_{j}^2]}< \infty$ and $\|\p\|_{\infty} < \infty$. Choosing $\alpha = 5$ in Equation~\eqref{eq:eq261022_9}, for $d \rightarrow \infty$, we have
\begin{align*}
     \left(\frac{d}{M}\right)^{\frac{1}{\alpha}} \frac{M A}{\sigma_{d}} & \rightarrow 0 \text{ for } \sqrt{m} N^{\left(\frac{4}{5}\right)} = o\left(d^{\left(\frac{3N-8}{10N}\right)} \right). \numberthis \label{eq:eq261022_10}
\end{align*}
Thus, from Theorem~\ref{thm:clt_grpah_var}, we have
\begin{align*}
    \frac{\Tilde{p}_{l} - \E[\Tilde{p}_{l}]}{\sigma_{d}} &\overset{\mathcal{D}}{\to} \mathcal{N}(0,1).\\
    \implies \frac{\Tilde{p}_{l} - 0}{\|\p\|/\sqrt{m}} &=\frac{\Tilde{p}_{l}}{\| \p \| /\sqrt{m}} \overset{\mathcal{D}}{\to} \mathcal{N}(0,1).\\
    \implies \Tilde{p}_{l} &\overset{\mathcal{D}}{\to} \mathcal{N}\left(0,\frac{\|\p\|^2}{m} \right). \numberthis \label{eq:eq261022_11}
\end{align*}
Equation~\eqref{eq:eq261022_11} completes a proof of the theorem.
\end{proof}

The following theorem is a is a multivariate extension of Theorem~\ref{thm:hcs_normality}. It states the multivariate asymptotic normality of the higher-order count sketch vector.
\begin{thm} \label{thm:hcs_normality_multi}
Let $\p \in \R^d$ and $\hcsb{\p}\in \R^{m_1 \times \cdots \times m_N}$ be its higher order count sketch. We define $\Tilde{\p} := \Vecc(\hcsb{\p}) \in \R^{m}$  s.t. $m = \prod_{k=1}^{N} m_{k}$ and $k \in [N]$. If $~\forall j \in [d],$  $0<\E\left[|p_j|^{2}\right]<\infty$, $\|\p\|_{\infty} \leq \infty$ and  $\sqrt{m} N^{\left(\frac{4}{5}\right)} = o\left(d^{\left(\frac{3N-8}{10N}\right)} \right)$  then as $d \to \infty$, we have
\begin{align}
\Tilde{\p}
\overset{\mathcal{D}}{\to} \mathcal{N} \left( \mathbf{0},
\boldsymbol{\Sigma_d}
 \right)
\end{align}
where $\mathbf{0}$ is zero vector and $\boldsymbol{\Sigma}_{d}$ is diagonal matrix having $\|{\p}\|^2/{m}$ as diagonal entries.
\end{thm}

\begin{proof}
Let $\p \in \R^d$ and $\hcsb{\p}\in \R^{m_1 \times \cdots \times m_N}$ be its higher order count sketch. From the definition of higher order count sketch (Definition~\ref{def:hcs}), we have
\begin{align*} 
  \hcs(\p)_{l_1, \ldots, l_N} &:= \sum_{h_1(i_1) = l_1, \ldots, h_{N}(i_N) = l_N} s_1(i_1) s_2(i_2) \cdots s_N(i_N) p_{j} \numberthis \label{eq:eq261022_51} \\
\end{align*}
where, $\hcs(\p) \in \R^{m_1 \times \cdots \times m_N}$, $ d = \prod_{k=1}^{N} d_{k}$ and $  j= \sum_{k=2}^{N} \left(i_{k} \prod_{t=1}^{k-1} d_{t} \right) + i_1$. Let $ m := \prod_{k=1}^{N} m_{k}$ and $\Tilde{\p} := \Vecc(\hcsb{\p}) \in \R^{m}$. Let $Z_{jl}$ be an indicator random variable (where, $  j= \sum_{k=2}^{N} \left(i_{k} \prod_{t=1}^{k-1} d_{t} \right) + i_1 \in [d]$ and $  l= \sum_{k=2}^{N} \left(l_{k} \prod_{t=1}^{k-1} m_{t} \right) + l_1 \in [m]$), which takes value $1$ if $h_1(i_1) = l_1, \ldots, h_{N}(i_N) = l_N$, otherwise takes value $0$. We can write the components of $\Tilde{\p}$ as follows:
\begin{align*}
    \Tilde{p}_{l} &:= \sum_{i_1 \in [d_1], \ldots, i_{N} \in [d_N]} Z_{jl} s_1(i_1) s_2(i_2) \cdots s_N(i_N) p_{j}. \numberthis \label{eq:eq261022_52} 
\end{align*}
Let 
\begin{align*}
    \mathbf{Y}_{j} &:= \begin{bmatrix}
                       Z_{j1} s_1(i_1) s_2(i_2) \cdots s_N(i_N) p_{j}\\
                       \vdots\\
                       Z_{jm} s_1(i_1) s_2(i_2) \cdots s_N(i_N) p_{j}
                      \end{bmatrix}. \\
\implies    \Tilde{\p} &= \sum_{j \in [d]} \mathbf{Y}_{j}. \numberthis
\end{align*}
We compute the following terms first
\begin{align*}
   \E \left[\mathbf{Y}_{j} \right] &= \begin{bmatrix}
         \E\left[Z_{j1} s_1(i_1) s_2(i_2) \cdots s_N(i_N) p_{j} \right]\\
         \vdots \\
         \E\left[Z_{jm} s_1(i_1) s_2(i_2) \cdots s_N(i_N) p_{j} \right]
    \end{bmatrix} \\
    &= \begin{bmatrix}
         \E\left[Z_{j1}\right] \E[s_1(i_1)] \E[s_2(i_2)] \cdots \E[s_N(i_N)] p_{j}\\
         \vdots \\
         \E\left[Z_{jm}\right] \E[s_1(i_1)] \E[s_2(i_2)] \cdots \E[s_N(i_N)] p_{j}
    \end{bmatrix}\\
    &= \mathbf{0}.  \numberthis \label{eq:eq261022_53}
\end{align*}
\begin{align*}
    \Cov(\mathbf{Y}_{j}, \mathbf{Y}_j) &= \begin{bmatrix}
         \E[Z_{j1}] p_{j}^2 & \cdots & 0\\
         \vdots & \ddots & \vdots\\
         0 & \cdots & \E[Z_{jm}] p_{j}^2
    \end{bmatrix} = \begin{bmatrix}
         \frac{p_{j}^2}{m} & \cdots & 0\\
         \vdots & \ddots & \vdots\\
         0 & \cdots & \frac{p_{j}^2}{m}
    \end{bmatrix}. \numberthis \label{eq:eq261022_54}\\
    \Cov\left( \mathbf{Y}_{j}, \mathbf{Y}_{j'}\right) &= \begin{bmatrix}
    0 & \cdots & 0 \\
    \vdots & \ddots & \vdots\\
    0 & \cdots & 0
    \end{bmatrix}.  \numberthis \label{eq:eq261022_55}
\end{align*}
Equations~\eqref{eq:eq261022_54} and \eqref{eq:eq261022_55} hold true due to the following.
\begin{align*}
    &\Cov\big(Z_{jl} s_1(i_1) s_2(i_2) \cdots s_N(i_N) p_{j}, \, Z_{jl'} s_1(i_1) s_2(i_2) \cdots s_N(i_N) p_{j}\big) \\
    &=  \E \Big[Z_{jl} s_1(i_1) s_2(i_2) \cdots s_N(i_N) p_{j}  Z_{jl'} s_1(i_1) s_2(i_2) \cdots s_N(i_N) p_{j} \Big] \notag\\
    &\qquad - \Big(\E[Z_{jl} s_1(i_1) s_2(i_2) \cdots s_N(i_N) p_{j}] \E[Z_{jl'} s_1(i_1) s_2(i_2) \cdots s_N(i_N) p_{j}] \Big)\\
    &= \E[Z_{jl}Z_{jl'}] p_{j}^2 - 0 \\
    & = \begin{cases}
    \frac{1}{m}p_j^2, & \text{ if } l = l'\\
    0, & \text{ otherwise},
    \end{cases}   \numberthis\\
    &\hspace{4.3cm}\left[\because Z_{jl} \text{ and }Z_{jl'} \text{ can't be $1$ simultaneously for $l \neq l'$} \right].
\end{align*}
\begin{align*}
    &\Cov \Big(Z_{jl} s_1(i_1) s_2(i_2) \cdots s_N(i_N) p_{j},  \, Z_{j'l'} s_1(i_1') s_2(i_2') \cdots s_N(i_N') p_{j'} \Big) \notag\\
    &\quad = \E \Big[Z_{jl} s_1(i_1) s_2(i_2) \cdots s_N(i_N) p_{j}, \, Z_{j'l'} s_1(i_1') s_2(i_2') \cdots s_N(i_N') p_{j'} \Big] \\
    &\qquad- \Big(\E[Z_{jl} s_1(i_1) s_2(i_2) \cdots s_N(i_N) p_{j}] \E[Z_{j'l'} s_1(i_1') s_2(i_2') \cdots s_N(i_N') p_{j'}] \Big)\\
    &\quad =\E[Z_{jl}Z_{j'l'}] \E[s_{1}(i_1) s_{1}(i_1')] \cdots \E[s_{N}(i_N) s_{N}(i_N')] p_{j} p_{j}' - 0= 0 \quad \text{ for } j \neq j'. \numberthis \label{eq:eq261022_56}
\end{align*}
Equation~\eqref{eq:eq261022_56} hold true because for $j \neq j'$ at least one $i_{k} \neq i_{k'}$ for $k \in [N]$ implies $\E[s_{k}(i_k) s_{k}(i_k')] = \E[s_{k}(i_k)] \E[s_{k}(i_k')]=0$ as $s_{k}(\cdot)$ is 2-wise independent hash function. From Equations~\eqref{eq:eq261022_53}, \eqref{eq:eq261022_54} and \eqref{eq:eq261022_55}, we have
\begin{align*}
    \E[\Tilde{\p}] &= \mathbf{0}. \numberthis \label{eq:eq261022_57}\\
   \boldsymbol{\Sigma}_{d} &:= \Cov(\Tilde{\p}) = \sum_{j \in [d]} \Cov(\mathbf{Y}_{j}, \mathbf{Y}_{j}) + \sum_{j \neq j'} \Cov(\mathbf{Y}_{j}, \mathbf{Y}_{j'})\\
    &= \begin{bmatrix}
         \frac{\|{\p}\|^2_2}{m} & \cdots & 0\\
         \vdots & \ddots & \vdots\\
         0 & \cdots & \frac{\|{\p}\|^2_2}{m}
    \end{bmatrix} + \begin{bmatrix}
    0 & \cdots & 0 \\
    \vdots & \ddots & \vdots\\
    0 & \cdots & 0
    \end{bmatrix} = \begin{bmatrix}
         \frac{\|{\p}\|^2}{m} & \cdots & 0\\
         \vdots & \ddots & \vdots\\
         0 & \cdots & \frac{\|{\p}\|^2}{m}
    \end{bmatrix} .\numberthis \label{eq:eq261022_58}\\
\end{align*}
For any unit vector $\a \in \R^m$, define a random variable $S_d := \a^T \Tilde{\p}$. The expected value and variance of $S_d$ are
\begin{align}
    \E[S_d] := \a^T \E[\Tilde{\p}] = 0 \numberthis \label{eq:eq051024_3}
\end{align}
and 
\begin{align}
    \sigma_d^2 =\Var(S_d):= \a^T \Cov(\Tilde{\p}) \a = \frac{\|p\|^2}{m} (a_1^2 + \cdots + a_m^2) = \frac{\|p\|^2}{m}.
\end{align}

To complete the proof, we need to prove that Equation~\eqref{eq:eq051024_0} of Theorem~\ref{thm:clt_grpah_vec} holds true for some value of $\alpha$. We recall it as follows:
\begin{align*}
    \left(\frac{d}{M}\right)^{\frac{1}{\alpha}} \frac{M A}{\sigma_d}  \rightarrow 0 \text{ as } d \rightarrow \infty
\end{align*}
where $\alpha$ is an integer, $A$ is an upper bound on $\Y_i$'s and is equal to the $\ell_{\infty}$ norm of vector $\p$, and $M$ denotes the maximum degree of the dependency graph generated by the random vectors $\mathbf{Y}_{j}$, $j\in[d]$ and is equal to $ \sum_{i=1}^{N} d_{i} - N$. For ease of analysis, we assume that $d_{1} = d_{2}= \ldots=d_{N} = d^{\frac{1}{N}}$. Now, we compute the following
\begin{align*}
 \left(\frac{d}{M}\right)^{\frac{1}{\alpha}} \frac{M A}{\sigma_d}    &= \left( \frac{d}{\sum_{i=1}^{N} d_{i} - N}\right)^{\frac{1}{\alpha}} \left(\sum_{i=1}^{N} d_{i} - N \right) \, \|\p\|_{\infty} \, \frac{\sqrt{m}}{\|\p\|}\\
 &= \frac{\sqrt{m} d^{\frac{1}{\alpha}}}{\left(\sum_{i=1}^{N} d_{i} - N \right)^{\frac{1}{\alpha}-1}} \, \frac{\|\p\|_{\infty}}{\|\p\|}\\
 &= \sqrt{m} \, d^{\frac{1}{\alpha}}\, \left(\sum_{i=1}^{N} d_{i} - N \right)^{1-\frac{1}{\alpha}} \, \frac{\|\p\|_{\infty}}{\|\p\|}\\
 &= \sqrt{m} \, d^{\frac{1}{\alpha}} \, \left(Nd^{\frac{1}{N}} - N \right)^{1-\frac{1}{\alpha}}  \, \frac{\|\p\|_{\infty}}{\|\p\|}\\ 
 &= \sqrt{m} \, d^{\frac{1}{\alpha}} \, \left(Nd^{\frac{1}{N}} - N \right)^{1-\frac{1}{\alpha}} \,  \frac{\|\p\|_{\infty}}{\sqrt{\sum_{j\in [d]} p_{j}^2}}\\
 &= \sqrt{m} \, d^{\frac{1}{\alpha}} \, \left(Nd^{\frac{1}{N}} - N \right)^{1-\frac{1}{\alpha}} \, \frac{\|\p\|_{\infty}}{\sqrt{d} \times \sqrt{\sum_{j\in [d]} \frac{p_{j}^2}{d}}}\\
 &= \frac{\sqrt{m} \, d^{\frac{1}{\alpha}} \, N^{1-\frac{1}{\alpha}} \left(d^{\frac{1}{N}} - 1 \right)^{1-\frac{1}{k}}}{\sqrt{d} } \, \frac{\|\p\|_{\infty}}{\sqrt{\E[p_{j}^2]}}\\
 &\leq \frac{\sqrt{m} \, d^{\frac{1}{\alpha}} \, N^{1-\frac{1}{\alpha}} \left(d^{\frac{1}{N}} \right)^{1-\frac{1}{\alpha}}}{\sqrt{d} } \, \frac{\|\p\|_{\infty}}{\sqrt{\E[p_{j}^2]}}\\
 &= \frac{\sqrt{m} \,  N^{\left(1-\frac{1}{\alpha}\right)} \, d^{\left(\frac{N+\alpha-1}{\alpha N}\right)}}{\sqrt{d} } \, \frac{\|\p\|_{\infty}}{\sqrt{\E[p_{j}^2]}}\\
 &= \frac{\sqrt{m} \,  N^{\left(1-\frac{1}{\alpha}\right)}}{d^{\left(\frac{1}{2}- \frac{N+\alpha-1}{\alpha N}\right)} } \, \frac{\|\p\|_{\infty}}{\sqrt{\E[p_{j}^2]}}\\
 &= \frac{\sqrt{m} \,  N^{\left(1-\frac{1}{\alpha}\right)}}{d^{\left(\frac{\alpha N-2N-2 \alpha +2}{2kN}\right)} } \, \frac{\|\p\|_{\infty}}{\sqrt{\E[p_{j}^2]}}\\
 &= \frac{\sqrt{m} \,  N^{\left(1-\frac{1}{\alpha}\right)}}{d^{\left(\frac{\alpha N-2N-2 \alpha +2}{2 \alpha N}\right)} } \, \frac{\|\p\|_{\infty}}{\sqrt{\E[p_{j}^2]}}\\
 &  \rightarrow 0 \text{ as } d \rightarrow \infty,  \numberthis 
 \label{eq:eq261022_60} \\
 & \quad\text{ for } \alpha > \frac{2(N-1)}{(N-2)} \text{ and }  \sqrt{m} \,  N^{\left(1-\frac{1}{\alpha}\right)} = o\left(d^{\left(\frac{ \alpha N-2N-2 \alpha +2}{2 \alpha N}\right)} \right).\\
\end{align*}
Equation~\eqref{eq:eq261022_60} hold true for $\alpha > \frac{2(N-1)}{(N-2)}$ \text{ and }  $\sqrt{m} \,  N^{\left(1-\frac{1}{\alpha}\right)} = o\left(d^{\left(\frac{\alpha N-2N-2 \alpha +2}{2 \alpha N}\right)} \right)$, provided $0 < \sqrt{\E[p_{j}^2]}< \infty$ and $\|\p\|_{\infty} < \infty$. For $\alpha = 5$ and $d \rightarrow \infty$, from Equation~\eqref{eq:eq261022_60}, we have
\begin{align*}
     \left(\frac{d}{M}\right)^{\frac{1}{\alpha }} \frac{M A}{\sigma_{d}} & \rightarrow 0 \text{ for } \sqrt{m} N^{\left(\frac{4}{5}\right)} = o\left(d^{\left(\frac{3N-8}{10N}\right)} \right). \numberthis \label{eq:eq261022_61}
\end{align*}
Thus, form Theorem~\ref{thm:clt_grpah_vec}, we have
\begin{align*}
    \Tilde{\p} &\overset{\mathcal{D}}{\to} \mathcal{N}\left(\mathbf{0}, \boldsymbol{\Sigma}_{d} \right). \numberthis \label{eq:eq261022_62}
\end{align*}
Equation~\eqref{eq:eq261022_62} completes a proof of the theorem.
\end{proof}

In the following corollary, building on the findings of Theorems~\ref{thm:hcs_normality} and~\ref{thm:hcs_normality_multi}, we show the distributional independence of the entries of the hashcode obtained via \HCSELSH. Its proof is along the same argument as that of Corollary~\ref{cor:indp_cs} and can be easily proven.
\begin{cor} \label{cor:hcs_eucl_ind}
Let $\p \in \R^d$ and $\hcsb{\p}\in \R^{m_1 \times \cdots \times m_N}$ be its higher order count sketch. Let $\Tilde{\p} = \Vecc(\hcsb{\p})\in \R^m$ where $m = \prod_{k=1}^{N} m_{k}$. If $~\forall j \in [d],$  $0<\E\left[|p_j|^{2}\right]<\infty$, $\|\p\|_{\infty} \leq \infty$ and  $\sqrt{m} N^{\left(\frac{4}{5}\right)} = o\left(d^{\left(\frac{3N-8}{10N}\right)} \right)$  then as $d \to \infty$, the elements $\Tilde{p}_1, \ldots, \Tilde{p}_{m}$ are independent.
\end{cor}

The subsequent theorem states that the \HCSELSH~ is a valid LSH. We can prove it using the results of Theorem~\ref{thm:hcs_normality} and Corollary~\ref{cor:hcs_eucl_ind}; and following similar proof steps to that of Theorem~\ref{thm:CS_Eucl_LSH_property}.
\begin{thm} \label{thm:HCS_Eucl_LSH_property}
For any $\p, \q \in \R^d$ such that $R := || \p -\q ||$ and  $\sqrt{m} N^{\left(\frac{4}{5}\right)} = o\left(d^{\left(\frac{3N-8}{10N}\right)} \right)$ then following holds true asymptotically as $d\to \infty$ 
\begin{align}
    &p(R) = Pr[g'(\p)_l = g'(\q)_l]  = \int_0^{w} \frac{1}{R} f\left(\frac{t}{R} \right)\left(1 - \frac{t}{w} \right), ~~ l \in [m], \text{ and, } \numberthis \label{eq:eq_datar_prob_hcs}\\
    &Pr[g'(\p)_1 = g'(\q)_1, \ldots, g'(\p)_m = g'(\q)_m] = p(R)^m.  \numberthis \label{eq:eq281022_a}
\end{align}
where $f(\cdot)$  denotes the density function of the absolute value of the standard normal distribution. 
\end{thm}


From Theorem~\ref{thm:HCS_Eucl_LSH_property}, it is clear that the probability of collision declines monotonically with $R = ||\p -\q||$. Hence,  according to Definition~\ref{def:LSH}, our proposal \HCSELSH~ is $(R_1,R_2, P_{1}^{m}, P_2^{m})$-sensitive  for $P_1 = p(1)$, $P_2 = p(R)$ and $R_2/R_1 = R$.

\begin{remark}
Compared to \ELSH, the evaluation of the proposed \HCSELSH ~(Definition~\ref{def:HCS_Eucl_LSH}) hash function requires a lesser number of operations as well as less space. The time and space complexity of \ELSH~(Definition~\ref{def:e2lsh}) to generate an $m$-sized hash code of a $d$-dimensional input vector is $O(md)$, 
whereas, the time complexity of our proposal \HCSELSH~is $O(d)$, and its space complexity is $O(N\sqrt[N]{d})  \ll O(d)$ when $N=o(d)$ (Remark~\ref{rem:hcs_space}).
\end{remark}
}
\section{Improving \SRP} \label{sec:cosine}

\subsection{Improvement Using \CS}
This subsection presents our LSH for cosine similarity, which improves over SRP. The main idea of our proposal is to compute the count-sketch of the input vector and then take the element-wise sign of the resulting vector. The main technical challenge is to show that after computing the count-sketch of input vector pairs, a tuple generated by corresponding entries of the resultant vectors follows the bivariate normal distribution and is independent of each other.  Our proof technique exploits the multivariate Lyapunov central limit theorem~\ref{thm:multivariate_Lyapunov_CLT} for this purpose. In  Theorem~\ref{thm:cs_bi_variate}, we show that the corresponding entries of the count-sketch vectors of input pairs follow an asymptotic bivariate normal distribution. Consequently,  in Theorem~\ref{thm:cs_multi_variate_srp}, we show that the concatenation of the $m$-dimensional sketch pairs obtained after applying count-sketch on input pairs follow the asymptotic multivariate normal distribution.  Building on the results of  Theorems~\ref{thm:cs_bi_variate} and \ref{thm:cs_multi_variate_srp},  we show the independence in Corollary~\ref{cor:indp_cs_srp}. Finally, in Theorem~\ref{thm:ubiased_cssrp}, we show that the hashcodes obtained using our proposal satisfy the property of LSH for cosine similarity.

\begin{definition}[\CSSRP] \label{def:cs_consine_lsh}
 We denote our proposal as a hash function $\xi(\cdot)$ that takes a vector $\mathbf{u} \in \R^d$ as input, first compute a count-sketch vector of $\uu$ (say $\csb{\mathbf{u}}$) using Definition~\ref{def:def_cs},  and
then compute the sign of each component of the compressed vector.  We define it as follows:
\begin{align}
\xi(\uu) &= (\sgn(\cs(\uu)_{1}), \ldots, \sgn(\cs(\uu)_{m}) )  
\end{align}
where, $\sgn(\cs(\uu)_{l}) = 1$ if $\cs(\uu)_{l} > 0$, otherwise $\sgn(\cs(\uu)_{l}) = 0$, for $l \in [m]$.
\end{definition}

The following theorem states that if we compress two arbitrary vectors using count-sketch and concatenate their sketch values for an arbitrary index, the resulting vector follows the asymptotic bi-variate normal distribution.

\begin{thm} \label{thm:cs_bi_variate}
Let $\uu$, $\vv \in \R^d$ and $\csb{\uu}$, $\csb{\vv} \in \R^{m}$ be their corresponding count sketch vectors. If $~\forall j \in [d],$  $\E\left[(u_j^2+v_j^2)^{\frac{2+\delta}{2}}\right]$ is finite and  $m = o(d)$ for some $\delta>0$,  then as $d \to \infty$
\begin{equation}
\begin{bmatrix}
\cs(\uu)_{l}  \\
\cs(\vv)_{l}
\end{bmatrix}
\overset{\mathcal{D}}{\to} \mathcal{N} \left( \mathbf{0},
\V_d
 \right) \quad \text{for} \quad l \in [m],
  \end{equation}
where  $\V_d=
\frac{1}{m}\begin{bmatrix}
||\uu||^2 &  \langle \uu, \vv \rangle\\
\langle \uu, \vv \rangle & ||\vv||^2
\end{bmatrix}$ and $\mathbf{0}$ is a $2$-dimensional vector with each entry as  zero.
\end{thm}

\begin{proof}
Let $\uu$, $\vv \in \R^d$ and $\csb{\uu}$, $\csb{\vv} \in \R^{m}$ be their corresponding count sketch vectors. From the definition of count-sketch (Definition~\ref{def:def_cs}), we have:
\begin{align*} 
  \cs(\uu)_{l} &:= \sum_{h(j) = l} s(j) u_{j} = \sum_{j \in [d]} K(j,l) s(j) u_j, \numberthis \label{eq:eqcs00021}\\
    \cs(\vv)_{l} &:= \sum_{h_1(i_1) = l} s(j) v_{j} = \sum_{j \in [d]} K(j,l) s(j) v_j \numberthis  \label{eq:eqcs00022} 
\end{align*}
where $l \in [m]$ and $K(j,l)$ is an indicator random variable that takes value $1$ if $h(j) = l$, otherwise takes value $0$. Define a sequence of  random vectors $\{\X_{j}\}_{j=1}^{d}$, where
\begin{align}
    \X_{j} &:= \begin{bmatrix}
    K(j,l) s(j) u_j.\\
     K(j,l) s(j) v_j
    \end{bmatrix}.
\end{align}
The expected value of $\X_{j}$
\begin{align*}
    \E[\X_{j}] &= \begin{bmatrix}
    \E[  K(j,l) s(j) u_j] \\
    \E[ K(j,l) s(j) v_j]
    \end{bmatrix} = \begin{bmatrix}
    0 \\
    0
    \end{bmatrix} = \mathbf{0}. \numberthis \label{eq:eqcd00025}
\end{align*}
Equation~\eqref{eq:eqcd00025} holds because $\E[s(j)] = 0$. The covariance of $\X_j$ is
\begin{align*}
&\Cov(\X_{j})\notag\\
&~= \begin{bmatrix}
    \Cov( K(j,l) s(j) u_j,K(j,l) s(j) u_j ) &  \Cov( K(j,l) s(j) u_j, K(j,l) s(j) v_j)\\
    \Cov( K(j,l) s(j) u_j, K(j,l) s(j) v_j) & \Cov(K(j,l) s(j) v_j,K(j,l) s(j) v_j)
    \end{bmatrix}  \\
    & ~= \frac{1}{m}\begin{bmatrix}
     u_{j}^2 &  u_{j}v_{j}\\
      u_{j}v_{j} &   u_{j}^2
    \end{bmatrix}. \numberthis \label{eq:eqcs00026}
\end{align*}
Equation~\eqref{eq:eqcs00026} holds true due to the following:
\begin{align*}
\E \left[K(j,l) s(j) u_j \right]  = 0  = \E \left[K(j,l) \right] \E[s(j)]  v_{j}.
\end{align*}
\begin{align*}
\Cov( K(j,l) s(j) u_j,K(j,l) s(j) u_j ) &= \E \left[ K(j,l)^2 s(j)^2 u_j^2\right] - \E \left[K(j,l) s(j) u_j\right]^2\\
&= E \left[ K(j,l)\right] u_{j}^2 - 0 = \frac{1}{m} u_{j}^2. 
\end{align*}
Similarly, 
\begin{align*}
    \Cov( K(j,l) s(j) v_j,K(j,l) s(j) v_j ) &= \frac{1}{m} v_{j}^2.
\end{align*}
\begin{align*}
   \Cov( K(j,l) s(j) u_j, K(j,l) s(j) v_j)  &= \E[ K(j,l)^2 s(j)^2 u_j v_j] \notag\\
   &\hspace{2cm}- \E[K(j,l) s(j) u_j] \E[K(j,l) s(j) v_j]\\
    &= \E[K(j,l)]u_{j} v_{j} - 0 = \frac{1}{m} u_{j} v_{j}.
\end{align*}
To compute the expression
\begin{align*}
    \E[||\X_{j} - \E[\X_{j}]||^{2+\delta}] &= \E[||\X_{j}||^{2+\delta}]. \numberthis \label{eq:eqcs00030}
\end{align*}
First, we compute the following:
\begin{align*}
    ||\X_{j}||^2 &=  K(j,l)^2 s(j)^2 u_{j}^2 +  K(j,l)^2 s(j)^2 v_{j}^2
    \\
    &=  K(j,l) u_{j}^2 +  K(j,l)v_{j}^2\\
    &= (u_{j}^2 + v_{j}^2)  K(j,l)
\end{align*}
implies
\begin{align*}
    ||\X_{j}||^{2 + \delta}  = (u_{j}^2 + v_{j}^2)^{\frac{2 + \delta}{2}}  K(j,l). \numberthis \label{eq:eqcs00031}
\end{align*}
From Equations~\eqref{eq:eqcd00025}, \eqref{eq:eqcs00030} and \eqref{eq:eqcs00031}, we have
\begin{align*}
    \E\left[||\X_{j} - \E[\X_j]||^{2+\delta} \right] &= \E\left[||\X_{j}||^{2+\delta} \right]\\
    &= (u_{j}^2 + v_{j}^2)^{\frac{2 + \delta}{2}} \E[ K(j,l)] = \frac{1}{m} (u_{j}^2 + v_{j}^2)^{\frac{2 + \delta}{2}}. \numberthis \label{eq:eqcs00032}
\end{align*}
Let $\V_d := \sum_{j=1}^{d} \Cov(\X_{j})$. Thus from Equation~\eqref{eq:eqcs00026}, we have
\begin{align*}
    \V_d &    = \frac{1}{m}\begin{bmatrix}
      \sum_{j=1}^{d} u_{j}^2 & \sum_{j=1}^{d} u_{j}v_{j}\\
      \sum_{j=1}^{d} u_{j}v_{j} & \sum_{j=1}^{d} u_{j}^2
    \end{bmatrix} = \frac{1}{m}\begin{bmatrix}
  ||\uu||^2 & \langle\uu,\vv \rangle \\
  \langle \uu,\vv \rangle&||\vv||^2
\end{bmatrix}.
\end{align*}

Note that the matrix $\V_{d}$ is symmetric positive definite matrix. Thus $\V_{d}^{-1}$ is also symmetric positive definite and is
 \begin{align*}
 \V_{d}^{-1} & = \frac{m}{(||\uu||^2 ||\vv||^2 - \langle \uu,\vv\rangle^2)}\begin{bmatrix}
     ||\vv||^2 & -\langle \uu,\vv\rangle\\
     -\langle \uu,\vv\rangle & ||\uu||^2
    \end{bmatrix}.
 \end{align*}
Due to the facts that for any square matrix $\mathbf{A}$, $||\mathbf{A}||_{F}^2 = \Tr \left(\mathbf{AA}^{T} \right)$ and for any positive definite matrix $\mathbf{B}$ there exists a matrix $\mathbf{C}$ such that $\mathbf{B}=\mathbf{CC}$, we have
\begin{align*}
	||\V_{d}^{-1/2}||^2 &= \Tr \left( \V_{d}^{-1/2} ( \V_{d}^{-1/2} )^{T}\right) \nonumber\\
                      &= \Tr \left( \V_{d}^{-1/2} \V_{d}^{-1/2} \right) \qquad \left[\because \V_{d}^{-1/2} \text{ is symmetric}\right] \nonumber\\
					   &= \Tr(\V_{d}^{-1}) \nonumber\\
					   &= \frac{m (||\uu||^2 + ||\vv||^2)}{||\uu||^2 ||\vv||^2 - \langle \uu,\vv\rangle^2}.\numberthis \label{eq:eqncs30}
\end{align*}
Here,  $\Tr$ denotes a trace operator. To complete the proof, we need to show that Equation~\eqref{eq:multivariate_lyapunov} of Theorem~\ref{thm:multivariate_Lyapunov_CLT} in the main paper holds true, \textit{i.e.} 
$$\lim_{d \to \infty} \left \|{\V_d}^{-\frac{1}{2}} \right \|^{2+\delta} \E\left[||\X_{j} - \E[\X_j]||^{2+\delta} \right] =  \lim_{d \to \infty}\left \|{\V_d}^{-\frac{1}{2}} \right \|^{2+\delta} \sum_{j=1}^{d} \E[||\X_{j}||^{2+\delta}] = 0.$$ 
From Equations \eqref{eq:eqncs30} and~\eqref{eq:eqcs00032}, we have 
\begin{align*}
    &||{\V_d}^{-\frac{1}{2}}||^{2+\delta} \sum_{j=1}^{d} \E \left[||\X_{j}||^{2+\delta} \right]\nonumber\\
    & =  \left(\frac{ m(||\uu||^2 + ||\vv||^2) }{||\uu||^2 ||\vv||^2-\langle\uu, \vv\rangle^2}\right)^{\frac{2+\delta}{2}} \cdot \frac{1}{m} \cdot \sum_{i=1}^{d}(u_{j}^2 + v_{j}^2 )^{\frac{2+\delta}{2}}\\ 
    &=   (m)^{\delta/2} \cdot\left(\frac{ ||\uu||^2 +||\vv||^2}{||\uu||^2 ||\vv||^2-\langle\uu, \vv\rangle^2}\right)^{\frac{2+\delta}{2}} \cdot \sum_{j=1}^{d}(u_{j}^2 + v_{j}^2)^{\frac{2+\delta}{2}}\\
    &=\left(\frac{m}{d}\right)^{\delta/2}  \cdot  \, \left(\frac{ \sum_{j=1}^{d}\frac{u_{i}^2}{d} + \sum_{j=1}^{d} \frac{v_{j}^2}{d} }{\sum_{j=1}^{d}\frac{u_{j}^2}{d} \sum_{j=1}^{d}\frac{v_{j}^2}{d} - \left(\sum_{j=1}^{d} \frac{u_{j}v_{j}}{d}\right)^2}\right)^{\frac{2+\delta}{2}} \, \cdot \sum_{j=1}^{d} \left(\frac{(u_{j}^2 + v_{j}^2)^{\frac{2+\delta}{2}}}{d}  \right) \\
    &= \left(\frac{m}{d}\right)^{\delta/2} \, \cdot \left(\frac{ \E[u_{j}^2] + \E[v_{j}^2] }{\E[u_{j}^2]\E[v_{j}^2] - \left(\E[u_{j}v_{j}]\right)^2}\right)^{\frac{2+\delta}{2}} \, \cdot \E\left[(u_{i}^2 + v_{i}^2)^{\frac{2+\delta}{2}}\right]\\
    &  {\to 0 ~~~~~ as~~~~~ d \to \infty}. \numberthis \label{eq:eqncs35}
\end{align*}
Equation~\eqref{eq:eqncs35} holds  as $d$ tends to infinity due to the fact that    $ m=o\left(d\right)$ and $\E\left[(u_{j}^2 + v_{j}^2)^{\frac{2+\delta}{2}}\right]<\infty $. Thus, due to Theorem~\ref{thm:multivariate_Lyapunov_CLT}, we have 
\begin{align*}
& \V_{d}^{-1/2}
\sum_{j=1}^d \X_{j}
\overset{}{\to} \mathcal{N} \left( \mathbf{0}, \mathbf{I}
 \right). \\
  \implies
  &\V_{d}^{-1/2}\begin{bmatrix}
\cs(\uu)_{l} \\
\cs(\vv)_{l}
\end{bmatrix}
\overset{\mathcal{D}}{\to}  \mathcal{N} \left( \mathbf{0}, \mathbf{I}
 \right).\\
 \implies
  &\begin{bmatrix}
\cs(\uu)_{l} \\
\cs(\vv)_{l}
\end{bmatrix}
\overset{\mathcal{D}}{\to}  \mathcal{N} \left( \mathbf{0}, \V_d
 \right). \numberthis \label{eq:eqcs600}
\end{align*}
Equation~\eqref{eq:eqcs600} completes a proof of the theorem.
\end{proof}

{\color{black}The following theorem states the vector obtained after concatenation of the corresponding count-sketches of the input vectors follows an asymptotic multivariate normal distribution.

}
\begin{thm} \label{thm:cs_multi_variate_srp}
Let $\uu$, $\vv \in \R^d$ and $\csb{\uu}$, $\csb{\vv} \in \R^{m}$ be their corresponding count-sketch vectors. If $~\forall j \in [d],$  $\E\left[(u_j^2+v_j^2)^{\frac{2+\delta}{2}}\right]$ is finite and  $m=o\left(d^{\frac{\delta}{2(2 + \delta)}}\right)$ for some $\delta>0$,  then 
\begin{equation}
\begin{bmatrix}
\csb{\uu}  \\
\csb{\vv}
\end{bmatrix}
\overset{\mathcal{D}}{\to} \mathcal{N} \left( \mathbf{0},
\V_d 
 \right), \text{ as } d \rightarrow \infty
  \end{equation}
where   $\V_d=
\frac{1}{m}\begin{bmatrix}
  ||\uu||^2 & \cdots & 0 & \langle\uu,\vv \rangle  & \cdots & 0\\
  \vdots & \ddots & \vdots & \vdots & \ddots & \vdots \\
  0 & \cdots & ||\uu||^2  & 0 & \cdots &  \langle\uu,\vv \rangle \\
  \langle\uu,\vv \rangle  & \cdots & 0 & ||\vv||^2 & \cdots & 0 \\
  \vdots & \ddots & \vdots & \vdots & \ddots & \vdots \\
  0 & \cdots &  \langle\uu,\vv \rangle & 0 & \cdots & ||\vv||^2
\end{bmatrix} $ and $\mathbf{0} \in \R^{2m}$ is a zero vector.
\end{thm}

\begin{proof}
Let $\uu$, $\vv \in \R^d$ and $\csb{\uu}$, $\csb{\vv} \in \R^{m}$ be their corresponding count-sketch vectors. From the definition of count-sketch (Definition~\ref{def:def_cs}), we  have:
\begin{align*} 
  \cs(\uu)_{l} &:= \sum_{h(j) = l} s(j) u_{j} = \sum_{j \in [d]} K(j,l) s(j) u_j. \numberthis \label{eq:eqcs00021_new}\\
    \cs(\vv)_{l} &:= \sum_{h_1(i_1) = l} s(j) v_{j} = \sum_{j \in [d]} K(j,l) s(j) v_j. \numberthis \label{eq:eqcs00022_new} 
\end{align*}
where $l \in [m]$ and $K(j,l)$ is an indicator random variable that takes value $1$ if $h(j) = l$, otherwise takes value $0$. We define a sequence of  random vectors $\{\X_{j}\}_{j=1}^{d}$ as follows:
\begin{align}
    \X_{j} &:= \begin{bmatrix}
    K(j,1) s(j) u_j\\
    \vdots\\
     K(j,m) s(j) u_j\\
     K(j,1) s(j) v_j\\
     \vdots\\
     K(j,m) s(j) v_j
    \end{bmatrix}.
\end{align}
The expected value of $\X_{j}$ is
\begin{align*}
    \E[\X_{j}] &= \begin{bmatrix}
    \E[  K(j,1) s(j) u_j] \\
    \vdots\\
    \E[  K(j,m) s(j) u_j]\\
    \E[ K(j,1) s(j) v_j]\\
    \vdots\\
    \E[ K(j,m) s(j) v_j]
    \end{bmatrix} = \begin{bmatrix}
    0 \\
    \vdots\\
    0\\
    0\\
    \vdots\\
    0
    \end{bmatrix} = \mathbf{0}. \numberthis \label{eq:eqcd00025_new}\\
\end{align*}
Equation~\eqref{eq:eqcd00025_new} holds because $\E[s(j)] = 0$. The covariance of $\X_j$ is 
\begin{align*}
\Cov(\X_{j})
    & = \frac{1}{m}\begin{bmatrix}
     u_{j}^2 & \cdots & 0 &  u_{j}v_{j}& \cdots& 0\\
     \vdots & \ddots & \vdots & \vdots & \ddots& \vdots\\
     0 & \cdots & u_{j}^2 & 0 & \cdots & u_{j}v_{j}\\
     u_{j}v_{j} & \cdots & 0 & v_{j}^2 & \cdots & 0\\
     \vdots & \ddots & \vdots & \vdots & \ddots & \vdots\\
      0 & \cdots & u_{j}v_{j} & 0 & \cdots & v_{j}^2 
    \end{bmatrix}. \numberthis \label{eq:eqcs00026_new}
\end{align*}
Equation~\eqref{eq:eqcs00026_new} holds true due to the followings:
\begin{align*}
\E \left[K(j,l) s(j) u_j \right]  = 0 =  \E \left[K(j,l) \right] \E[s(j)]  v_{j}.
\end{align*}
\begin{align*}
\Cov( K(j,l) s(j) u_j,K(j,l) s(j) u_j ) &= \E \left[ K(j,l)^2 s(j)^2 u_j^2\right] - \E \left[K(j,l) s(j) u_j\right]^2\\
&= E \left[ K(j,l)\right] u_{j}^2 - 0 = \frac{1}{m} u_{j}^2.  
\end{align*}
The above equation holds because $\E[K(j,l) = 1/m$ for  $l \in [m]$. Similarly, 
\begin{align*}
    \Cov( K(j,l) s(j) v_j,K(j,l) s(j) v_j ) &= \frac{1}{m} v_{j}^2.
\end{align*}
\begin{align*}
   \Cov( K(j,l) s(j) u_j, K(j,l) s(j) v_j)  &= \E[ K(j,l)^2 s(j)^2 u_j v_j] \notag\\
   & \qquad \qquad - \E[K(j,l) s(j) u_j] \E[K(j,l) s(j) v_j]\\
    &= \E[K(j,l)]u_{j} v_{j} - 0 = \frac{1}{m} u_{j} v_{j}.
\end{align*}
To compute the expression
\begin{align*}
    \E[||\X_{j} - \E[\X_{j}]||^{2+\delta}] &= \E[||\X_{j}||^{2+\delta}] \numberthis \label{eq:eqcs00030_new}
\end{align*}
first, we compute the following:
\begin{align*}
    ||\X_{j}||^2 &=  \sum_{l=1}^m K(j,l) s(j)^2 (u_j^2 + v_j)^2
    = (u_{j}^2 + v_{j}^2)  \sum_{l=1}^{m}K(j,l) .\\
    \implies ||\X_{j}||^{2 + \delta} & = (u_{j}^2 + v_{j}^2)^{\frac{2 + \delta}{2}} \quad \left[\because  \sum_{k=1}^{m} K(j,l) = 1 \right]. \numberthis \label{eq:eqcs00031_new}
\end{align*}
From Equations~\eqref{eq:eqcs00030_new} and \eqref{eq:eqcs00031_new}, we have
\begin{align*}
   \E\left[||\X_{j} - \E[\X_j]||^{2+\delta} \right] &=  \E\left[||\X_{j}||^{2+\delta} \right] = \E \left[(u_{j}^2 + v_{j}^2)^{\frac{2 + \delta}{2}} \right] = (u_{j}^2 + v_{j}^2)^{\frac{2 + \delta}{2}}. \numberthis \label{eq:eqcs00032_new}
\end{align*}
In the following, we define and compute $\V_{d}$ using Equation~\eqref{eq:eqcs00026_new}:
\begin{align*}
    \V_d &:= \sum_{j=1}^{d} \Cov(\X_{j}).\\
    & = \frac{1}{m}\begin{bmatrix}
      \sum_{j=1}^{d} u_{j}^2 & \cdots & 0 & \sum_{j=1}^{d} u_{j}v_{j} & \cdots & 0\\
      \vdots & \ddots & \vdots & \vdots & \ddots & \vdots \\
       0 & \cdots & \sum_{j=1}^{d} u_{j}^2 & 0 & \cdots & \sum_{j=1}^{d} u_{j}v_{j}\\
      \sum_{j=1}^{d} u_{j}v_{j} & \cdots & 0 & \sum_{j=1}^{d} v_{j}^2 & \cdots & 0 \\
      \vdots & \ddots & \vdots & \vdots & \ddots & \vdots  \\
      0 & \cdots & \sum_{j=1}^{d} u_{j}v_{j} & 0 & \cdots & \sum_{j=1}^{d} v_{j}^2 \\
    \end{bmatrix}.\\
    &= \frac{1}{m}\begin{bmatrix}
  ||\uu||^2 & \cdots & 0 & \langle\uu,\vv \rangle  & \cdots & 0\\
  \vdots & \ddots & \vdots & \vdots & \ddots & \vdots \\
  0 & \cdots & ||\uu||^2  & 0 & \cdots &  \langle\uu,\vv \rangle \\
  \langle\uu,\vv \rangle  & \cdots & 0 & ||\vv||^2 & \cdots & 0 \\
  \vdots & \ddots & \vdots & \vdots & \ddots & \vdots \\
  0 & \cdots &  \langle\uu,\vv \rangle & 0 & \cdots & ||\vv||^2
\end{bmatrix}.
\end{align*}

Note that the matrix $\V_{d}$ is symmetric positive definite matrix. Thus $\V_{d}^{-1}$ is also symmetric positive definite and is
 \begin{align*}
 \V_{d}^{-1} & = \frac{m}{||\uu||^2 ||\vv||^2 - \langle \uu,\vv\rangle^2} \hspace{-0.2cm} \begin{bmatrix}
     ||\vv||^2 & \cdots & 0 & -\langle \uu,\vv\rangle & \cdots & 0\\
     \vdots & \ddots & \vdots & \vdots & \ddots & \vdots \\
     0 & \cdots & ||\vv||^2 & 0 & \cdots & -\langle \uu,\vv\rangle \\
     -\langle \uu,\vv\rangle & \cdots & 0 & ||\uu||^2 & \cdots & 0 \\
     \vdots & \ddots & \vdots & \vdots & \ddots & \vdots \\
     0 & \cdots & -\langle \uu,\vv\rangle &  0 & \cdots & ||\uu||^2
    \end{bmatrix}.
 \end{align*}
Due to the facts that for any square matrix $\mathbf{A}$, $||\mathbf{A}||_{F}^2 = \Tr \left(\mathbf{AA}^{T} \right)$ and for any positive definite matrix $\mathbf{B}$ there exists a matrix $\mathbf{C}$ such that $\mathbf{B}=\mathbf{CC}$, we have
\begin{align*}
	||\V_{d}^{-1/2}||^2 &= \Tr \left( \V_{d}^{-1/2} ( \V_{d}^{-1/2} )^{T}\right) \nonumber\\
                      &= \Tr \left( \V_{d}^{-1/2} \V_{d}^{-1/2} \right) \qquad \left[\because \V_{d}^{-1/2} \text{ is symmetric}\right] \nonumber\\
					   &= \Tr(\V_{d}^{-1}) \nonumber\\
					   &= \frac{m^2 (||\uu||^2 + ||\vv||^2)}{||\uu||^2 ||\vv||^2 - \langle \uu,\vv\rangle^2}.\numberthis \label{eq:eqncs30_new}
\end{align*}
Here,  $\Tr$ denotes a trace operator. To complete the proof, we need to show that the Equation~\eqref{eq:multivariate_lyapunov} of Theorem~\ref{thm:multivariate_Lyapunov_CLT} holds true, \textit{i.e.}
$$\lim_{d \to \infty} \left\|{\V_d}^{-\frac{1}{2}} \right\|^{2+\delta} \sum_{j=1}^{d} \E[||\X_{j}||^{2+\delta}] = 0.$$ 
From Equations~\eqref{eq:eqncs30_new} and~\eqref{eq:eqcs00032_new}, we have
\begin{align*}
    & \left \|{\V_d}^{-\frac{1}{2}} \right \|^{2+\delta} \sum_{j=1}^{d} \E \left[||\X_{j}||^{2+\delta} \right]\nonumber\\
    & ~=  \left(\frac{ m^2(||\uu||^2 + ||\vv||^2) }{||\uu||^2 ||\vv||^2-\langle\uu, \vv\rangle^2}\right)^{\frac{2+\delta}{2}} \cdot \sum_{i=1}^{d}(u_{j}^2 + v_{j}^2 )^{\frac{2+\delta}{2}} \\
    &~=   (m)^{2+\delta} \cdot\left(\frac{ ||\uu||^2 +||\vv||^2}{||\uu||^2 ||\vv||^2-\langle\uu, \vv\rangle^2}\right)^{\frac{2+\delta}{2}} \cdot \sum_{j=1}^{d}(u_{j}^2 + v_{j}^2)^{\frac{2+\delta}{2}}\\
    &~=\left(\frac{m}{d^{\frac{\delta}{2(2 + \delta)}}}\right)^{2 + \delta}  \hspace{-0.3cm} \cdot  \, \left(\frac{ \sum_{j=1}^{d}\frac{u_{i}^2}{d} + \sum_{j=1}^{d} \frac{v_{j}^2}{d} }{\sum_{j=1}^{d}\frac{u_{j}^2}{d} \sum_{j=1}^{d}\frac{v_{j}^2}{d} - \left(\sum_{j=1}^{d} \frac{u_{j}v_{j}}{d}\right)^2}\right)^{\frac{2+\delta}{2}} \, \hspace{-0.35cm} \cdot \sum_{j=1}^{d} \left(\frac{(u_{j}^2 + v_{j}^2)^{\frac{2+\delta}{2}}}{d}  \right) \\
    &~= \left(\frac{m}{d^{\frac{\delta}{2(2 + \delta)}}}\right)^{2 + \delta} \, \cdot \left(\frac{ \E[u_{j}^2] + \E[v_{j}^2] }{\E[u_{j}^2]\E[v_{j}^2] - \left(\E[u_{j}v_{j}]\right)^2}\right)^{\frac{2+\delta}{2}} \,  \cdot \E\left[(u_{i}^2 + v_{i}^2)^{\frac{2+\delta}{2}}\right]\\
    & ~ {\to 0 ~~~~~ as~~~~~ d \to \infty}. \numberthis\label{eq:eqncs35_new}
\end{align*}
 {Equation~\eqref{eq:eqncs35_new} holds  as $d$ become large due to the fact that    $ m=o\left(d^{\frac{\delta}{2(2 + \delta)}}\right)$ and $\E\left[(u_{j}^2 + v_{j}^2)^{\frac{2+\delta}{2}}\right]<\infty $. }
Thus, due to Theorem~\ref{thm:multivariate_Lyapunov_CLT}, we have 
\begin{align*}
& \V_{d}^{-1/2}
\sum_{j=1}^d \X_{j}
\overset{}{\to} \mathcal{N} \left( \mathbf{0}, \mathbf{I}
 \right). \\
 \implies
  &\V_{d}^{-1/2}\begin{bmatrix}
\csb{\uu} \\
\csb{\vv}
\end{bmatrix}
\overset{\mathcal{D}}{\to}  \mathcal{N} \left( \mathbf{0}, \mathbf{I}
 \right) \implies \begin{bmatrix}
\csb{\uu} \\
\csb{\vv}
\end{bmatrix}
\overset{\mathcal{D}}{\to}  \mathcal{N} \left( \mathbf{0}, \V_d
 \right). \numberthis\label{eq:eqcs600_new}
\end{align*}
Equation~\eqref{eq:eqcs600_new} completes a proof of the theorem.
\end{proof}

 Building on the result of  Theorems~\ref{thm:cs_bi_variate} and \ref{thm:cs_multi_variate_srp}, in the following corollary, we prove the independence of the concatenation of the corresponding elements of $\csb{\uu}$ and $\csb{\vv}$, and show that the joint probability distribution of $[\csb{\uu}, \csb{\vv}]$ is equal to the product of the distribution of tuples $[\cs(\uu)_l,\cs(\vv)_{l}]$ for $l \in [m]$. 
 
\begin{cor} \label{cor:indp_cs_srp}
Let $\uu, \vv \in \R^{d}$ and $\csb{\uu}, \csb{\vv} \in \R^{m}$ be the corresponding count-sketch vector obtained using Equation~\eqref{eq:eq_cs}. If  $\forall ~ j \in [d]$, $\E \left[(u_{j} + v_{j})^{\frac{2 + \delta}{2}} \right]$ is finite and $m = o \left(d^{\frac{\delta}{2(2+\delta)}}\right)$ for some $\delta > 0$, then as $d \rightarrow \infty$, the tuples $\begin{bmatrix} \cs(\uu)_1\\  \cs(\vv)_{1} \end{bmatrix}, \ldots, \begin{bmatrix} \cs(\uu)_m\\  \cs(\vv)_{m} \end{bmatrix}$ are independent.
\end{cor}
\begin{proof}
 From Theorem~\ref{thm:cs_multi_variate_srp}, we have 
 \begin{align*}
     \begin{bmatrix} \csb{\uu}\\
     \csb{\vv}
     \end{bmatrix} \sim \mathcal{N}(\mathbf{0},\V_{d})
 \end{align*}
where $\mathbf{0} \in \R^{2m}$  and  
    $ \V_{d} = \frac{1}{m}\begin{bmatrix}
          ||\uu||^2 & \cdots & 0 & \langle \uu, \vv \rangle & \cdots & 0\\
          \vdots & \ddots & \vdots & \vdots & \ddots & \vdots\\
          0 & \cdots & ||\uu||^2 &  0 & \cdots & \langle \uu, \vv \rangle  \\
          \langle \uu, \vv \rangle & \cdots & 0 & ||\vv||^2 & \cdots & 0 \\
          \vdots & \ddots & \vdots & \vdots & \ddots & \vdots \\
          0 & \cdots & \langle \uu, \vv \rangle &  0 & \cdots & ||\vv||^2
     \end{bmatrix}$.
 Hence, 
 \begin{align*}
     &\Pr\left(
                \begin{bmatrix}
                    \cs(\uu)_{1}\\
                    \vdots\\
                    \cs(\uu)_{m}\\
                    \cs(\vv)_{1}\\
                    \vdots \\
                    \cs(\vv)_{m}
                \end{bmatrix}  
                = \begin{bmatrix}
                    \Tilde{u}_{1}\\
                    \vdots \\
                    \Tilde{u}_{m}\\
                    \Tilde{v}_{1}\\
                    \vdots \\
                    \Tilde{v}_{m}
                \end{bmatrix} \right) \notag\\
     &= \frac{1}{\sqrt{(2\pi)^m det(\V_{d})}} \cdot \exp{\left(-\frac{1}{2} \begin{bmatrix}
                    \Tilde{u}_{1} & \hdots & \Tilde{u}_{m} & \Tilde{v}_{1} & \hdots & \Tilde{v}_{m}
                \end{bmatrix}  \V_d^{-1} \begin{bmatrix}
                    \Tilde{u}_{1}\\
                    \vdots \\
                    \Tilde{u}_{m}\\
                    \Tilde{v}_{1}\\
                    \vdots \\
                    \Tilde{v}_{m}
                \end{bmatrix} \right)}\\
     &= \frac{1}{\sqrt{(2\pi)^m det(\V_{d})}} \cdot \exp{}\Bigg(-\frac{m}{2(||\uu||^2 ||\vv||^2 - \langle \uu,\vv\rangle^2)} \notag\\   
     &\begin{bmatrix}
                    \Tilde{u}_{1} & \hdots & \Tilde{u}_{m} & \Tilde{v}_{1} & \hdots & \Tilde{v}_{m}
                \end{bmatrix}   \hspace{-0.15cm}
                \begin{bmatrix}
                     ||\vv||^2 & \cdots & 0 & -\langle \uu,\vv\rangle & \cdots & 0\\
                     \vdots & \ddots & \vdots & \vdots & \ddots & \vdots \\
                     0 & \cdots & ||\vv||^2 & 0 & \cdots & -\langle \uu,\vv\rangle \\
                     -\langle \uu,\vv\rangle & \cdots & 0 & ||\uu||^2 & \cdots & 0 \\
                     \vdots & \ddots & \vdots & \vdots & \ddots & \vdots \\
                     0 & \cdots & -\langle \uu,\vv\rangle &  0 & \cdots & ||\uu||^2
                \end{bmatrix} \hspace{-0.1cm}
                \begin{bmatrix}
                    \Tilde{u}_{1}\\
                    \vdots \\
                    \Tilde{u}_{m}\\
                    \Tilde{v}_{1}\\
                    \vdots \\
                    \Tilde{v}_{m}
                \end{bmatrix}  \Bigg)\\
    &= \frac{1}{\sqrt{(2\pi)^m \left(\frac{m}{2(||\uu||^2 ||\vv||^2 - \langle \uu,\vv\rangle^2)}\right)^m}} \exp{\left(-\frac{m \left(\sum_{l=1}^{m} \left(\Tilde{u}_{l}^2 ||\vv||^2  +  \Tilde{v}_{l}^2 ||\uu||^2 - 2\Tilde{u}_{1}\Tilde{v}_{1} \langle \uu,\vv\rangle \right)  \right)}{2(||\uu||^2 ||\vv||^2 - \langle \uu,\vv\rangle^2)} \right)}\\
    &= \frac{1}{\sqrt{2\pi \, \left(\frac{m}{(||\uu||^2 ||\vv||^2 - \langle \uu,\vv\rangle^2)}\right)}}  \exp{\left(-\frac{m\left(\Tilde{u}_{1}^2 ||\vv||^2  +  \Tilde{v}_{1}^2 ||\uu||^2 - 2\Tilde{u}_{1}\Tilde{v}_{1} \langle \uu,\vv\rangle \right)  }{2(||\uu||^2 ||\vv||^2 - \langle \uu,\vv\rangle^2)} \right)}\times \notag\\
    & \cdots \times \frac{1}{\sqrt{2\pi \, \left(\frac{m}{2(||\uu||^2 ||\vv||^2 - \langle \uu,\vv\rangle^2)}\right)^m}} \exp{\left(-\frac{m\left(\Tilde{u}_{m}^2 ||\vv||^2  +  \Tilde{v}_{m}^2 ||\uu||^2 - 2\Tilde{u}_{m}\Tilde{v}_{m} \langle \uu,\vv\rangle \right) }{2(||\uu||^2 ||\vv||^2 - \langle \uu,\vv\rangle^2)} \right)}\\
    &= \Pr\left(\begin{bmatrix}
    \cs(\uu)_{1}\\
    \cs(\vv)_{1}
    \end{bmatrix} 
    = \begin{bmatrix}
        \Tilde{u}_{1}\\
        \Tilde{v}_{1}
    \end{bmatrix} \right)
    \times \cdots \times 
    \Pr\left(\begin{bmatrix}
    \cs(\uu)_{m}\\
    \cs(\vv)_{m}
    \end{bmatrix} 
    = \begin{bmatrix}
        \Tilde{u}_{m}\\
        \Tilde{v}_{m}
    \end{bmatrix}
     \right). \numberthis \label{eq:eq_280922_1_new}
 \end{align*}
 
Equation~\eqref{eq:eq_280922_1_new} implies the distributional independence of $\begin{bmatrix} \cs(\uu)_1\\  \cs(\vv)_{1} \end{bmatrix}, \ldots,$ $ \begin{bmatrix} \cs(\uu)_m\\  \cs(\vv)_{m} \end{bmatrix}$ and completes a proof of the theorem.
\end{proof}

Finally, building on the results mentioned in Theorem~\ref{thm:cs_bi_variate} and Corollary~\ref{cor:indp_cs_srp}, we show that our proposal mentioned in Definition~\ref{def:cs_consine_lsh} gives an LSH for cosine similarity. 

\begin{thm} \label{thm:ubiased_cssrp}
Let $\xi(\uu)$, $\xi(\vv)$ be $m$-dimensional binary vector of $\uu,\vv\in \R^d$ respectively, obtained using our proposal (Definition~\ref{def:cs_consine_lsh}). Suppose $\theta_{(\uu,\vv)}=\cos^{-1}\left(\frac{\uu^T\vv}{\|\uu\| \|\vv\|}\right)$ denotes the angular similarity between $\uu$ and $\vv$. Then for $d \rightarrow \infty$  and 
 $m=o\left(d^\frac{\delta}{2(2+\delta)}\right)$,
    $\Pr( \xi(\uu)_{l} = \xi(\vv)_{l} )  = 1-\frac{\theta_{(\uu,\vv)}}{\pi}$ \text{ and }
    $\Pr[\xi(\uu)_{1} = \xi(\vv)_{1}, \ldots, \xi(\uu)_{m} = \xi(\vv)_{m}] = \left(1-\frac{\theta_{(\uu,\vv)}}{\pi}\right)^m$.
\end{thm}
\begin{proof} 
 From Definition~\ref{def:cs_consine_lsh}, for $l \in [m]$, we have
 \begin{align*}
     \xi(\uu)_{l} = \sgn(\cs(\uu)_{l}) \quad \text{ and } \quad  \xi(\vv)_{l} = \sgn(\cs(\vv)_{l}).
 \end{align*}
 From Theorem~\ref{thm:cs_bi_variate}, for any $l\in[m]$, $\begin{bmatrix}
 \cs(\uu)_{l} & \cs(\vv)_{l}
 \end{bmatrix}^{T}$ follows asymptotic bi-variate normal distribution for $m=o\left(d\right)$ implies is bivariate normal for $m = o(d^\frac{\delta}{2(2+\delta)})$ as well. Hence,  a trivial consequence of  the asymptotic normality yields
 \begin{align*}
      \Pr[\xi(\uu)_{l} = \xi(\vv)_{l}]  \rightarrow  1-\frac{\theta_{(\uu,\vv)}}{\pi} \quad \left[\text{see Lemma 6 of Li \etal~\cite{li2006very}}\right]. 
      \numberthis \label{eq:eq8888}
 \end{align*}
 Using independence result stated in Corollary~\ref{cor:indp_cs_srp} and Equation~\eqref{eq:eq8888}, for $d \rightarrow \infty$ and $m =o(d^\frac{2}{2(2+\delta)})$, we have
 \begin{align*} 
   \Pr[\xi(\uu)_{1} = \xi(\vv)_{1}, \ldots, \xi(\uu)_{m} = \xi(\vv)_{m}] = \left(1-\frac{\theta_{(\uu,\vv)}}{\pi}\right)^m.
\end{align*}
\end{proof}

Let $S= \frac{\uu^T\vv}{\|\uu\| \|\vv\|}$ denotes the cosine similarity between $\uu$ and $\vv$. From Theorem~\ref{thm:ubiased_cssrp}, it is evident that the probability of a collision decreases monotonically with $S$. Hence,  due to Definition~\ref{def:LSH}, our proposal \CSSRP ~defined in Definition~\ref{def:cs_consine_lsh} is $(R_1,R_2, P_{1}^{m}, P_2^{m})$-sensitive for $R_1 = S$, $R_2 = cS$, $P_{1} = (1-\cos^{-1}(S)/\pi)$  and $P_{2} = (1-\cos^{-1}(cS)/\pi)$.

 \begin{remark} \label{rem:CSSRP}
    We note that our proposal \CSSRP~offers a more space and time efficient LSH for cosine similarity compared to  \SRP ~\cite{charikar2002similarity}. To compute an $m$-sized hashcode of a $d$-dimensional input, the time and space complexity of \SRP ~is $O(md)$, whereas our proposal's time and space complexity is $O(d)$. However,  our results hold when $m=o\left(d^\frac{\delta}{2(2+\delta)}\right)$ and $d$ is large. 
\end{remark}
 
{\color{black}
\subsection{Improvement Using \HCS} 

 In this subsection, we provide another solution that gives further advantage on space complexity while retaining the speedup obtained via $\CSSRP$. The core idea of our approach is first to consider the input vector $\mathbf{p}\in\R^d$ as a mode $N$ tensor with $d^{\frac{1}{N}}$ dimensional along each mode; followed by compressing it via higher order count sketch;  consequently taking the elementwise sign of the resulting sketch tensor, and finally reshaping the tensor into a vector. Similar to \CSSRP~the  main technical challenge is to demonstrate that after computing higher order count sketch of input vector pairs, entries of the resultant sketches follow the bivariate normal distribution and simultaneously are independent of each other. We define our proposal for LSH for cosine similarity based on a higher-order count sketch as follows:
\begin{definition}[HCS-SRP] \label{def:hcs_consine_lsh}
 We denote our proposal as a hash function $\xi'(\cdot)$ that takes a vector $\uu \in \R^d$ as input, first compute a higher order count sketch of $\uu$ (say $\hcsb{\uu}$) using Definition~\ref{def:hcs}, flatten it to a vector form (say $\Tilde{\uu}:= \Vecc(\hcsb{\uu}) \in \R^m$), and then compute the sign of each component of the compressed vector.  We define it as follows:
\begin{align}
\xi'(\uu) &= (\sgn(\Tilde{u}_1), \ldots, \sgn(\Tilde{u}_{m}) )  
\end{align}
where, $\sgn(\Tilde{u}_{l}) = 1$ if $\Tilde{u}_{l} > 0$, otherwise $\sgn(\Tilde{u}_{l}) = 0$, for $l \in [m]$.
\end{definition}

The following theorem demonstrates that if we compress two arbitrary input vectors using a higher order count sketch and concatenate their corresponding sketch values for an arbitrary index, the resultant vector follows the asymptotic bi-variate normal distribution.

\begin{thm} \label{thm:hcs_bi_variate}
Let $\uu$, $\vv \in \R^d$ and $\hcsb{\uu}$, $\hcsb{\vv} \in \R^{m_1 \times \cdots \times m_N}$ be their corresponding higher order count sketches. We define $\Tilde{\uu} := \Vecc(\hcs(\uu)) \in \R^{m}$ and $\Tilde{\vv} := \Vecc(\hcs(\vv)) \in \R^{m}$ s.t. $  m = \prod_{k=1}^{N} m_{k}$ and  $  l= \sum_{k=2}^{N} \left(l_{k} \prod_{t=1}^{k-1} m_{t} \right) + l_1 \in [m]$ for any  $l_{k} \in [m_{k}]$ and $k \in [N]$. If $~\forall j \in [d]$, $  {\E[u_{j}^2]}$, $\E[v_{j}^2]$, $\E[u_jv_j]$ is finite and  $\sqrt{m} N^{\left(\frac{4}{5}\right)} = o\left(d^{\left(\frac{3N-8}{10N}\right)} \right)$,  then as $d \to \infty$,   we have

\begin{equation}
\begin{bmatrix}
\Tilde{u}_{l}  \\
\Tilde{v}_{l}
\end{bmatrix}
\overset{\mathcal{D}}{\to} \mathcal{N} \left( \mathbf{0},
\Sig_d
 \right), \qquad l \in [m],
  \end{equation}
 where  $\Sig_d=
\frac{1}{m}\begin{bmatrix}
||\uu||^2 &  \langle \uu, \vv \rangle\\
\langle \uu, \vv \rangle & ||\vv||^2
\end{bmatrix}$, $\mathbf{0}$ is a $2$-dimensional vector with each entry as  zero.
\end{thm}
\begin{proof}
Let $\uu$, $\vv \in \R^d$ and $\hcsb{\uu}$, $\hcsb{\vv} \in \R^{m_1 \times \cdots \times m_N}$ be their corresponding higher order count sketches. From the definition of higher order count sketch (Definition~\ref{def:hcs}), we have:
\begin{align*} 
  \hcs(\uu)_{l_1, \ldots, l_N} &:= \sum_{h_1(i_1) = l_1, \ldots, h_{N}(i_N) = l_N} s_1(i_1) s_2(i_2) \cdots s_N(i_N) u_{j},  \numberthis \label{eq:eq281022_1} \\
    \hcs(\vv)_{l_1, \ldots, l_N} &:= \sum_{h_1(i_1) = l_1, \ldots, h_{N}(i_N) = l_N} s_1(i_1) s_2(i_2) \cdots s_N(i_N) v_{j} \numberthis \label{eq:eq281022_2} 
\end{align*}
where, for $k \in [N]$, $l_k \in [m_k]$, $i_k \in [d_k]$ and 
$  j= \sum_{k=2}^{N} \left(i_{k} \prod_{t=1}^{k-1} d_{t} \right) + i_1$. 
Let $ m := \prod_{k=1}^{N} m_{k}$, $\Tilde{\uu} := \Vecc(\hcsb{\uu}) \in \R^{m}$, $\Tilde{\vv} := \Vecc(\hcsb{\vv}) \in \R^{m}$. Let $Z_{jl}$ be an indicator random variable (where, $ j= \sum_{k=2}^{N} \left(i_{k} \prod_{t=1}^{k-1} d_{t} \right) + i_1$ and $  l= \sum_{k=2}^{N} \left(l_{k} \prod_{t=1}^{k-1} m_{t} \right) + l_1$), which takes value $1$ if $h_1(i_1) = l_1, \ldots, h_{N}(i_N) = l_N$, otherwise takes value $0$. For $l \in [m]$, we can write the $l$-th components of sketch vectors $\Tilde{\uu}$ and $\Tilde{\vv}$ as follows:
\begin{align*}
    \Tilde{u}_{l} &:= \sum_{i_1 \in [d_1], \ldots, i_{N} \in [d_N]} Z_{jl} s_1(i_1) s_2(i_2) \cdots s_N(i_N) u_{j}, \numberthis \label{eq:eq281022_3} \\
     \Tilde{v}_{l} &:= \sum_{i_1 \in [d_1], \ldots, i_{N} \in [d_N]} Z_{jl} s_1(i_1) s_2(i_2) \cdots s_N(i_N) v_{j}. \numberthis\label{eq:eq281022_4} 
\end{align*}

Let $\hat{u}_{jl} = Z_{jl} s_1(i_1) s_2(i_2) \cdots s_N(i_N) u_{j}$  and $\hat{v}_{jl} = Z_{jl} s_1(i_1) s_2(i_2) \cdots s_N(i_N) v_{j}$, where, $  j= \sum_{k=2}^{N} \left(i_{k} \prod_{t=1}^{k-1} d_{t} \right) + i_1 \in [d]$ and $l \in [m]$ . Define a sequence of  random vectors $\{\Y_{j}\}_{j=1}^{d}$ as follows:
\begin{align*}
    \Y_{j} &:= \begin{bmatrix}
    \hat{u}_{jl}\\
    \hat{v}_{jl}
    \end{bmatrix}.\\
    \implies \begin{bmatrix}
    \Tilde{u}_{l}\\
    \Tilde{v}_{l}
    \end{bmatrix} &= \sum_{j \in [d]} \Y_{j}. \numberthis \label{eq:eq281022_5}
\end{align*}
The expected value of $\begin{bmatrix}
    \Tilde{u}_{l} &
    \Tilde{v}_{l}
    \end{bmatrix}^T$ is
\begin{align*}
    \E\begin{bmatrix}
    \Tilde{u}_{l}\\
    \Tilde{v}_{l}
    \end{bmatrix}  &= \sum_{j \in[d]} \E\left[ \Y_{j} \right] = \sum_{j \in[d]} \begin{bmatrix}
    \E[ \hat{u}_{jl}] \\
    \E[\hat{v}_{jl}]
    \end{bmatrix} = \begin{bmatrix}
    0 \\
    0
    \end{bmatrix} \numberthis \label{eq:eq281022_6} = \mathbf{0}.
\end{align*}
Equation~\eqref{eq:eq281022_6} holds because $\E[s_1(i_1)\cdots s_N(i_N)] = \E[s_1(i_1)] \cdots  \E[s_N(i_N)] = 0$. 
The covariance  of  $\begin{bmatrix}
    \Tilde{u}_{l} &
    \Tilde{v}_{l}
    \end{bmatrix}^T$
    is 
\begin{align*}
  \Cov\left(\begin{bmatrix}
    \Tilde{u}_{l} \\
    \Tilde{v}_{l}
    \end{bmatrix} \right) &= \Cov \left(\sum_{j \in [d]}\Y_{j} \right) = \sum_{j \in [d]} \Cov(\Y_{j}, \Y_{j}) + \sum_{j \neq j'} \Cov(\Y_{j}, \Y_{j'}). \numberthis \label{eq:eq281022_7}
\end{align*}
We compute the terms of Equation~\eqref{eq:eq281022_7} one by one and are
\begin{align*}
    \Cov(\Y_j, \Y_j) &= \begin{bmatrix}
     \Cov(\hat{u}_{jl},\hat{u}_{jl}) &  \Cov(\hat{u}_{jl}, \hat{v}_{jl})\\
    \Cov(\hat{u}_{jl}, \hat{v}_{jl}) & \Cov(\hat{v}_{jl},\hat{v}_{jl})
    \end{bmatrix}\\
    &=  \begin{bmatrix}
     \E\left[ Z_{jl}\right] u_{j}^2 & \E\left[ Z_{jl}\right] u_{j}v_{j}\\
      \E\left[ Z_{jl}\right] u_{j}v_{j} &  \E\left[ Z_{jl}\right] v_{j}^2
    \end{bmatrix} =  \begin{bmatrix}
     \frac{ u_{j}^2}{m} &\frac{u_{j}v_{j}}{m}\\
      \frac{u_{j}v_{j}}{m} & \frac{v_{j}^2}{m}
    \end{bmatrix}. \numberthis \label{eq:eq281022_8}\\
    \Cov(\Y_j, \Y_j') &=\begin{bmatrix}
     \Cov(\hat{u}_{jl},\hat{u}_{j'l}) &  \Cov(\hat{u}_{jl}, \hat{v}_{j'l})\\
    \Cov(\hat{v}_{jl}, \hat{u}_{j'l}) & \Cov(\hat{v}_{jl},\hat{v}_{j'l})
    \end{bmatrix} =  \begin{bmatrix}
     0 & 0\\
     0 & 0
    \end{bmatrix}. \numberthis \label{eq:eq281022_9}
\end{align*}
Equation~\eqref{eq:eq281022_9} holds true due to the following.
\begin{align*}
  \Cov(\hat{u}_{jl},\hat{u}_{j'l}) &=  \E[Z_{jl} Z_{j'l}] \E[s_{1}(i_1)s_{1}(i_1')]\cdots\E[s_{N}(i_N)s_{N}(i_N')]u_j u_j' -0 = 0.  \numberthis \label{eq:eq281022_10}\\
  \Cov(\hat{u}_{jl}, \hat{v}_{j'l})&=  \E[Z_{jl} Z_{j'l}] \E[s_{1}(i_1)s_{1}(i_1')]\cdots\E[s_{N}(i_N)s_{N}(i_N')]u_j v_j' - 0=0.
  \numberthis \label{eq:eq281022_11}
 \end{align*}
Equations~\eqref{eq:eq281022_10} and \eqref{eq:eq281022_11} hold true because $j \neq j'$  implies  for at least one $k \in [N], i_{k} \neq i_k'$ which leads to $\E[s_{k}(i_{k}) s_{k}(i_{k}')] = \E[ s_{k}(i_{k})] \E[s_{k}(i_{k}')] = 0$. Similarly,
 \begin{align*}
     \Cov(\hat{v}_{jl},\hat{v}_{j'l}) & = 0 \quad \text{ and } \quad 
     \Cov(\hat{v}_{jl},\hat{u}_{j'l})  = 0.
 \end{align*}
 Let 
 \begin{align*}
     \boldsymbol{\Sigma}_d &:= \Cov\left( \sum_{j \in [d]} \Y_j\right) =\Cov\left(\begin{bmatrix}
    \Tilde{u}_{l} \\
    \Tilde{v}_{l}
    \end{bmatrix} \right).
 \end{align*}
Form Equations~\eqref{eq:eq281022_7}, \eqref{eq:eq281022_8} and \eqref{eq:eq281022_9}, we have
\begin{align*}
    \boldsymbol{\Sigma}_d &= \begin{bmatrix}
     \frac{ \| \uu \|^2}{m} &\frac{\langle \uu, \vv \rangle}{m}\\
      \frac{\langle \uu, \vv \rangle}{m} & \frac{ \| \vv \|^2}{m}
    \end{bmatrix}. \numberthis \label{eq:eq281022_14}
\end{align*}
For any unit vector $\a \in \R^2$, define a random variable 
$$S_d := \sum_{j\in[d]}\a^T \Y_j = \a^T\begin{bmatrix}
    \Tilde{u}_{l} \\
    \Tilde{v}_{l}
    \end{bmatrix}. $$
The expected value and variance of $S_d$ are
\begin{align}
    \E[S_d] := \a^T \E\begin{bmatrix}
    \Tilde{u}_{l} \\
    \Tilde{v}_{l}
    \end{bmatrix} = 0, \numberthis \label{eq:eq101024_1}
\end{align}
and 
{\color{black}
\begin{align}
    \sigma_d^2 =\Var(S_d):= \a^T \Cov\left(\begin{bmatrix}
    \Tilde{u}_{l} \\
    \Tilde{v}_{l}
    \end{bmatrix} \right) \a = \frac{1}{m}\left(\|\uu\|^2 a_1^2 + \|\vv\|^2 a_2^2 +2 \langle \uu,\vv\rangle a_1a_2 \right).
\end{align}
}

To complete the proof, we need to prove that Equation~\eqref{eq:eq051024_0} of Theorem~\ref{thm:clt_grpah_vec} holds true for some value of $\alpha$. We recall it as follows:
\begin{align*}
    \left(\frac{d}{M}\right)^{\frac{1}{\alpha}} \frac{M A}{\sigma_d}  \rightarrow 0 \text{ as } d \rightarrow \infty
\end{align*}
where $k$ is an integer, $A$ is an upper bound on the norm of $Y_j$'s and is a sum of the $\ell_{\infty}$ of $\uu$ and $\vv$, i.e., $A = \|\uu\|_{\infty}+ \|\vv\|_{\infty} $ and $M$ denotes the maximum degree of the dependency graph generated by the random vectors $\mathbf{Y}_{j}$, $j\in[d]$ and equal to $\sum_{i=1}^{N} d_{i} - N$. For ease of analysis, we assume that $d_{1} =\ldots=d_{N} = d^{\frac{1}{N}}$ and compute the following
\begin{align*}
 &\left(\frac{d}{M}\right)^{\frac{1}{k}} \frac{M A}{\sigma_d} = \left( \frac{d}{\sum_{i=1}^{N} d_{i} - N}\right)^{\frac{1}{k}} \cdot \left(\sum_{i=1}^{N} d_{i} - N \right) \cdot (\|\uu\|_{\infty}+ \|\vv\|_{\infty}) \notag\\
 &\hspace{7cm} \cdot \frac{\sqrt{m} }{\sqrt{||\uu||^2 a_{1}^2 + ||\vv||^2 a_2^2 + 2\langle \uu,\vv\rangle a_1a_2}}\\
 &~= \sqrt{m} \, d^{\frac{1}{k}}\, \left(\sum_{i=1}^{N} d_{i} - N \right)^{1-\frac{1}{k}} \cdot\left(\sum_{i=1}^{N} d_{i} - N \right) \cdot (\|\uu\|_{\infty}+ \|\vv\|_{\infty}) \notag\\
 &\hspace{7cm} \cdot \frac{1 }{\sqrt{||\uu||^2 a_{1}^2 + ||\vv||^2 a_2^2 + 2\langle \uu,\vv\rangle a_1a_2}}\\
 &~= \sqrt{m} \, d^{\frac{1}{k}} \, \left(Nd^{\frac{1}{N}} - N \right)^{1-\frac{1}{k}} \cdot (\|\uu\|_{\infty} + \|\vv\|_{\infty}) \notag\\
 &\hspace{3cm} \cdot \frac{1}{\sqrt{||\uu||^2 a_{1}^2 + ||\vv||^2 a_2^2 + 2\langle \uu,\vv\rangle a_1a_2}},~\left[ \because d_{1} =  \ldots=d_{N} = d^{\frac{1}{N}} \right]\\
 &~= \sqrt{m} \, d^{\frac{1}{k}} \, \left(N d^{\frac{1}{N}} - N \right)^{1-\frac{1}{k}} \cdot  (\|\uu\|_{\infty} + \|\vv\|_{\infty}) \notag\\
 &\hspace{5cm}\cdot \frac{1}{\sqrt{a_1^2\sum_{j\in [d]} u_j^2  + a_2^2\sum_{j\in [d]} v_j^2 + 2a_1a_2 \sum_{j\in [d]} u_jv_j}}\\
 &~= \frac{\sqrt{m} \, d^{\frac{1}{k}} \, \left(N d^{\frac{1}{N}} - N \right)^{1-\frac{1}{k}}}{\sqrt{d}} \cdot  (\|\uu\|_{\infty} + \|\vv\|_{\infty}) \notag\\
 &\hspace{5cm}\cdot \frac{1}{\sqrt{a_1^2\sum_{j\in [d]} \frac{u_j^2}{d}  + a_2^2\sum_{j\in [d]} \frac{v_j^2}{d} + 2a_1a_2 \sum_{j\in [d]} \frac{u_jv_j}{d}}}\\
 &~= \frac{\sqrt{m} \, d^{\frac{1}{k}} \, \left(N d^{\frac{1}{N}} - N \right)^{1-\frac{1}{k}}}{\sqrt{d}} \cdot  (\|\uu\|_{\infty} + \|\vv\|_{\infty}) \notag\\
 & \hspace{7cm} \cdot \frac{1}{\sqrt{a_1^2 \E[u_j^2]  + a_2^2 \E[v_j^2] + 2a_1a_2 \E[u_jv_j]}}\\
 &~= \frac{\sqrt{m} \, d^{\frac{1}{k}} \, N^{1-\frac{1}{k}} \left(d^{\frac{1}{N}} - 1 \right)^{1-\frac{1}{k}}}{\sqrt{d} } \cdot (\|\uu\|_{\infty} + \|\vv\|_{\infty}) \notag \\
 &\hspace{6.8cm} \cdot \frac{1}{\sqrt{a_1^2 \E[u_j^2]  + a_2^2 \E[v_j^2] + 2a_1a_2 \E[u_jv_j]}}\\
 &~\leq \frac{\sqrt{m} \, d^{\frac{1}{k}} \, N^{1-\frac{1}{k}} \left(d^{\frac{1}{N}} \right)^{1-\frac{1}{k}}}{\sqrt{d} } \cdot  (\|\uu\|_{\infty} + \|\vv\|_{\infty}) \notag\\
 &\hspace{6.8cm}\cdot \frac{1}{\sqrt{a_1^2 \E[u_j^2]  + a_2^2 \E[v_j^2] + 2a_1a_2 \E[u_jv_j]}}\\
 &~= \frac{\sqrt{m} \,  N^{\left(1-\frac{1}{k}\right)} \, d^{\left(\frac{N+k-1}{kN}\right)}}{\sqrt{d} } \cdot  (\|\uu\|_{\infty} + \|\vv\|_{\infty}) \notag\\
 &\hspace{6.8cm}\cdot \frac{1}{\sqrt{a_1^2 \E[u_j^2]  + a_2^2 \E[v_j^2] + 2a_1a_2 \E[u_jv_j]}}\\
 &~= \frac{\sqrt{m} \,  N^{\left(1-\frac{1}{k}\right)}}{d^{\left(\frac{1}{2}- \frac{N+k-1}{kN}\right)} } \cdot  (\|\uu\|_{\infty} + \|\vv\|_{\infty}) \cdot \frac{1}{\sqrt{a_1^2 \E[u_j^2]  + a_2^2 \E[v_j^2] + 2a_1a_2 \E[u_jv_j]}}\\
 &~= \frac{\sqrt{m} \,  N^{\left(1-\frac{1}{k}\right)}}{d^{\left(\frac{kN-2N-2k+2}{2kN}\right)} } \cdot  (\|\uu\|_{\infty} + \|\vv\|_{\infty}) \cdot \frac{1}{\sqrt{a_1^2 \E[u_j^2]  + a_2^2 \E[v_j^2] + 2a_1a_2 \E[u_jv_j]}}\\
 &~= \frac{\sqrt{m} \,  N^{\left(1-\frac{1}{k}\right)}}{d^{\left(\frac{kN-2N-2k+2}{2kN}\right)} } \, (\|\uu\|_{\infty} + \|\vv\|_{\infty}) \cdot \frac{1}{\sqrt{a_1^2 \E[u_j^2]  + a_2^2 \E[v_j^2] + 2a_1a_2 \E[u_jv_j]}}\\
 & ~ \rightarrow 0 \text{ as } d \rightarrow \infty \text{ for } k > \frac{2(N-1)}{(N-2)} \text{ and }  \sqrt{m} \,  N^{\left(1-\frac{1}{k}\right)} = o\left(d^{\left(\frac{kN-2N-2k+2}{2kN}\right)} \right). \numberthis \label{eq:eq281022_17}
\end{align*}
Equation~\eqref{eq:eq281022_17} hold true for $k > \frac{2(N-1)}{(N-2)}$ \text{ and }  $\sqrt{m} \,  N^{\left(1-\frac{1}{k}\right)} = o\left(d^{\left(\frac{kN-2N-2k+2}{2kN}\right)} \right)$, provided $ {\E[u_{j}^2]}$, $\E[v_{j}^2]$, $\E[u_jv_j]$, and,  $\|\uu\|_{\infty} + \|\vv\|_{\infty}$ is finite. For $k = 5$ and $d \rightarrow \infty$, from Equation~\eqref{eq:eq281022_17}, we have
\begin{align*}
     \left(\frac{d}{M}\right)^{\frac{1}{k}} \frac{M A}{\sigma_{d}} & \rightarrow 0 \text{ for } \sqrt{m} N^{\left(\frac{4}{5}\right)} = o\left(d^{\left(\frac{3N-8}{10N}\right)} \right). \numberthis \label{eq:eq281022_18}
\end{align*}
Thus, form Theorem~\ref{thm:clt_grpah_vec}, we have
\begin{align*}
    \begin{bmatrix}
\Tilde{u}_{l}  \\
\Tilde{v}_{l}
\end{bmatrix} = \sum_{j \in [d]} \Y_i= \overset{\mathcal{D}}{\to} \mathcal{N}\left(\mathbf{0}, \boldsymbol{\Sigma}_{d} \right). \numberthis \label{eq:eq281022_19}
\end{align*}
Equation~\eqref{eq:eq281022_19} completes a proof of the theorem.
\end{proof}

The theorem below states that the vector formed by concatenating the corresponding higher-order count sketches of the input vectors asymptotically follows a multivariate normal distribution. 

\begin{thm} \label{thm:hcs_muti_variate_srp}
Let $\uu$, $\vv \in \R^d$ and $\hcsb{\uu}$, $\hcsb{\vv} \in \R^{m_1 \times \cdots \times m_N}$ be their corresponding higher order count sketch. We define $\Tilde{\uu} := \Vecc(\hcs(\uu)) \in \R^{m}$ and $\Tilde{\vv} := \Vecc(\hcs(\vv)) \in \R^{m}$ s.t. $  m = \prod_{k=1}^{N} m_{k}$ and  $  l= \sum_{k=2}^{N} \left(l_{k} \prod_{t=1}^{k-1} m_{t} \right) + l_1 \in [m]$ for any  $l_{k} \in [m_{k}]$ and $k \in [N]$. If $~\forall j \in [d],$  ${\E[u_{j}^2]}$, $\E[v_{j}^2]$, $\E[u_jv_j]$ is finite and  $\sqrt{m} N^{\left(\frac{4}{5}\right)} = o\left(d^{\left(\frac{3N-8}{10N}\right)} \right)$,  then as $d \to \infty$,   we have

\begin{equation}
\begin{bmatrix}
\Tilde{\uu}  \\
\Tilde{\vv}
\end{bmatrix}
\overset{\mathcal{D}}{\to} \mathcal{N} \left( \mathbf{0},
\Sig_d 
 \right) \quad  \text{for} \quad l \in [m],
  \end{equation}
where  $\Sig_d=
\frac{1}{m}\begin{bmatrix}
  ||\uu||^2 & \cdots & 0 & \langle\uu,\vv \rangle  & \cdots & 0\\
  \vdots & \ddots & \vdots & \vdots & \ddots & \vdots \\
  0 & \cdots & ||\uu||^2  & 0 & \cdots &  \langle\uu,\vv \rangle \\
  \langle\uu,\vv \rangle  & \cdots & 0 & ||\vv||^2 & \cdots & 0 \\
  \vdots & \ddots & \vdots & \vdots & \ddots & \vdots \\
  0 & \cdots &  \langle\uu,\vv \rangle & 0 & \cdots & ||\vv||^2
\end{bmatrix}$ and $\mathbf{0} $ is a zero vector.
\end{thm}

\begin{proof}
Let $\uu$, $\vv \in \R^d$ and $\hcsb{\uu}$, $\hcsb{\vv} \in \R^{m_1 \times \cdots \times m_N}$ be their corresponding higher order count sketches. From the definition of higher order count sketch, we have:
\begin{align*} 
  \hcs(\uu)_{l_1, \ldots, l_N} &:= \sum_{h_1(i_1) = l_1, \ldots, h_{N}(i_N) = l_N} s_1(i_1) s_2(i_2) \cdots s_N(i_N) u_{j}, \numberthis \label{eq:eq281022_51} \\
    \hcs(\vv)_{l_1, \ldots, l_N} &:= \sum_{h_1(i_1) = l_1, \ldots, h_{N}(i_N) = l_N} s_1(i_1) s_2(i_2) \cdots s_N(i_N) v_{j} \numberthis \label{eq:eq281022_52} 
\end{align*}
where, for $k \in [N]$, $l_k \in [m_k]$, $i_k \in [d_k]$ and $  j= \sum_{k=2}^{N} \left(i_{k} \prod_{t=1}^{k-1} d_{t} \right) + i_1$. Let $ m := \prod_{k=1}^{N} m_{k}$, $\Tilde{\uu} := \Vecc(\hcsb{\uu}) \in \R^{m_1m_2\cdots m_N}$, $\Tilde{\vv} := \Vecc(\hcsb{\vv}) \in \R^{m_1m_2\cdots m_N}$ and $Z_{jl}$ be an indicator random variable (where, $  j=\sum_{k=2}^{N} \left(i_{k} \prod_{t=1}^{k-1} d_{t} \right)$ $  + i_1$ and $  l= \sum_{k=2}^{N} \left(l_{k} \prod_{t=1}^{k-1} m_{t} \right) + l_1$), which takes value $1$ if $h_1(i_1) = l_1, \ldots, h_{N}(i_N) = l_N$, otherwise takes value $0$. We can write the $l$-th component of the sketch vectors $\Tilde{\uu}$ and $\Tilde{\vv}$ as follows:
\begin{align*}
    \Tilde{u}_{l} &:= \sum_{i_1 \in [d_1], \ldots, i_{N} \in [d_N]} Z_{jl} s_1(i_1) s_2(i_2) \cdots s_N(i_N) u_{j}, \numberthis  \label{eq:eq281022_53} \\
     \Tilde{v}_{l} &:= \sum_{i_1 \in [d_1], \ldots, i_{N} \in [d_N]} Z_{jl} s_1(i_1) s_2(i_2) \cdots s_N(i_N) v_{j}. \numberthis  \label{eq:eq281022_54} 
\end{align*}

Let $\hat{u}_{jl} = Z_{jl} s_1(i_1) s_2(i_2) \cdots s_N(i_N) u_{j}$  and $\hat{v}_{jl} = Z_{jl} s_1(i_1) s_2(i_2) \cdots s_N(i_N) v_{j}$, where, $ j= \sum_{k=2}^{N} \left(i_{k} \prod_{t=1}^{k-1} d_{t} \right) + i_1 \in [d]$ and $l \in [m]$ . We define a sequence of  random vectors $\{\Y_{j}\}_{j=1}^{d}$ as follows:
\begin{align*}
    \Y_{j} &:= \begin{bmatrix}
    \hat{\uu}_{j}\\
    \hat{\vv}_{j}
    \end{bmatrix} = \begin{bmatrix}
    \hat{u}_{j1}\\
    \vdots\\
    \hat{u}_{jm}\\
    \hat{v}_{j1}\\
    \vdots\\
    \hat{v}_{jm}
    \end{bmatrix}.\\
  \implies \begin{bmatrix}
    \Tilde{\uu}\\
    \Tilde{\vv}
    \end{bmatrix} &= \sum_{j \in [d]} \Y_{j}. \numberthis \label{eq:eq281022_55}
\end{align*}
The expected value of $\begin{bmatrix}
    \Tilde{\uu} &
    \Tilde{\vv}
    \end{bmatrix}^T$ is
\begin{align*}
    \E \begin{bmatrix}
    \Tilde{\uu} \\ \Tilde{\vv}
    \end{bmatrix}  &= \sum_{j \in[d]} \E\left[ \Y_{j} \right]= \sum_{j \in[d]} \begin{bmatrix}
    \E[\hat{u}_{j1}]\\
    \vdots\\
    \E[\hat{u}_{jm}]\\
    \E[\hat{v}_{j1}]\\
    \vdots\\
    \E[\hat{v}_{jm}]
    \end{bmatrix}
    = \begin{bmatrix}
    0 \\
    \vdots\\
    0\\
    0\\
    \vdots\\
    0
    \end{bmatrix} = \mathbf{0}. \numberthis \label{eq:eq281022_56}\\
\end{align*}
Equation~\eqref{eq:eq281022_56} holds because $\E[s_1(i_1)\cdots s_N(i_N)] = \E[s_1(i_1)] \cdots  \E[s_N(i_N)] = 0$. 
We compute the covariance matrix  of  $\begin{bmatrix}
    \Tilde{\uu} &
    \Tilde{\vv}
    \end{bmatrix}^T$ as follows:
\begin{align*}
  \Cov\left(\begin{bmatrix}
    \Tilde{\uu} \\
    \Tilde{\vv}
    \end{bmatrix} \right) &= \Cov \left(\sum_{j \in [d]}\Y_{j} \right)= \sum_{j \in [d]} \Cov(\Y_{j}, \Y_{j}) + \sum_{j \neq j'} \Cov(\Y_{j}, \Y_{j'}). \numberthis \label{eq:eq281022_57}
\end{align*}
We compute each term of Equation~\eqref{eq:eq281022_57} one by one.
\begin{align*}
 & \Cov(\Y_{j}, \Y_j) = \begin{bmatrix}
    \Cov(\hat{\uu}_{j},\hat{\uu}_{j}) &  \Cov(\hat{\uu}_{j}, \hat{\vv}_{j})\\
    \Cov(\hat{\uu}_{j}, \hat{\vv}_{j}) & \Cov(\hat{\vv}_{j},\hat{\vv}_{j})
    \end{bmatrix} \\
    & = \begin{bmatrix}
     \Cov(\hat{u}_{j1}, \hat{u}_{j1}) & \cdots & \Cov(\hat{u}_{j1}, \hat{u}_{jm}) & \Cov(\hat{u}_{j1}, \hat{v}_{j1}) & \cdots & \Cov(\hat{u}_{j1}, \hat{v}_{jm})\\
     \vdots & \ddots & \vdots & \vdots & \ddots & \vdots\\
     \Cov(\hat{u}_{jm}, \hat{u}_{j1}) & \cdots & \Cov(\hat{u}_{jm}, \hat{u}_{jm})& \Cov(\hat{u}_{jm}, \hat{v}_{j1}) & \cdots & \Cov(\hat{u}_{jm}, \hat{v}_{jm})\\
     \Cov(\hat{v}_{j1}, \hat{u}_{j1}) & \cdots & \Cov(\hat{v}_{jm}, \hat{u}_{jm}) & \Cov(\hat{v}_{j1}, \hat{v}_{j1}) & \cdots & \Cov(\hat{v}_{j1}, \hat{v}_{jm})\\
     \vdots & \ddots & \vdots &\vdots & \ddots & \vdots\\
     \Cov(\hat{v}_{jm}, \hat{u}_{j1}) & \cdots & \Cov(\hat{v}_{jm}, \hat{u}_{jm}) & \Cov(\hat{v}_{jm}, \hat{v}_{j1}) & \cdots & \Cov(\hat{v}_{jm}, \hat{v}_{jm})
    \end{bmatrix}\\
    & = \begin{bmatrix}
     \frac{u_{j}^2}{m} & \cdots & 0 & \frac{u_{j}v_{j}}{m} & \cdots & 0\\
     \vdots & \ddots & \vdots & \vdots & \ddots & \vdots\\
     0 & \cdots & \frac{u_{j}^2}{m} & 0 & \cdots &  \frac{u_{j}v_{j}}{m}\\
     \frac{u_{j}v_{j}}{m} & \cdots & 0 & \frac{v_{j}^2}{m} & \cdots & 0\\
     \vdots & \ddots & \vdots &\vdots & \ddots & \vdots\\
     0 & \cdots & \frac{u_{j}v_{j}}{m} & 0 & \cdots & \frac{v_{j}^2}{m} 
    \end{bmatrix}.\numberthis \label{eq:eq281022_58}\\
\end{align*}
\begin{align*}
   &Cov(\Y_{j}, \Y_j') 
    = \begin{bmatrix}
    \Cov(\hat{\uu}_{j},\hat{\uu}_{j'}) &  \Cov(\hat{\uu}_{j}, \hat{\vv}_{j'})\\
    \Cov(\hat{\uu}_{j}, \hat{\vv}_{j'}) & \Cov(\hat{\vv}_{j},\hat{\vv}_{j'})
    \end{bmatrix} \\
    & = \begin{bmatrix}
     \Cov(\hat{u}_{j1}, \hat{u}_{j'1}) & \cdots & \Cov(\hat{u}_{j1}, \hat{u}_{j'm}) & \Cov(\hat{u}_{j1}, \hat{v}_{j'1}) & \cdots & \Cov(\hat{u}_{j1}, \hat{v}_{j'm})\\
     \vdots & \ddots & \vdots & \vdots & \ddots & \vdots\\
     \Cov(\hat{u}_{jm}, \hat{u}_{j'1}) & \cdots & \Cov(\hat{u}_{jm}, \hat{u}_{j'm})& \Cov(\hat{u}_{jm}, \hat{v}_{j'1}) & \cdots & \Cov(\hat{u}_{jm}, \hat{v}_{j'm})\\
     \Cov(\hat{v}_{j1}, \hat{u}_{j'1}) & \cdots & \Cov(\hat{v}_{jm}, \hat{u}_{j'm}) & \Cov(\hat{v}_{j1}, \hat{v}_{j'1}) & \cdots & \Cov(\hat{v}_{j1}, \hat{v}_{j'm})\\
     \vdots & \ddots & \vdots &\vdots & \ddots & \vdots\\
     \Cov(\hat{v}_{jm}, \hat{u}_{j'1}) & \cdots & \Cov(\hat{v}_{jm}, \hat{u}_{j'm}) & \Cov(\hat{v}_{jm}, \hat{v}_{j'1}) & \cdots & \Cov(\hat{v}_{jm}, \hat{v}_{j'm})
    \end{bmatrix}\\
    & = \begin{bmatrix}
    0 & \cdots & 0 & 0 & \cdots & 0\\
     \vdots & \ddots & \vdots & \vdots & \ddots & \vdots\\
     0 & \cdots & 0 & 0 & \cdots &  0\\
     0 & \cdots & 0 & 0 & \cdots & 0\\
     \vdots & \ddots & \vdots &\vdots & \ddots & \vdots\\
     0 & \cdots & 0 & 0 & \cdots & 0 
    \end{bmatrix}.\numberthis \label{eq:eq281022_59}
\end{align*}
Equations~\eqref{eq:eq281022_58} and \eqref{eq:eq281022_59} hold true due to the followings:
\begin{align*}
\E \left[\hat{u}_{jl}\right] &= \E \left[Z_{jl} s_1(i_1) s_2(i_2) \cdots s_N(i_N) u_{j} \right] =0.\\
\E \left[\hat{v}_{jl}\right] &= \E \left[Z_{jl} s_1(i_1) s_2(i_2) \cdots s_N(i_N) v_{j} \right] = 0.
\end{align*}
\begin{align*}
\Cov(\hat{u}_{jl}, \hat{u}_{jl}) &= \Cov(Z_{jl} s_1(i_1) s_2(i_2) \cdots s_N(i_N) u_{j}, Z_{jl} s_1(i_1) s_2(i_2) \cdots s_N(i_N) u_{j})\\
&= \E \left[ Z_{jl}^2 s_1(i_1)^2 s_2(i_2)^2 \cdots s_N(i_N)^2 u_{j}^2\right] - \E \left[ Z_{jl} s_1(i_1) s_2(i_2) \cdots s_N(i_N) u_{j}\right]^2\\
&= \E \left[ Z_{jl}\right] u_{j}^2 - 0\\
&= \E \left[ Z_{jl}\right] u_{j}^2\\
&= \frac{u_j^2}{m}.
\end{align*}
Similarly, 
\begin{align*}
    \Cov(\hat{v}_{jl}, \hat{v}_{jl}) &= \frac{v_{j}^2}{m}.
\end{align*}

\begin{align*}
    \Cov(\hat{u}_{jl}, \hat{v}_{jl}) & = \Cov(Z_{jl} s_1(i_1) s_2(i_2) \cdots s_N(i_N) u_{j}, Z_{jl} s_1(i_1) s_2(i_2) \cdots s_N(i_N) v_{j})\\
    & = \E[ Z_{jl} s_1(i_1) s_2(i_2) \cdots s_N(i_N) u_{j} Z_{jl} s_1(i_1) s_2(i_2) \cdots s_N(i_N) v_{j}] \notag\\
    &\qquad- \E[Z_{jl} s_1(i_1) s_2(i_2) \cdots s_N(i_N) u_{j}] \E[Z_{jl} s_1(i_1) s_2(i_2) \cdots s_N(i_N) v_{j}]\\
    & = \E[Z_{jl}^2 s_1(i_1)^2  s_2(i_2)^2 \cdots s_N(i_N)^2 u_{j} v_{j}] - 0\\
    &= \E[Z_{jl}]u_{j} v_{j}\\
    & = \frac{u_j v_j}{m}.
\end{align*}

\begin{align*}
  \Cov(\hat{u}_{jl},\hat{u}_{j'l}) &=  \E[Z_{jl} Z_{j'l}] \E[s_{1}(i_1)s_{1}(i_1')]\cdots\E[s_{N}(i_N)s_{N}(i_N')]u_j u_j' -0 = 0. \numberthis \label{eq:eq281022_10_1}\\
  \Cov(\hat{u}_{jl}, \hat{v}_{j'l})&=  \E[Z_{jl} Z_{j'l}] \E[s_{1}(i_1)s_{1}(i_1')]\cdots\E[s_{N}(i_N)s_{N}(i_N')]u_j v_j' - 0 =0. \numberthis \label{eq:eq281022_11_1}
 \end{align*}
 Similarly,  we can prove the following holds true.
 \begin{align*}
     \Cov(\hat{v}_{jl},\hat{u}_{j'l}) & = 0.\\
     \Cov(\hat{v}_{jl},\hat{u}_{jl'}) & = 0.\\
     \Cov(\hat{v}_{jl},\hat{u}_{j'l'}) &= 0.
 \end{align*}
Let 
 \begin{align*}
     \boldsymbol{\Sigma}_d &:= \Cov\left( \sum_{j \in [d]} \Y_j\right) = \Cov\left(\begin{bmatrix}
    \Tilde{\uu} \\
    \Tilde{\vv}
    \end{bmatrix} \right).
 \end{align*}
Form Equations~\eqref{eq:eq281022_57}, \eqref{eq:eq281022_58} and \eqref{eq:eq281022_59}, we have
\begin{align*}
    \boldsymbol{\Sigma}_d &= \begin{bmatrix}
     \frac{\|\uu\|_2^2}{m} & \cdots & 0 & \frac{\langle \uu, \vv \rangle}{m} & \cdots & 0\\
     \vdots & \ddots & \vdots & \vdots & \ddots & \vdots\\
     0 & \cdots & \frac{\|\uu\|_2^2}{m} & 0 & \cdots &  \frac{\langle \uu, \vv \rangle}{m}\\
     \frac{\langle \uu, \vv \rangle}{m} & \cdots & 0 & \frac{\|\vv\|_2^2}{m} & \cdots & 0\\
     \vdots & \ddots & \vdots &\vdots & \ddots & \vdots\\
     0 & \cdots & \frac{\langle \uu, \vv \rangle}{m} & 0 & \cdots & \frac{\|\vv\|_2^2}{m} 
    \end{bmatrix}. \numberthis \label{eq:eq281022_60}
\end{align*}
For any unit vector $\a \in \R^{2m}$, define a random variable 
$$S_d := \sum_{j\in[d]}\a^T \Y_j = \a^T\begin{bmatrix}
    \Tilde{\uu}\\
    \Tilde{\vv}
    \end{bmatrix}. $$
The expected value and variance of $S_d$ is
\begin{align}
    \E[S_d] := \a^T \E\begin{bmatrix}
    \Tilde{\uu} \\
    \Tilde{\vv}
    \end{bmatrix} = 0 \numberthis \label{eq:eq101024_2}
\end{align}
and 
\begin{align}
    \sigma_d^2 &=\Var(S_d):= \a^T \Cov\left(\begin{bmatrix}
    \Tilde{\uu} \\
    \Tilde{\vv}
    \end{bmatrix}  \right) \a \notag\\
    & =\frac{1}{m} \left(\|\uu\|^2 \sum_{i=1}^m a_i^2 + \|\vv\|^2 \sum_{i=m+1}^{2m} a_i^2 + 2 \langle \uu, \vv \rangle \sum_{i=1}^m a_i a_{m+i} \right).
\end{align}

To complete the proof, we need to prove that Equation~\eqref{eq:eq051024_0} of Theorem~\ref{thm:clt_grpah_vec} holds true for some value of $\alpha$. We recall it as follows:
\begin{align*}
    \left(\frac{d}{M}\right)^{\frac{1}{\alpha}} \frac{M A}{\sigma_d}  \rightarrow 0 \text{ as } d \rightarrow \infty
\end{align*}
where $k$ is an integer, $A$ is upper bound on the norm of $\Y_j$'s and is $\|\uu\|_{\infty} + \|\vv\|_{\infty}$ and $M$ denotes the maximum degree of the dependency graph generated by the random vectors $\mathbf{Y}_{j}$, $j\in[d]$ and equal to $\sum_{i=1}^{N} d_{i} - N$. For ease of analysis, we assume that $d_{1} =\ldots=d_{N} = d^{\frac{1}{N}}$ and compute the following
\begin{align*}
 &\left(\frac{d}{M}\right)^{\frac{1}{k}} \frac{M A}{\sigma_d} = \left( \frac{d}{\sum_{i=1}^{N} d_{i} - N}\right)^{\frac{1}{k}} \cdot \left(\sum_{i=1}^{N} d_{i} - N \right) \cdot (\|\uu\|_{\infty}+ \|\vv\|_{\infty}) \notag\\
 &\hspace{3.3cm} \cdot \frac{ \sqrt{m} }{\sqrt{\|\uu\|^2 \sum_{i=1}^m a_i^2 + \|\vv\|^2 \sum_{i=m+1}^{2m} a_i^2 + 2 \langle \uu, \vv \rangle \sum_{i=1}^m a_i a_{m+i}}}\\
 &~= \sqrt{m} \, d^{\frac{1}{k}}\, \left(\sum_{i=1}^{N} d_{i} - N \right)^{1-\frac{1}{k}} \cdot\left(\sum_{i=1}^{N} d_{i} - N \right) \cdot (\|\uu\|_{\infty}+ \|\vv\|_{\infty}) \notag\\
&\hspace{3.3cm} \cdot \frac{ 1 }{\sqrt{\|\uu\|^2 \sum_{i=1}^m a_i^2 + \|\vv\|^2 \sum_{i=m+1}^{2m} a_i^2 + 2 \langle \uu, \vv \rangle \sum_{i=1}^m a_i a_{m+i}}}\\
&= \sqrt{m} \, d^{\frac{1}{k}} \, \left(Nd^{\frac{1}{N}} - N \right)^{1-\frac{1}{k}} \cdot (\|\uu\|_{\infty} + \|\vv\|_{\infty}) \notag\\
 &\hspace{3cm} \cdot \frac{ 1 }{\sqrt{\|\uu\|^2 \sum_{i=1}^m a_i^2 + \|\vv\|^2 \sum_{i=m+1}^{2m} a_i^2 + 2 \langle \uu, \vv \rangle \sum_{i=1}^m a_i a_{m+i}}}\\
 &\hspace{8cm}~\left[ \because d_{1} =  \ldots=d_{N} = d^{\frac{1}{N}} \right]\\
 &~= \sqrt{m} \, d^{\frac{1}{k}} \, \left(N d^{\frac{1}{N}} - N \right)^{1-\frac{1}{k}} \cdot  (\|\uu\|_{\infty} + \|\vv\|_{\infty}) \notag\\
 & \cdot \frac{ 1 }{\sqrt{\left( \sum_{i=1}^m a_i^2\right) \sum_{j \in [d]} u_j^2  + \left(\sum_{i=m+1}^{2m} a_i^2\right) \sum_{j \in [d]} v_j^2  + 2 \left(\sum_{i=1}^m a_i a_{m+i} \right) \sum_{j\in [d]} u_jv_j }}\\
 &= \frac{\sqrt{m} \, d^{\frac{1}{k}} \, \left(N d^{\frac{1}{N}} - N \right)^{1-\frac{1}{k}}}{\sqrt{d}} \cdot  (\|\uu\|_{\infty} + \|\vv\|_{\infty}) \notag\\
 & \cdot \frac{ 1 }{\sqrt{\left( \sum_{i=1}^m a_i^2\right) \sum_{j \in [d]} \frac{u_j^2}{d}  + \left(\sum_{i=m+1}^{2m} a_i^2\right) \sum_{j \in [d]} \frac{v_j^2}{d}  + 2 \left(\sum_{i=1}^m a_i a_{m+i} \right) \sum_{j\in [d]} \frac{u_jv_j}{d} }}\\
  &= \frac{\sqrt{m} \, d^{\frac{1}{k}} \, \left(N d^{\frac{1}{N}} - N \right)^{1-\frac{1}{k}}}{\sqrt{d}} \cdot  (\|\uu\|_{\infty} + \|\vv\|_{\infty}) \notag\\
 &\hspace{2cm} \cdot \frac{ 1 }{\sqrt{\left( \sum_{i=1}^m a_i^2\right) \E[u_j^2]  + \left(\sum_{i=m+1}^{2m} a_i^2\right) \E[v_j^2]  + 2 \left(\sum_{i=1}^m a_i a_{m+i} \right) \E[u_jv_j]}}\\
 &= \frac{\sqrt{m} \, d^{\frac{1}{k}} \, N^{1-\frac{1}{k}} \, \left( d^{\frac{1}{N}} - 1\right)^{1-\frac{1}{k}}}{\sqrt{d}} \cdot  (\|\uu\|_{\infty} + \|\vv\|_{\infty}) \notag\\
 & \hspace{2cm}\cdot \frac{ 1 }{\sqrt{\left( \sum_{i=1}^m a_i^2\right) \E[u_j^2]  + \left(\sum_{i=m+1}^{2m} a_i^2\right) \E[v_j^2]  + 2 \left(\sum_{i=1}^m a_i a_{m+i} \right) \E[u_jv_j]}}\\
 &\leq \frac{\sqrt{m} \, d^{\frac{1}{k}} \, N^{1-\frac{1}{k}} \, \left( d^{\frac{1}{N}}\right)^{1-\frac{1}{k}}}{\sqrt{d}} \cdot  (\|\uu\|_{\infty} + \|\vv\|_{\infty}) \notag\\
 & \hspace{2cm}\cdot \frac{ 1 }{\sqrt{\left( \sum_{i=1}^m a_i^2\right) \E[u_j^2]  + \left(\sum_{i=m+1}^{2m} a_i^2\right) \E[v_j^2]  + 2 \left(\sum_{i=1}^m a_i a_{m+i} \right) \E[u_jv_j]}}\\
 &= \frac{\sqrt{m} \,  N^{\left(1-\frac{1}{k}\right)} \, d^{\left(\frac{N+k-1}{kN}\right)}}{\sqrt{d} } \cdot  (\|\uu\|_{\infty} + \|\vv\|_{\infty}) \notag\\
 & \hspace{2cm}\cdot \frac{ 1 }{\sqrt{\left( \sum_{i=1}^m a_i^2\right) \E[u_j^2]  + \left(\sum_{i=m+1}^{2m} a_i^2\right) \E[v_j^2]  + 2 \left(\sum_{i=1}^m a_i a_{m+i} \right) \E[u_jv_j]}}\\
 &= \frac{\sqrt{m} \,  N^{\left(1-\frac{1}{k}\right)}}{d^{\left(\frac{1}{2}- \frac{N+k-1}{kN}\right)} } \cdot (\|\uu\|_{\infty} + \|\vv\|_{\infty}) \notag\\
 & \hspace{2cm}\cdot \frac{ 1 }{\sqrt{\left( \sum_{i=1}^m a_i^2\right) \E[u_j^2]  + \left(\sum_{i=m+1}^{2m} a_i^2\right) \E[v_j^2]  + 2 \left(\sum_{i=1}^m a_i a_{m+i} \right) \E[u_jv_j]}}\\
 &= \frac{\sqrt{m} \,  N^{\left(1-\frac{1}{k}\right)}}{d^{\left(\frac{kN-2N-2k+2}{2kN}\right)} } \cdot (\|\uu\|_{\infty} + \|\vv\|_{\infty}) \notag\\
  & \hspace{2cm}\cdot \frac{ 1 }{\sqrt{\left( \sum_{i=1}^m a_i^2\right) \E[u_j^2]  + \left(\sum_{i=m+1}^{2m} a_i^2\right) \E[v_j^2]  + 2 \left(\sum_{i=1}^m a_i a_{m+i} \right) \E[u_jv_j]}}\\
& ~ \rightarrow 0 \text{ as } d \rightarrow \infty \text{ for } k > \frac{2(N-1)}{(N-2)} \text{ and }  \sqrt{m} \,  N^{\left(1-\frac{1}{k}\right)} = o\left(d^{\left(\frac{kN-2N-2k+2}{2kN}\right)} \right). \numberthis \label{eq:eq181024_1} 
 \end{align*}
Equation~\eqref{eq:eq181024_1} hold true for $k > \frac{2(N-1)}{(N-2)}$ \text{ and }  $\sqrt{m} \,  N^{\left(1-\frac{1}{k}\right)} = o\left(d^{\left(\frac{kN-2N-2k+2}{2kN}\right)} \right)$, provided $ {\E[u_{j}^2]}$, $\E[v_{j}^2]$, $\E[u_jv_j]$, and,  $\|\uu\|_{\infty} + \|\vv\|_{\infty}$ is finite. For $k = 5$ and $d \rightarrow \infty$, from Equation~\eqref{eq:eq181024_1}, we have
\begin{align*}
     \left(\frac{d}{M}\right)^{\frac{1}{k}} \frac{M A}{\sigma_{d}} & \rightarrow 0 \text{ for } \sqrt{m} N^{\left(\frac{4}{5}\right)} = o\left(d^{\left(\frac{3N-8}{10N}\right)} \right). \numberthis \label{eq:eq181024_2}
\end{align*}
Thus, form Theorem~\ref{thm:clt_grpah_vec}, we have
\begin{align*}
    \begin{bmatrix}
\Tilde{\uu}  \\
\Tilde{\vv}
\end{bmatrix} = \sum_{j \in [d]} \Y_i \overset{\mathcal{D}}{\to} \mathcal{N}\left(\mathbf{0}, \boldsymbol{\Sigma}_{d} \right). \numberthis \label{eq:eq181024_3}
\end{align*}
Equation~\eqref{eq:eq181024_3} completes a proof of the theorem.

\end{proof}

{\color{black}
The following corollary states that the independence of the concatenation of the corresponding elements of $\hcsb{\uu}$ and $\hcsb{\vv}$ and is built on the results of Theorem \ref{thm:hcs_bi_variate} and Corollary~\ref{cor:hcs_srp_ind}.  We can prove it following a similar argument as that of Corollary~\ref{cor:indp_cs_srp}.
}
\begin{cor} \label{cor:hcs_srp_ind}
Let $\uu, \vv \in \R^{d}$ and $\hcsb{\uu}, \hcsb{\vv} \in \R^{m_1 \times \cdots \times m_N}$ be their corresponding higher count sketch vectors. Let $\Tilde{\uu}:=\Vecc(\hcsb{\uu}) \in \R^m$ and $\Tilde{\vv} = \Vecc(\hcsb{\vv}) \in \R^m$, where $m=\prod_{k=1}^N m_k$. If  $\forall ~ j \in [d]$,  ${\E[u_{j}^2]}$, $\E[v_{j}^2]$, $\E[u_jv_j]$ are finite and  $\sqrt{m} N^{\left(\frac{4}{5}\right)} = o\left(d^{\left(\frac{3N-8}{10N}\right)} \right)$, then as $d \rightarrow \infty$, the tuples $\begin{bmatrix} \Tilde{u}_1\\  \Tilde{v}_{1} \end{bmatrix}, \ldots, \begin{bmatrix} \Tilde{u}_m\\  \Tilde{v}_{m} \end{bmatrix}$ are independent.
\end{cor}

{\color{black}
The subsequent theorem states that the \HCSELSH~ is a valid LSH. A proof of the theorem follows similar proof steps to that of Theorem~\ref{thm:ubiased_cssrp} and requires the results of Theorem~\ref{thm:hcs_bi_variate} and Corollary~\ref{cor:hcs_srp_ind}. 
}
\begin{thm} \label{thm:hcssrp_lsh_property}
Let $\xi'(\uu)$, $\xi'(\vv)$ be $m$-dimensional binary vector of $\uu,\vv\in \R^d$ respectively, obtained using our proposal (Definition~\ref{def:hcs_consine_lsh}). Then for $d \rightarrow \infty$ and $\sqrt{m} N^{\left(\frac{4}{5}\right)} = o\left(d^{\left(\frac{3N-8}{10N}\right)} \right)$, 
    $\Pr[ \xi'(\uu)_{l} = \xi'(\vv)_{l} ]  = 1-\frac{\theta_{(\uu,\vv)}}{\pi}$ \text{ and } 
    $\Pr[\xi'(\uu)_{1} = \xi'(\vv)_{1}, \ldots, \xi'(\uu)_{m} = \xi'(\vv)_{m}] = \left(1-\frac{\theta_{(\uu,\vv)}}{\pi}\right)^m$
where 
$\theta_{(\uu,\vv)}=\cos^{-1}\left(\frac{\uu^T\vv}{\|\uu\| \|\vv\|}\right)$ denotes the angular similarity between vector $\uu$ and $\vv$.
\end{thm}

Suppose $S= \frac{\uu^T\vv}{\|\uu\| \|\vv\|}$ denotes the cosine similarity between $\uu$ and $\vv$. From Theorem~\ref{thm:hcssrp_lsh_property}, it is evident that the probability of a collision decreases monotonically with $S= \frac{\uu^T\vv}{\|\uu\| \|\vv\|}$. Hence,  due to Definition~\ref{def:LSH}, our proposal \HCSSRP ~defined in Definition~\ref{def:hcs_consine_lsh} is $(R_1,R_2, P_{1}^{m}, P_2^{m})$-sensitive for $R_1 = S$, $R_2 = cS$, $P_{1} = (1-\cos^{-1}(S)/\pi)$  and $P_{2} = (1-\cos^{-1}(cS)/\pi)$.

 \begin{remark}
   Compared to \SRP, the evaluation of the proposed \HCSSRP ~(Definition~\ref{def:hcs_consine_lsh}) hash function requires a lesser number of operations as well as less space. The time and space complexity of \SRP~(Definition~\ref{def:srp}) to generate an $m$-sized hash code of a $d$-dimensional input vector is $O(md)$, whereas, the time complexity of our proposal \HCSELSH~is $O(d)$, and its space complexity is $O(N\sqrt[N]{d})  \ll O(d)$ when $N=o(d)$.
 \end{remark}
 }
\section{Experiment}~\label{sec:experiments}
\paragraph*{\textbf{Machine Description:}} We performed all the experiments on a machine having the following configurations:  Intel(R) Core(TM) i7-8750H CPU  having a 2.20GHz processor and 16 GB RAM.

{\color{black}
\paragraph*{\textbf{Dataset Description:}} We evaluate our proposal on the tasks of retrieving the top-$k$ nearest neighbour using the following datasets:
\begin{itemize}
\item \textbf{Synthetic:} We randomly generate 50,000 data points, each representing a vector in a $10,000$-dimensional space.
  \item \textbf{Arcene Dataset~\cite{UCI,10.5555/2976040.2976109}:} This dataset contains $900$ real-valued data points in the $10,000$ dimension and is used for the cancer classification tasks.
  
\item \textbf{Gene expression cancer RNA-Seq~\cite{UCI}:} This dataset contains a random extraction of gene expressions of patients with different tumour types. The dataset contains $801$ data samples in $20,531$ dimensions.

\item \textbf{PEMS-SF~\cite{Dua:2019}:} This dataset contains the daily occupancy rates of various car lanes on the San Francisco Bay Freeway over a $15$-month period. This dataset consists of  $440$ real-valued data points in $1,38,672$ dimensions.

\item \textbf{Gist\footnote{http://corpus-texmex.irisa.fr/}:} The Gist dataset comprises $1$ million high-dimensional vectors, each representing a $960$-dimensional Gist feature of an image. For our experiments, we used the first $5,00,000$ data points.
\end{itemize}
We use the following performance metrics to compare the performance of the proposed methods with the baseline algorithms:
\begin{itemize}
    \item[--]  storage require to generate an $m$-length hash code, 
    \item[--] the running time, and,
    \item[--] indexing quality, measured by the trade-off between recall and query time.
\end{itemize}
}

\noindent\textbf{Space complexity analysis:} In order to compare the space usage of our proposals with the baseline method, we computed the number of bytes required to store their projection matrix/tensor for different values of the hash code size $m$. Figure~\ref{fig:space} depicts the comparison of the space used by our proposals and the baseline methods for various  values of hash code size $m$.\\
\begin{figure*}[h]
    \centering    \includegraphics[width=\textwidth]{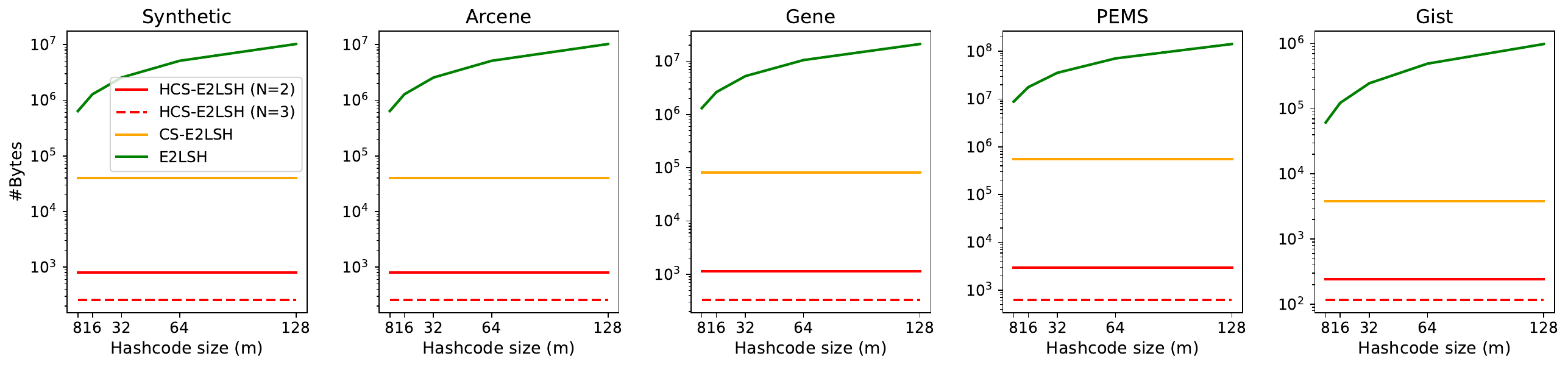}
    \vspace{-0.4cm}
    \caption{Comparison among \HCSELSH, \CSELSH, and \ELSH ~based on space used. \#Bytes denotes the number of bytes used, and $N$ is the order of the HCS tensor. The smaller the values of  \#Bytes indicates a more space-efficient algorithm. }
    \label{fig:space}
\end{figure*}
\noindent\textbf{Insight:} It is evident from Figure~\ref{fig:space} that the amount of space needed for the \HCSELSH~projection tensor decreases as the order (\textit{i.e.}, the value of $N$) increases.  The space usage of \ELSH~is the highest among all algorithms, and for a fixed size of the input vector, it increases with an increase in the hash code size $m$, whereas for our proposals (\CSELSH ~and \HCSELSH), it remains constant. The space requirement of \HCSELSH~is smallest among all the baselines. It is worth mentioning that these experimental findings on space usage align with theoretical space complexity bounds mentioned in Table~\ref{tab:tab1}.\\

\noindent \textbf{Running time analysis:}
Our objective is to compare the running time among the baselines. For this purpose, we generate an $m$-sized hash code for an input vector using \ELSH, \CSELSH, and \HCSELSH ~for different values of $m$. We note the time each technique takes to generate the hash code for all the datasets. We repeat this procedure several times and compute the average. We performed this procedure on a randomly sampled input vector for each dataset and presented the running time comparison among \ELSH, \CSELSH, and \HCSELSH~ \textit{w.r.t.} hash code size $m$ in Figure~\ref{fig:time}.
\begin{figure*}[h]
    \centering
\includegraphics[width=\textwidth]{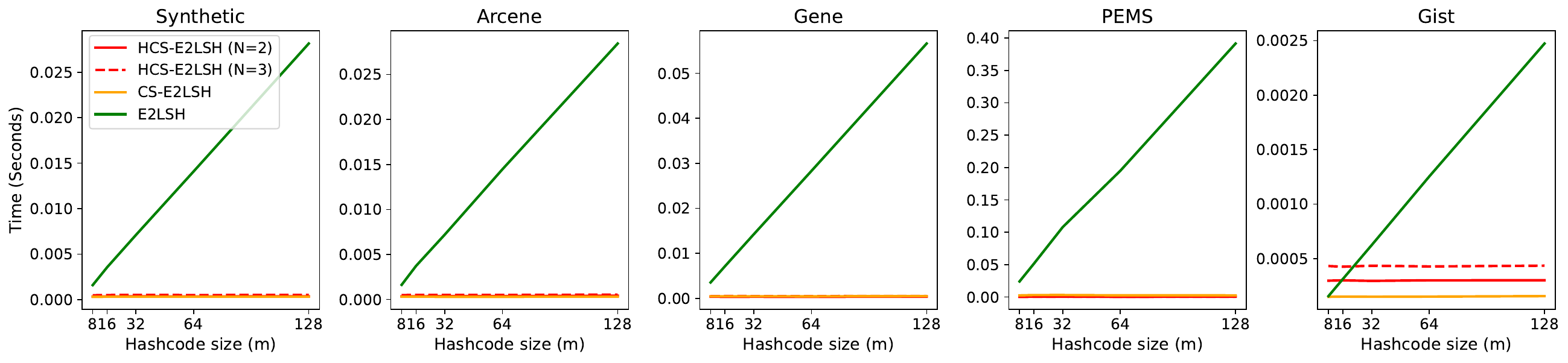}
\vspace{-0.4cm}
    \caption{Comparison on average running time for computing hashcode among \HCSELSH, \CSELSH~and \ELSH~for different values of hash code size $m$. $N$ is the order of the HCS tensor. The smaller value of time indicates a more time-efficient algorithm.}
    \label{fig:time}
\end{figure*}

\noindent\textbf{Insight:} From Figure~\ref{fig:time}, it is evident that the average running time of \ELSH~is the highest. The running time of \ELSH~increases with an increase in the hash code size $m$, whereas it remains constant for our proposals \HCSELSH~ and \CSELSH. These experimental findings align with our theoretical time complexity results in Table~\ref{tab:tab1}.  \textcolor{black}{For the Gist dataset, \HCSELSH~incurs higher runtime than \ELSH~and \CSELSH~at smaller hashcode sizes because for modest input and smaller hashcode sizes, the overhead of tensor reshaping and multi-mode projections outweighs the computational benefits of \HCSELSH. But for other datasets with higher input dimensions, these computational benefits are evident even at smaller hashcode sizes.} \\

We note that the projection step is common for locality-sensitive hash functions for cosine similarity and Euclidean distance. Our improved variants for \SRP ~and \ELSH~improve this projection step using the idea of replacing them with a sparse projection matrix (or tensor). Thus, the space and running time improvement hold for our improved variants of  \SRP ~ and \ELSH. \\


\noindent \textbf{Indexing quality:} 
We aim to show that our proposed improvements to \ELSH ~and \SRP ~exhibit similar behavior to their respective baselines in terms of hash collision properties and query retrieval performance. To do so, we randomly divide the datasets into $90\%$ and $10\%$ partitions, referring to the former as the training partition and the later as the query partition.  For the synthetic and Gist datasets, we use the last 100 points as the query set and the rest as a training set.
 For each point in the query partition, we record the top-$k$ similar points (under cosine similarity and Euclidean distance) in the original datasets. We refer to this set as $S$. We used standard ($m$, $L$) parameterized LSH~\cite{indyk1998approximate} algorithms to retrieve the top-$k$ elements. In the ($m$, $L$) parametrized LSH algorithm, we generate $L$ independent hash tables of $m$-sized hash codes and recover the top-$k$ elements for each query item using these tables (see~\cite{indyk1998approximate} for more detail). In case of \ELSH~(Definition~\ref{def:e2lsh}) and \SRP~(Definition~\ref{def:srp}), to generate an $m$-sized hash code we need to concatenate the $m$ hash functions; however, for our proposals \CSELSH~(Definition~\ref{def:CS_Eucl_LSH}), \HCSELSH~(Definition~\ref{def:HCS_Eucl_LSH}), \CSSRP~(Definition~\ref{def:cs_consine_lsh}), and \HCSSRP~(Definition~\ref{def:hcs_consine_lsh}), we can generate an $m$-sized hash code just using a single evaluation of hash functions.

In \HCSSRP ~and  \HCSELSH, we consider the input vector as $N$ order tensor, \textit{i.e.},  we reshape the $d$ dimensional input vector to $d_1 \times \cdots \times d_N$ dimensional array such that $d = \prod_{i=1}^N d_{i}$; $d_{i}$'s are positive integers. In case $d$ is not factorable as a product of $N$ integers, we pad the input vector with zero values to make it factorable. For \HCSSRP ~and  \HCSELSH, we experimented with $N$ equal to $2$ and $3$. In our experiments, we fix $k=50$, $m = 8$, and vary $L$ from 1 to 50 for both cosine similarity and Euclidean distance. For each value of $L$, we retrieve the top-50 nearest neighbors for each query point using the corresponding similarity measure (i.e., cosine or Euclidean) from the respective hash tables. The retrieved set is denoted as $S'$. We compute the \textit{recall} for each query point as: $\texttt{recall} = \frac{|S \cap S'|}{|S|}$ where $S$ is the ground-truth set of top-50 nearest neighbors. We then average the recall across all query points to obtain the recall for that setting. For each value of $L$, we also record the \textit{total query time}, defined as the cumulative time taken to process all query points. This entire process is repeated 20 times for each $L$, and we report the average recall and total query time across runs as the representative performance measure for each configuration. We start with $L = 1$ and increment $L$ by $1$ until the average recall reaches $0.99$ or $L = 50$, whichever comes first.

We run the above experiment on our datasets using Euclidean distance as the underlying similarity measure. We use \ELSH, \CSELSH, and \HCSELSH~as baselines, and present the recall vs. total query time and recall vs. $L$ trade-offs in Figures~\ref{fig:e2lsh_rec_time} and~\ref{fig:e2lsh_rec_l}, respectively. We repeat the same experiment for cosine similarity using \SRP, \CSSRP, and \HCSSRP, and present the corresponding results in Figures \ref{fig:srp_rec_time} and \ref{fig:srp_rec_l}

\textbf{Note:} In these experiments, total query time reflects the cost of retrieving the top-$k$ elements across $L$ hash tables. The recall vs. 
$L$ plots indicate how effectively each method maps query points and true neighbors to the same bucket, highlighting the proportion of actual top-$k$ neighbors retrieved. The recall vs. query time plots, in turn, reveal whether a method incurs more false negatives while maintaining similar true positive rates. Together, these trade-offs empirically evaluate the collision behavior and retrieval effectiveness of the proposed methods.
\begin{figure}
    \centering
    \includegraphics[width=\textwidth]{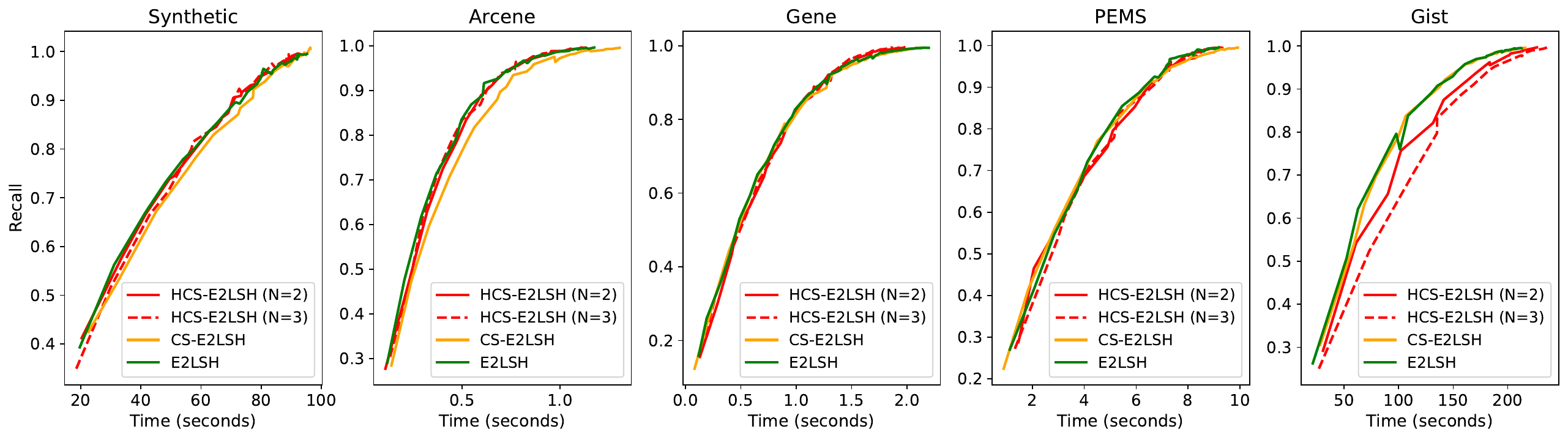}
    \caption{Trade-off between average recall and total query time over $20$ runs for \HCSELSH, \CSELSH, and \ELSH. $N$ denotes the order of the high-order count sketch tensor.}
    \label{fig:e2lsh_rec_time}
\end{figure}
\begin{figure}
    \centering
    \includegraphics[width=\textwidth]{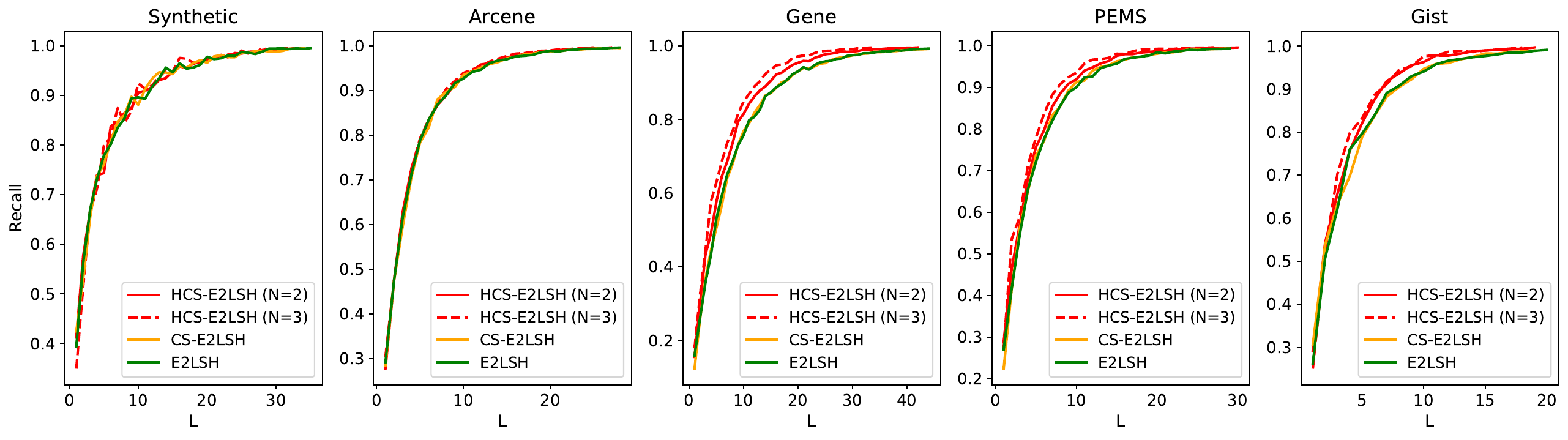}
    \caption{Trade-off between average recall and $L$ for \HCSELSH, \CSELSH, and \ELSH.}
    \label{fig:e2lsh_rec_l}
\end{figure}
\begin{figure}[H]
    \centering
    \includegraphics[width=\textwidth]{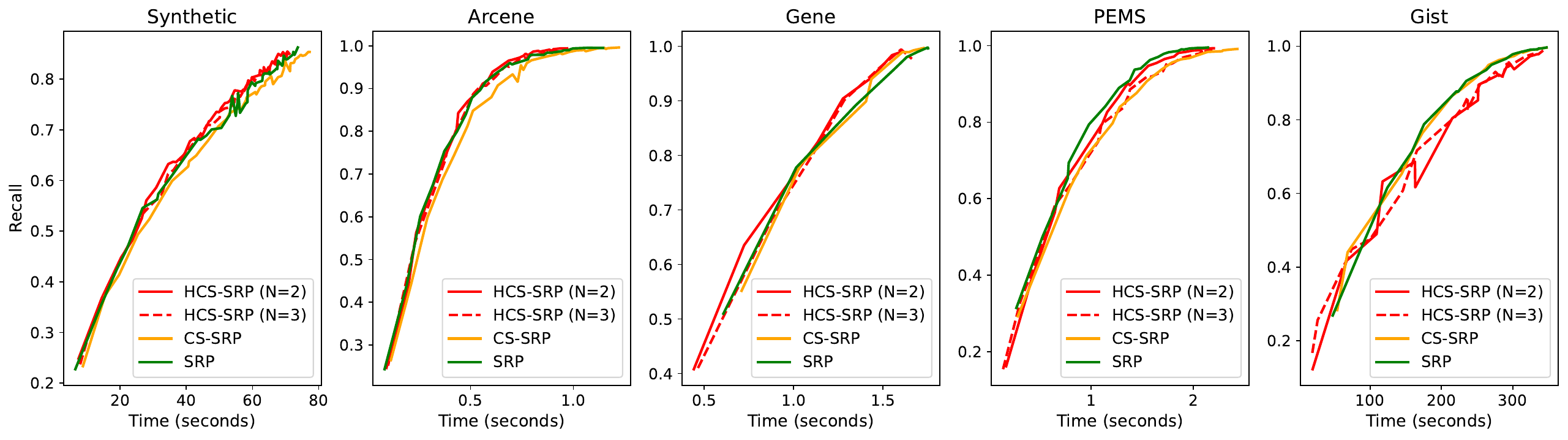}
    \caption{ Trade-off between average recall and total query time over $20$ runs for \HCSSRP, \CSSRP, and \SRP.}
    \label{fig:srp_rec_time}
\end{figure}

\begin{figure}[H]
    \centering
    \includegraphics[width=\textwidth]{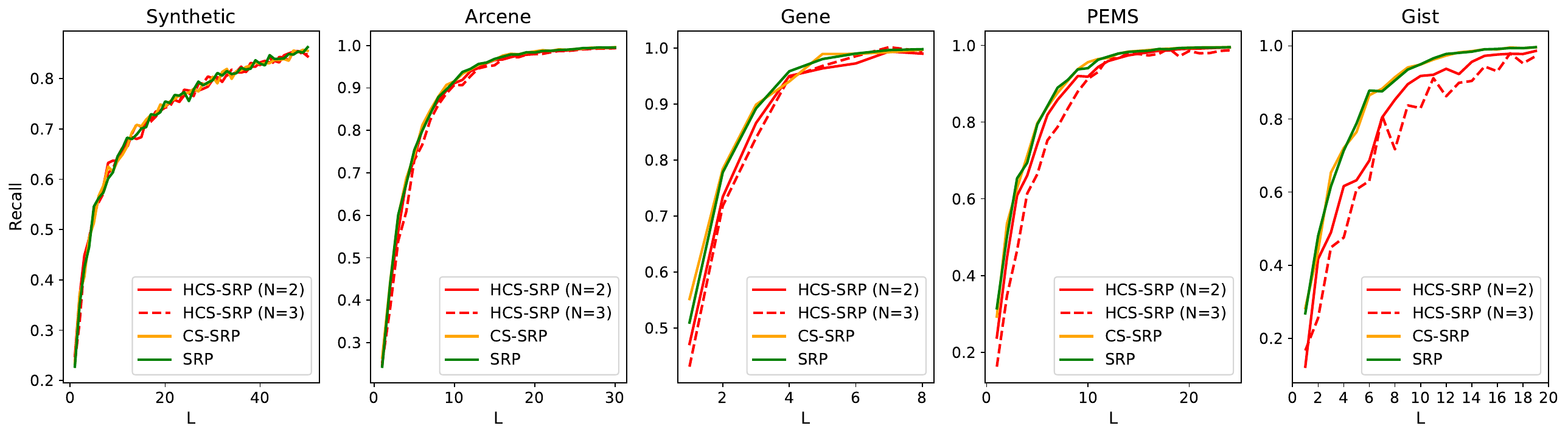}
    \caption{Trade-off between average recall and $L$ for \HCSSRP, \CSSRP, and \SRP.}
    \label{fig:srp_rec_l}
\end{figure}

\noindent \textbf{Insight:} From Figures~\ref{fig:e2lsh_rec_time} and~\ref{fig:e2lsh_rec_l}, we observe that across all datasets—except Gist—the average recall of \HCSELSH, \CSELSH, and \E2LSH ~is nearly identical, indicating that the proposed methods empirically preserve the same hash collision behavior as the baseline. For the Gist dataset, \HCSELSH ~shows slightly lower recall compared to \E2LSH ~and \CSELSH, although the results remain comparable. This deviation can be attributed to the Gist dataset’s moderate dimensionality, where the theoretical asymptotic behavior may not fully take effect (see Theorem~\ref{thm:HCS_Eucl_LSH_property}), causing \HCSELSH~to deviate slightly from the baseline performance.  Similar insights we can infer for our proposals \CSSRP ~and \HCSSRP, from the Figures~\ref{fig:srp_rec_time} and \ref{fig:srp_rec_l} in case of locality-sensitive hash functions for cosine similarity.

\section{Conclusion} \label{sec:conclusion}
In this work, we considered the locality-sensitive hash function for Euclidean distance and cosine similarity, namely, \ELSH ~and \SRP, respectively. The time and space complexity of these LSHs to generate an $m$ dimension hash code for a $d$-dimensional vector is $O(md)$, which becomes unaffordable for higher values of $m$ and $d$. We addressed this and suggested improved proposals for \ELSH~and \SRP. At the core, our idea is to use a sparse projection matrix from the count sketch~\cite{charikar2002similarity} and the higher-order count sketch~\cite{shi2019higher}. Our proposals for Euclidean distance and cosine similarity offer space and time-efficient algorithms as compared to the baselines  \ELSH ~and \SRP, respectively, while simultaneously achieving comparable accuracy for the top-$k$ nearest neighbour search task. Compared to \CSELSH ~and \CSSRP,  \HCSELSH ~and \HCSSRP ~are more space efficient; however, all the proposals enjoy the same time complexity, \textit{i.e.}, $O(d)$ for a $d$-dimensional input vector. We give a rigorous theoretical analysis of our proposals and complement them via empirical experiments on several datasets. Our proposals (\CSELSH, \HCSELSH, \CSSRP, and \HCSSRP) are simple, effective, easy to implement, and can be easily adopted in practice.

 \bibliographystyle{elsarticle-num-names} 
\bibliography{references}
\newpage

\end{document}